\documentclass[12pt]{article}
\usepackage{amsmath,amssymb,bm}
\usepackage{epsfig}
\usepackage[mathscr]{euscript}
\usepackage[affil-it]{authblk}
\usepackage{color}
\usepackage{amsthm}
\usepackage{titlesec}

\renewcommand\topfraction{.95}   
\renewcommand{\theequation}{\thesection.\arabic{equation}}
\renewcommand{\thefootnote}{\fnsymbol{footnote}}
\numberwithin{equation}{section}

\newlength{\extraspace}
\newlength{\extraspaces}
\newcommand{\be}{\begin{equation}
	\addtolength{\abovedisplayskip}{\extraspaces}
	\addtolength{\belowdisplayskip}{\extraspaces}
	\addtolength{\abovedisplayshortskip}{\extraspace}
	\addtolength{\belowdisplayshortskip}{\extraspace}}
\newcommand{\ee}{\end{equation}}

\newcommand{\ba}{\begin{eqnarray}
	\addtolength{\abovedisplayskip}{\extraspaces}
	\addtolength{\belowdisplayskip}{\extraspaces}
	\addtolength{\abovedisplayshortskip}{\extraspace}
	\addtolength{\belowdisplayshortskip}{\extraspace}}
\newcommand{\ea}{\end{eqnarray}}

\newcommand{\bas}{\begin{eqnarray*}
		\addtolength{\abovedisplayskip}{\extraspaces}
		\addtolength{\belowdisplayskip}{\extraspaces}
		\addtolength{\abovedisplayshortskip}{\extraspace}
		\addtolength{\belowdisplayshortskip}{\extraspace}}
	\newcommand{\eas}{\end{eqnarray*}}

\newcounter{subequation}[equation]
\makeatletter

\expandafter\let\expandafter
\reset@font\csname reset@font\endcsname

\def\subeqnarray{\arraycolsep1pt
	\def\@eqnnum\stepcounter##1{\stepcounter{subequation}%
		{\reset@font\rm(\theequation\alph{subequation})}}
	\jot5mm     \eqnarray}

\def\subarray{\arraycolsep1pt
	\def\@eqnnum\stepcounter##1{\stepcounter{subequation}%
		{\reset@font\rm(\alph{subequation})}}
	\jot5mm     \eqnarray}

\makeatother

\titleformat{\section}
{\normalfont\large\bfseries}{\thesection.}{0.5em}{}	

\titlespacing*{\section}
{0pt}{0pt}{4ex}

\titleformat{\subsection}
{\normalfont\bfseries}{\thesubsection}{0.5em}{}

\titlespacing*{\subsection}
{0pt}{6ex}{1.3ex}

\theoremstyle{definition}{
}

\theoremstyle{theorem}{
\newtheorem{theorem}{Theorem}[section]}

\theoremstyle{remark}{}
\theoremstyle{theorem}{
\newtheorem{proposition}{Proposition}[section]}

\theoremstyle{corollary}{
\newtheorem{corollary}[theorem]{Corollary}}

\theoremstyle{lemma}{
\newtheorem{lemma}[theorem]{Lemma}}

\begin{document}
	
\begin{titlepage}
	
\renewcommand{\thefootnote}{\fnsymbol{footnote}}
\makebox[1cm]{}
\vspace{1cm}
	
\begin{center}
	
\mbox{{\Large \bf  Bonus Properties of States of Low Energy}}\\[3mm]
\mbox{{\Large \bf }}
\vspace{2.8cm}
		
{\sc Rudrajit Banerjee}\footnote{email: {\tt rub18@pitt.edu}} 
{\sc and Max Niedermaier\footnote{email: {\tt mnie@pitt.edu}}}
\\[8mm]
{\small\sl Department of Physics and Astronomy}\\
{\small\sl University of Pittsburgh, 100 Allen Hall}\\
{\small\sl Pittsburgh, PA 15260, USA}
\vspace{18mm}
		
{\bf Abstract} \\[1mm]
\begin{quote}
States of Low Energy (SLE) are exact Hadamard states defined on 
arbitrary Friedmann-Lema\^{i}tre spacetimes. They are constructed from 
a fiducial state by minimizing the Hamiltonian's 
expectation value after averaging with a temporal window function. 
We show the SLE to be expressible solely in terms of the 
(state independent) commutator function. They also admit a
convergent series expansion in powers of the spatial momentum,
both for massive and for massless theories. In the massless case
the leading infrared behavior is found to be Minkowski-like for 
{\it all} scale factors. This provides a new cure for the    
infrared divergences in Friedmann-Lema\^{i}tre spacetimes 
with accelerated expansion. In consequence, massless SLE are viable 
candidates for pre-inflationary vacua and in a soluble model are 
shown to entail a qualitatively correct primordial power spectrum. 
\end{quote} 
\end{center}
\vfill
	
\setcounter{footnote}{0}
\end{titlepage}


\setlength\paperheight {305mm}   %
\renewcommand\topfraction{.99}   
\thispagestyle{empty}
\makebox[1cm]{}
\vspace{-23mm}
\begin{samepage}
	\tableofcontents
\end{samepage}
%
%
%
\newpage


\section{Introduction}\label{sec1}

For perturbatively defined quantum field theories on globally 
hyperbolic spacetimes there is a general consensus that the free 
state on which perturbation theory is based should be a Hadamard 
state. By-and-large the Hadamard property is necessary and sufficient 
for the existence of Wick powers of arbitrary order and hence for 
the perturbative series to be termwise well-defined at any order,
see \cite{Moretti,Hadamardnec} for recent accounts. On the other hand, 
Hadamard states are surprisingly difficult to construct concretely
\cite{HadamardNullbook,JunkerS,SJHadamard} even for background 
spacetimes with some degree of symmetry (other than maximal). 
The well-known adiabatic iteration \cite{ParkerTomsbook} has 
certain characteristics necessary for the Hadamard property 
built in, but is not convergent and cannot be fruitfully 
extended to small spatial momenta. The iteration can, however, serve 
as a conduit to establish the existence of states locally 
indistinguishable from Hadamard states \cite{JunkerS}.

An important class of backgrounds are generic Friedmann-Lema\^{i}tre 
cosmologies, where a construction of exact Hadamard states has become 
available only relatively recently \cite{Olbermann}. These 
States of Low Energy (SLE) arise by minimizing the Hamiltonian's 
expectation value {\it after} 
averaging with a temporal window function $f$. The temporal averaging is 
crucial and avoids the pathologies \cite{FRWHamdiag0} 
of the  earlier instantaneous diagonalization procedure. The 
construction of a SLE takes some fiducial solution $S$ of the 
homogeneous wave equation a starting point, considers 
arbitrary Bogoliubov transformations thereof, and then minimizes 
the temporal average of the energy   
with respect to them. Olbermann's theorem \cite{Olbermann} states 
that (for a massive free quantum field theory an a Friedmann-Lema\^{i}tre 
background) 
the minimizing solution $T[S]$ gives rise to an exact Hadamard state.
For given $S$ the minimizer $T[S]$ is unique up to a phase.    

Here we show that the SLE have a number of bonus properties 
that make them mathematically even more appealing and which 
also render them good candidates for vacuum-like states in a 
pre-inflationary period. Specifically, we show that for a given temporal averaging function $f$: 
\begin{itemize}
\item[(a)] The SLE two-point function $W[S]$ based on a fiducial solution $S$ is a Bogoliubov invariant, $W[a S+bS^\ast]=W[S]$, with $a,b\in \mathbb{C}$ , $|a|^2-|b|^2=1$. Hence $W[S]$ is independent of the choice of fiducial solution $S$.

\item[(b)] The minimization over Bogoliubov parameters relative to a given $S$ 
can be replaced by a minimization over initial data, without 
reference to any fiducial solution. The resulting expression 
for the SLE solution $T[\Delta]$ is fully determined by the 
(Bogoliubov invariant and state independent) commutator function 
$\Delta$, making manifest the uniqueness of the SLE. The minimization 
over initial data has a natural interpretation in the Schr\"{o}dinger 
picture.   
\item[(c)] The SLE solution admits a {\it convergent} series expansion  
in powers of the {(modulus of the)} spatial momentum, both for massive and for massless
theories. 
\item[(d)] In the massless case the leading infrared  
behavior is Minkowski-like for {\it all} cosmological scale 
factors. This provides a new cure for the long standing 
infrared divergences in Friedmann-Lema\^{i}tre backgrounds with 
accelerated expansion \cite{FRWinfra1}. 
\item[(e)] The modulus square of an SLE solution admits an 
asymptotic expansion in inverse odd powers of the  {(modulus of the)}  spatial momentum, which {is {\it independent of the window function $f$}. The coefficients of the expansion are {\it local}, recursively computable, and  generalize 
the heat kernel coefficients.} The asymptotics of the phase is 
governed by single integrals of the same coefficients. This short cuts 
the detour via the adiabatic expansion.
\end{itemize} 

Since linearized cosmological perturbations are described by 
massless free fields, the property (d) renders SLE a legitimate 
choice for a vacuum-like state in the early universe. Specifically, 
we argue that within the standard paradigm (classical 
Friedmann-Lema\^{i}tre backgrounds with selfinteracting scalar field) 
inflation must have been preceded by a period of non-accelerated 
expansion, for which the type with kinetic energy domination is 
mathematically preferred. The occurrence 
of the Bunch-Davies vacuum at the onset of inflation then requires extreme fine tuning.
In contrast, postulating 
a SLE for the primordial vacuum in the pre-inflationary phase 
is shown to automatically produce a qualitatively realistic 
power spectrum at the end of inflation.   

The paper is organized as follows. After introducing the SLE 
in the Heisenberg and the Schr\"{o}dinger pictures we establish 
properties (a) and (b) in Sections \ref{sec2.2} and \ref{sec2.3}, respectively. 
The existence of a convergent small momentum expansion is 
shown in Section \ref{sec3.1}, with the massless case detailed in 
Section \ref{sec3.2}. For large momentum, the existence of the WKB type 
expansion governed by generalized heat kernel coefficients 
is shown in Section \ref{sec4}. Finally, we study the viability of     
massless SLE as pre-inflationary vacua in Section \ref{sec5}.

\newpage 
\section{SLE in the Heisenberg and Schr\"{o}dinger pictures} \label{sec2}

A State of Low Energy (SLE) was originally defined in 
the Heisenberg picture by minimizing with respect to Bogoliubov 
parameters relating the corresponding solution of the wave 
equation to a reference solution $S$. As such, the SLE construction depends on 
the reference solution. Here we show that the SLE two-point function (which specifies the state completely) is independent of  $S$.
Next, the energy functional in the Schr\"{o}dinger picture  is naturally 
regarded as a function of the wave function's  initial data. 
By minimizing over initial data an alternative explicit expression 
for the SLE is obtained, which depends only on the 
(Bogoliubov invariant and state independent) commutator function.

\subsection{Homogeneous pure quasifree states in Heisenberg and 
Schr\"{o}dinger pictures}  \label{sec2.1}
Throughout, the background geometry will be a $1\!+\!d$ dimensional, 
spatially flat Friedmann-Lema\^{i}tre (FL) cosmology with line element
\begin{eqnarray}
	\label{FL} 
ds^2 = -\bar{N}(t)^2 dt^2 + a(t)^2 \delta_{ij} dx^i dx^j\,,
\end{eqnarray}
where $\bar{N}: \mathbb{R}_+ \rightarrow \mathbb{R}_+$ is the lapse function, 
$a: \mathbb{R}_+ \rightarrow \mathbb{R}_+$ is the cosmological scale factor, and 
$x^i, i=1,\ldots, d$ are adapted spatial coordinates. 
The form of the line element (\ref{FL}) is 
preserved under ${\rm Diff}[t_i,t_f] \times {\rm ISO}(d)$ 
transformations, where ${\rm Diff}[t_i,t_f]$ are endpoint 
preserving reparameterizations of some time interval $[t_i,t_f]$, 
$0 < t_i < t_f < \infty$, and the Euclidean group ${\rm ISO}(d)$ 
acts via global spatial diffeomorphisms connected to the identity. 
On this background we consider a scalar field 
$\chi:\mathbb{R}_+\times \mathbb{R}^d \rightarrow \mathbb{R}$, which is minimally coupled 
and initially selfinteracting with potential $U(\chi)$. 
Under the temporal reparameterizations 
$a(t)$ and $\chi(t,x)$ transform as scalars, while
$\bar{N}(t)$ and $\bar{n}(t) := \bar{N}(t)/a(t)^d$ are temporal densities, 
$\bar{n}'(t') = \bar{n}(t)/|\partial t'/\partial t|$, etc.. 
This is such that $\int_{t_i}^{t_f} \! dt \bar{N}(t) a(t)^p = 
\int_{t_i}^{t_f} \! dt \,\bar{n}(t) a(t)^{p+d}$ is  invariant 
for any $p$. Next, we expand the minimally coupled scalar field 
action on $[t_i, t_f]\times \mathbb{R}^d$ around a spatially homogeneous 
background scalar $\varphi(t)$ to quadratic order in the fluctuations
$\phi(t,x) := \chi(t,x) - \varphi(t)$. This gives a leading term 
$\bar{S}^{\varphi}$ (multiplied by a spatial volume term) whose field 
equation is one of the evolution equations for a FL cosmology. 
For $\varphi(t)$ solving it (with prescribed $a(t)$)
the term linear in the $\phi$ reduces to a boundary term
and may be omitted. The quadratic piece reads
\begin{eqnarray}
	\label{FLact}
S^{\phi} = \frac{1}{2}\int_{t_i}^{t_f} \! dt \int_{\Sigma} \! dx \,
\Big\{ \frac{1}{\bar{n}(t)} (\partial_t \phi)^2 - 
\bar{n}(t) a(t)^{2d} U''(\varphi) \phi^2 
- \bar{n}(t) a(t)^{2 d -2} \partial_i \phi 
\delta^{ij} \partial_j \phi\Big\}\,.
\end{eqnarray}
So far, $\varphi$ is for prescribed $a(t)$ a solution 
of $\- \partial_t(\bar{n}^{-1} \partial_t \varphi) + \epsilon_g \bar{n} a^{2d} 
U'(\varphi) =0$, but $a(t)$ itself is unconstrained. 
As far as the homogeneous background is concerned one 
could now augment the missing gravitational dynamics by the 
other FL field equations. This would turn $a(t), \varphi(t)$ into 
a solution of the Einstein equations and classical backreaction 
effects would be taken into account in the homogeneous sector. 
The standard ``Quantum Field Theory (QFT) on curved background'' 
viewpoint, on the other hand,  
treats the geometry as external, in which case 
(\ref{FLact}) adheres to the minimal coupling principle only 
if $U''(\varphi) = m_0^2$ is identified with a constant mass squared. 
In order to be able to switch back and forth between both 
settings we shall view $U''(\varphi) = m(t)^2$ formally as a time 
dependent mass and carry it along, specifying its origin
only when needed. In the field equations $\delta S^{\phi}/\delta \phi =0$ 
a spatial Fourier transform is natural, 
$\phi(t,x) = \int\! dp (2 \pi)^{-d} e^{ip x} \phi(t,p)$. Then 
$-\partial_i \delta ^{ij} \partial_j$ acts like $p^2 := 
p_i \delta^{ij} p_i$, which converts the field equation into 
an ordinary differential equation for each $p$ mode, viz
$[(\bar{n}^{-1} \partial_t)^2 + a(t)^{2d} m(t)^2 + a(t)^{2 d -2} p^2] 
\phi(t,p) =0$. 
\medskip

{\bf Homogeneous pure quasifree states.} 
On a FL background geometry there are, in general, infinitely 
many physically viable vacuum-like states for a QFT. A vacuum-like 
state is in particular a ``homogeneous pure quasifree'' state.
A ``state'' is normally defined algebraically as a positive 
linear functional over the Weyl algebra \cite{Moretti}. For the present 
purposes a ``state'' can be identified with the set of 
multi-point functions it gives rise to. Then ``quasifree'' 
means that all odd $n$-point functions 
in the state vanish while the even $n$-point functions can be 
expressed in terms of the two-point function $W(t,x;t',x')$ 
via Wick's theorem. Being a ``state'' entails certain properties 
of the two-point function that allow one to realize it
via the Gelfand-Naimark-Segal (GNS) construction in the form $(\Omega, u(t,x)
u(t',x') \Omega)$, for field operators $u(t,x)$ on vectors $\Omega$ 
 in the reconstructed state space.  ``Pure'' means 
that $\Omega$ cannot be written as a convex combination of other 
states. Finally, for a spatially flat FL background, ``homogeneous'' 
just means ``translation invariant'', i.e. $W(t,x;t',x')$ 
depends only on $x\!-\!x'$. 

The GNS reconstructed field operators $u(t,x)$ turn out to coincide with 
the Heisenberg field operators $\phi(t,x)$ (which are denoted 
by the same symbol as the classical field, as the latter will no longer 
occur.) The GNS vector $\Omega$ turns out to correspond to a Fock 
vacuum $|0_T\rangle$, annihilated by annihilation operators defined by 
a mode expansion of the Heisenberg field operator 
\begin{eqnarray}
	\label{phiexp} 
&& \phi(t,x) = \int\! \frac{dp}{(2\pi)^d} 
\big[ T_p(t) {\bf a}_T(p) e^{i px} + T_p(t)^* {\bf a}_T^*(p) e^{-i px} \big] \,,
\nonumber \\[1.5mm]
&& \big[ {\bf a}_T(p), {\bf a}_T^*(p')] = (2\pi)^d \delta(p-p')\,,
\quad {\bf a}_T(p) |0_T\rangle =0\,, 
\end{eqnarray}
where $T_p(t)$ is a complex solution of the above 
classical wave equation, and in the massless case $p=0$ needs to 
be excluded in the definition of $|0_T\rangle$. In order for 
the equal time commutation relations $[\phi(t,p), 
(\bar{n}^{-1} \partial_t \phi)(t,p')]  = i (2\pi)^d \delta(p + p')$ to 
hold, this solution must obey the Wronskian normalization 
condition $(\bar{n}^{-1} \partial_t T_p)(t) T_p(t)^* - 
(\bar{n}^{-1} \partial_t T_p)(t)^* T_p(t)= -i$. Then 
\begin{eqnarray}
	\label{FLcorr1} 
W(t,x;t',x') = \langle 0_T| \phi(t,x) \phi(t',x') |0_T\rangle
= \int\! \frac{dp}{(2 \pi)^d} 
\, T_p(t) \,T_p(t')^\ast \, e^{i p (x-x')}\,.
\end{eqnarray}
One sees that modulo phase choices a ``homogeneous pure quasifree'' 
state is characterized by a choice of Wronskian normalized solution 
$T_p(t)$ of the wave equation or, equivalently, by a choice of Fock vacuum 
$|0_T\rangle$ via (\ref{phiexp}).  
\medskip

{\bf Conventions.}
We briefly comment on our choice of conventions. In 
(\ref{phiexp}) often the ${\bf a}_T^*(p)$ is paired
with $T_p(t)$ not with $T_p(t)^*$. Then the sign 
in the Wronskian normalization condition has to 
be flipped correspondingly. 
More importantly, we seek to preserve temporal reparameterization 
invariance by carrying the lapse-like $\bar{n}(t) = \bar{N}(t)/a(t)^d$ 
along. Since in the wave equation $\bar{n}$ only occurs in the 
combination $\bar{n}^{-1} \partial_t$, it is convenient to introduce 
a new time function
\begin{eqnarray}
	\label{taudef} 
\tau := \int^t_{t_i} \! dt' \bar{n}(t') \,, \quad 
\partial_{\tau} = \bar{n}(t)^{-1} \partial_t\,,
\end{eqnarray}
for some $t_i$. Note that $\tau(t) = \tau'(t')$ is a scalar under 
reparameterizations $t' = \chi^0(t)$ of the coordinate time $t$, 
and that $d\tau = dt \bar{n}(t)$, $\bar{n}(t)^{-1} 
\delta(t,t') = \delta(\tau,\tau')$ are likewise invariant. 
Here $t' = \chi^0(t)$ with $\chi^0(t_i) = t_i < t_f 
= \chi^0(t_f)$ must be strictly increasing to qualify as a 
diffeomorphism. We write $a(\tau)$ for the cosmological scale factor 
viewed as a function of $\tau$ rather than $t$, and similarly for 
$m(\tau)$ as well as $T_p(\tau)$. The defining relations for $T_p(\tau)$ 
then read
\begin{eqnarray}
	\label{FLode1} 
&& \big[ \partial_{\tau}^2 + \omega_p(\tau)^2 ] T_p(\tau) =0\,,\quad 
\omega_p(\tau)^2 := a(\tau)^{2d} m(\tau)^2 + p^2 a(\tau)^{2 d -2}\,, 
\nonumber \\[1.5mm]
&& \partial_{\tau} T_p \,T_p^* - \partial_{\tau} T_p^* \,T_p = -i \,.
\end{eqnarray}
This setting has the advantage that the results in different 
time variables can be obtained by specialization:
\begin{eqnarray}
	\label{tgauges} 
\mbox{Cosmological time}&:& \bar{n}(t) = a(t)^{-d} \;\mbox{gauge, i.e.}
\;\bar{N}(t) =1\,,
\nonumber \\[1.5mm]
\mbox{Conformal time}&:& \bar{n}(t) = a(t)^{1-d} \;\mbox{gauge, i.e.}
\;\bar{N}(t) = a(t)\,, 
\nonumber \\[1.5mm]
\mbox{Proper time}&:& \bar{n}(t) =1\;\mbox{gauge, i.e.}
\;\bar{N}(t) = a(t)^d\,.
\end{eqnarray}
The first two gauges are standard; commonly one writes $\eta$ for $t$
in conformal time gauge.  The last gauge is the FL counterpart 
of the proper time gauge $\partial_t n(t,x) =0$ often adopted for the 
evolution of generic foliated spacetimes. 

Generally, $(\bar{n}^{-1} \partial_t)^2 = \bar{n}^{-2}(\partial_t^2 
- \bar{n}^{-1} \partial_t \bar{n} \partial_t)$ and the first order 
term can be removed by the redefinition $T_p(t) 
= \bar{n}(t)^{1/2} \chi_p(t)$. This gives 
\begin{eqnarray}
	\label{FLode2} 
&& \big[ \partial_t^2 + \bar{n}(t)^2 \omega_p(t)^2 + \bar{s}(t)\big] \chi_p(t)=0\,,
\nonumber \\[1.5mm]
&& \bar{s}(t) := \frac{1}{2} \frac{\partial_t^2 \bar{n}}{ \bar{n}} - 
\frac{3}{4} \Big( \frac{\partial_t \bar{n}}{\bar{n}} \Big)^2\,,
\nonumber \\[1.5mm]
&& \partial_t \chi_p \chi_p^* - (\partial_t \chi_p)^* \chi_p =-i\,.
\end{eqnarray} 
In conformal time, $\bar{n}(t) = a(t)^{1-d}$ the coefficient of 
$p^2$ is unity and after renaming $t$ into $\eta$ one has 
\begin{eqnarray}
	\label{FLode3} 
&& \Big[ \partial_{\eta}^2 + p^2 + \frac{m(\eta)^2}{a(\eta)^2} 
+ \bar{s}(\eta)\Big] \chi_p(\eta)=0\,,
\nonumber \\[1.5mm]
&& \bar{s}(\eta) := - \frac{d\!-\!1}{2} \frac{\partial_{\eta}^2 a}{a} - 
\frac{(d\!-\!3)(d\!-\!1)}{4} \Big( \frac{\partial_{\eta} a}{a} \Big)^2\,,
\nonumber \\[1.5mm]
&& \partial_{\eta} \chi_p \chi_p^* - (\partial_{\eta} \chi_p)^* \chi_p =-i\,.
\end{eqnarray}
We shall occasionally discretize the flat spatial sections of (\ref{FL}), 
which are isometric to $\mathbb{R}^d$, in order to regularize momentum integrals.
A hypercubic lattice 
$\Lambda =  \{ x = a_s(n_1, \ldots, n_d),\;n_j =0, \ldots , L\!-\!1\}$
suffices, with dual lattice $\hat{\Lambda} = 
\{ p = \frac{2\pi}{a_sL} (n_1, \ldots ,n_d) \,,\;n_j =0, \ldots, L\!-\!1 \}$, 
where $a_s>0$ is the spatial lattice spacing and $L \in \mathbb{N}$ is large. 
A discretized Fourier transform $\hat{f}: \hat{\Lambda} \rightarrow \mathbb{C}$ 
is defined for real valued functions $f: \Lambda \rightarrow \mathbb{R}$ with 
periodic boundary conditions $f(x+ a_s L \hat{\imath}) = f(x)$, 
$i =1,\ldots,d$. The direct and inverse transforms read
\begin{equation}
	\label{lattconv1} 
\hat{f}(p) = a_s^d \sum_{x \in \Lambda} e^{-i p x} f(x)\,,
\quad 
f(x) = \frac{1}{(a_sL)^d} \sum_{ p \in \hat{\Lambda}} 
e^{ i p \cdot x} \hat{f}(p) \,.
\end{equation}
The continuum limit is taken by first sending $L \rightarrow \infty$, 
which converts $(a_sL)^{-d} \sum_{ p \in \hat{\Lambda} }$ into an 
integral $(2\pi)^{-d} \int\! d^dp$ over the Brillouin zone 
$p \in [-\pi/a_s, \pi/a_s]^d$, and then taking $a_s \rightarrow 0$. 
As usual, the lattice Laplacian $\Delta_s$ acts by 
multiplication in Fourier space 
\begin{equation} 
\label{lattconv2} 
- \Delta_s \,e^{i p \cdot x} = \hat{p}^2 e^{i p \cdot x}\,,
\quad \hat{p}^2 := \sum_{j=1}^d \hat{p}_j^2 = 
\frac{4}{a_s^2} \sum_{j =1}^d \sin^2 \Big( \frac{p_j a_s}{2} \Big)\,.
\end{equation} 
Unless confusing we shall set $a_s\!=\!1$ and omit the `hat' on 
the Fourier transformed functions.  
\medskip

{\bf Heisenberg picture.} 
Time evolution in the Heisenberg picture is generated by 
the canonical Hamiltonian derived from (\ref{FLact}) with 
the field operators (\ref{phiexp}) inserted. After Fourier 
decomposition this leads to 
\begin{eqnarray}
	\label{HHam1} 
&& \mathbb{H}(\tau) = \int\! \frac{dp}{(2\pi)^d} 
\,\mathbb{H}_p(\tau) \,, \makebox[1cm]{ } 
\omega_p(\tau)^2 := m(\tau)^2 a(\tau)^{2 d} + p^2 a(\tau)^{2 d -2}\,, 
\nonumber \\[1.5mm]
&& \mathbb{H}_p(\tau) = \frac{1}{2} |\pi(\tau,p)|^2 + \frac{1}{2} 
\omega_p(\tau)^2 |\phi(\tau, p)|^2 
\\[2mm] 
&& \quad = \frac{1}{2} \big(|\partial_{\tau} T_p|^2 + \omega_p(\tau)^2 |T_p|^2\big) 
\big( {\bf a}_T(-p) {\bf a}_T^*(-p) + {\bf a}_T^*(p) {\bf a}_T(p) \big)
\nonumber \\[1.5mm]
&& \quad + \frac{1}{2} \big( (\partial_{\tau} T_p)^2 + \omega_p(\tau)^2 
T_p^2 \big) {\bf a}_T(-p) {\bf a}_T(p) 
+ \frac{1}{2} \big( (\partial_{\tau} T_p^*)^2 + \omega_p(\tau)^2 
(T_p^*)^2 \big) {\bf a}_T^*(p) {\bf a}_T^*(-p)\,.
\nonumber
\end{eqnarray}
In particular 
\begin{eqnarray}
	\label{HHamevol} 
\partial_{\tau} \phi(\tau, p) &\! =\! & i[ \mathbb{H}(\tau), \phi(\tau, p)] = \pi(\tau,p)\,,
\nonumber \\[1.5mm]
\partial_{\tau} \pi(\tau, p) &\! =\! & i[ \mathbb{H}(\tau), \pi(\tau, p)] = - 
\omega_p(\tau)^2 \phi(\tau,p)\,,
\end{eqnarray}
are the Heisenberg picture evolution equations. For later use 
we prepare their solution in terms of the (real, anti-symmetric) 
commutator function $\Delta_p(\tau',\tau)$ defined by  
\begin{eqnarray}
	\label{Deltadef1} 
&& \big[\partial_{\tau}^2 + \omega_p(\tau)^2\big] \Delta_p(\tau,\tau_0) = 0 = 
\big[\partial_{\tau_0}^2 + \omega_p(\tau_0)^2 \big]\Delta_p(\tau,\tau_0)\,,
\nonumber \\[1.5mm]
&& \Delta_p(\tau, \tau_0) = - \Delta_p(\tau_0, \tau)\,, \quad 
\partial_{\tau} \Delta_p(\tau, \tau_0) \big|_{\tau = \tau_0} = 1\,.
\end{eqnarray}
The terminology of course refers to the relations 
\begin{eqnarray}
\label{Deltadef2}  
i [ \phi(\tau,p), \phi(\tau_0,p_0)] &\! =\! & (2\pi)^d \delta(p+p_0) 
\Delta_p(\tau, \tau_0)\,, 
\nonumber \\[1.5mm]  
\Delta_p(\tau,\tau_0) &:=& i \big(T_p(\tau) T_p(\tau_0)^* 
- T_p(\tau)^* T_p(\tau_0) \big) \,, 
\end{eqnarray}
so that $\partial_{\tau} \Delta_p(\tau,\tau_0)|_{\tau=\tau_0} =1$ codes the equal 
time commutation relations. Note that any other Wronskian normalized 
complex solution defines the same commutator function, see Lemma \ref{lm2.3}.
The solution of the evolution equations (\ref{HHamevol}) then 
reads 
\begin{eqnarray}
	\label{Deltadef3}  
\phi(\tau,p) &\! =\! & \Delta_p(\tau,\tau_0) \pi(\tau_0,p) - \partial_{\tau_0} 
\Delta_p(\tau,\tau_0) \phi(\tau_0,p) \,,
\nonumber \\[1.5mm] 
\pi(\tau,p) &\! =\! & \partial_{\tau} \Delta_p(\tau,\tau_0) \pi(\tau_0,p) - 
\partial_{\tau} \partial_{\tau_0} \Delta_p(\tau,\tau_0) \phi(\tau_0,p) \,.
\end{eqnarray} 

The central object later on will be the Hamilton operator 
(\ref{HHam1}) averaged with a smooth positive window 
function $f(\tau)^2$ of compact support in $(\tau_i,\tau_f)$. 
We write   
\begin{eqnarray}
	\label{HHam2} 
\int\!\! d\tau f(\tau)^2 \, \mathbb{H}_p(\tau) &\! =\! & 
{\cal E}_p[T] \big( {\bf a}_T(-p) {\bf a}_T^*(-p) + {\bf a}_T^*(p) {\bf a}_T(p) \big)
\nonumber \\[1.5mm]
&+& {\cal D}_p[T] \,{\bf a}_T(-p) {\bf a}_T(p) 
+ {\cal D}_p[T]^* \,{\bf a}_T^*(p) {\bf a}_T^*(-p)\,,
\end{eqnarray}
with 
\begin{eqnarray}
	\label{HHam3} 
{\cal E}_p[T] \!&\!:=\!& \!\frac{1}{2} \int\! d\tau\, f(\tau)^2 \, 
\Big\{ |\partial_{\tau} T_p|^2 + \omega_p(\tau)^2 |T_p|^2 \Big\} > |{\cal D}_p[T]|\,,
\nonumber \\[1.5mm]
{\cal D}_p[T] \!&\!:=\!&\! \frac{1}{2} \int\! d\tau f(\tau)^2 \,\Big\{ 
(\partial_{\tau} T_p)^2 + \omega_p(\tau)^2 T_p^2 \Big\} \,.
\end{eqnarray} 
The above formulation preserves temporal reparameterization 
invariance through the use of $\tau$ from (\ref{taudef}). 
As a consequence, the solutions of the 
wave equation (\ref{FLode1}) can be interpreted as functions of 
the coordinate time $t$ with a functional dependence on $\bar{n}$. We shall 
occasionally do so and then (by slight abuse of notation) 
keep the function symbols, writing $T_p(\tau) = T_p(t)$, etc.. 
When fixing a gauge 
as in (\ref{tgauges}) one will however normally absorb 
additional powers of $a(t)$ into a redefined averaging function 
and frequency. Specifically, 
\begin{equation}
	\label{sleconv0} 
{\cal E}_p[T] = \!\frac{1}{2} \int\! dt\, f(t)^2 \bar{n}(t)^{-1} \, 
\Big\{ |\partial_t T_p|^2 + (\bar{n}(t)\omega_p(t))^2 |T_p|^2 \Big\}\,, 
\end{equation}
motivates 
\begin{eqnarray}
	\label{sleconv1}
f^{\rm cosm}(t)^2 &:=& f(t)^2 a(t)^d\,, \makebox[1cm]{ } \;\,\omega_p^{\rm cosm}(t) := a(t)^{-d} 
\omega_p(t) \,, 
\nonumber \\[1.5mm]
f^{\rm conf}(t)^2 &:=& f(t)^2 a(t)^{d-1}\,, \quad \;\;\, 
\omega_p^{\rm conf}(t) := a(t)^{1-d} \omega_p(t) \,, 
\nonumber \\[1.5mm]
f^{\rm prop}(t)^2 &:=& f(t)^2\,, \makebox[2cm]{ } \omega_p^{\rm prop}(t) := \omega_p(t) \,.
\end{eqnarray}
In cosmological time gauge this matches the conventions in 
\cite{Olbermann}.

The functional $\mathcal{E}_p[T]$ can be related to a point-split 
subtracted version of the 00-component of the energy momentum tensor 
\cite{Olbermann,Hackbook} and as such can be interpreted
as the energy {\it density} of a given $p$ mode. The same interpretation 
arises when the spatial sections are discretized. 
In the conventions of \eqref{lattconv1}, the main change is that the 
commutation relations in \eqref{phiexp} are replaced by 
$[{\bf a}_T(p), {\bf a}_T^*(p')] = L^d \delta_{p,p'}$. This gives  
$\mathcal{E}_p[T]$ (without subtractions) the interpretation as 
the energy density of the Hamiltonian's temporal average. 
Indeed, from \eqref{HHam2} one has 
\begin{eqnarray}
\label{HamDens}
\langle 0_T| \int\! d\tau f(\tau)^2 \mathbb{H}_p(\tau) | 0_T\rangle 
= L^d {\cal E}_p[T]\,.
\end{eqnarray}
\medskip

{\bf Schr\"{o}dinger picture.} Recall that the Heisenberg picture 
and the Schr\"{o}dinger picture are related by a unitary 
transformation implemented by the propagation operator
$U(\tau,\tau_0)$. The Schr\"{o}dinger picture is designed such 
that expectation values are the same as in the Heisenberg picture 
but the dynamical evolution is attributed to the states. Whence 
\begin{equation}
	\label{Schroed2} 
|\psi;\tau\rangle_{\mathsf{s}} := U(\tau, \tau_0)^{-1} |\psi\rangle\,, 
\quad 
A_{\mathsf{s}}(\tau) := U(\tau,\tau_0)^{-1}  A(\tau) U(\tau, \tau_0)\,.
\end{equation} 
Here $A(\tau)$ carries both the dynamical and potentially an 
explicit time dependence while $A_{\mathsf{s}}(\tau)$ carries only the 
residual explicit time dependence. The states $|\psi\rangle$ are 
normalizable and time independent while the Schr\"{o}dinger 
picture states evolve according to 
\begin{equation}
	\label{Schroed3} 
i \partial_{\tau} |\psi;\tau\rangle_\mathsf{s} = \mathbb{H}_\mathsf{s}(\tau) |\psi;\tau\rangle_\mathsf{s}\,,
\quad 
\mathbb{H}_{\mathsf{s}}(\tau) := U(\tau,\tau_0)^{-1} \mathbb{H}(\tau) U(\tau,\tau_0)\,.
\end{equation}
This is such that $\langle\psi| A(\tau) | \psi \rangle = 
{}_\mathsf{s}\!\langle\psi;\tau| A_\mathsf{s}(\tau) |\psi;\tau \rangle_\mathsf{s}$. 
As the propagation operator's generator one can alternatively 
take $\mathbb{H}(\tau)$ or $\mathbb{H}_{\mathsf{s}}(\tau)$; in terms of the path ordered 
exponentials one formally has 
\begin{eqnarray}
	\label{Schroed1} 
U(\tau, \tau_0) = \exp_+ \Big\{i\! \int_{\tau_0}^{\tau} \! ds \, \mathbb{H}(s) \Big\}
= \exp_- \Big\{i \!\int_{\tau_0}^{\tau} \! ds \, \mathbb{H}_\mathsf{s}(s) \Big\}\,,
\end{eqnarray}
where $\exp_+$ orders the operators from left to right 
in decreasing order of the argument and vice versa for 
$\exp_-$. Similar relations exist for the inverse. Note that 
only the $\exp_+$ versions will satisfy the usual composition law. 
Results on convergence properties will not be needed.

For the basic operators of our scalar QFT 
the Schr\"{o}dinger picture operators can be identified with 
the initial values of the Heisenberg picture operators. We 
transition to a lattice description (in order for the Schr\"{o}dinger picture to be rigorously defined) with $a_s =1$ and write 
\begin{equation}
	\label{Schroed4} 
\phi_\mathsf{s}(p) = \phi(\tau_0,p) =: u(p) \,, \quad 
\pi_\mathsf{s}(p) = \pi(\tau_0, p) =: -i L^d \frac{\delta}{\delta u(-p)}\,,\quad 
p \in \hat{\Lambda} \,.
\end{equation}
For the Hamiltonian this gives 
\begin{eqnarray}
	\label{Schroed5} 
\mathbb{H}_\mathsf{s}(\tau) = \frac{1}{2L^d} \sum_{p \in \hat{\Lambda}} 
\Big\{\!- L^{2d} \frac{\delta^2}{\delta u(p) \delta u(-p)} + 
\omega_p(\tau)^2 u(p) u(-p)\Big\}\,.
\end{eqnarray}
The matrix elements of the time averaged Heisenberg picture 
Hamiltonian become the time averages of the Schr\"{o}dinger 
picture matrix elements 
\begin{eqnarray}
	\label{Schroed6} 
\langle \psi| \int\! d\tau f(\tau)^2 \,\mathbb{H}(\tau) | \psi \rangle 
&\! =\! & \int\! d\tau \,f(\tau)^2 
{}_\mathsf{s}\langle \psi;\tau| \mathbb{H}_\mathsf{s}(\tau) | \psi;\tau \rangle_\mathsf{s}
\nonumber \\[1.5mm]
&\! =\! & \int\! d\tau \,f(\tau)^2 
{}_\mathsf{s}\langle \psi;\tau| i \partial_{\tau} | \psi;\tau \rangle_\mathsf{s}\,.
\end{eqnarray}

We state without derivation the counterpart of the Fock 
vacuum $|0_T\rangle$ in the Schr\"{o}dinger picture, see 
\cite{Kernelcurved1,Kernelcurved3,Kernelcurved3b,Kernelcurved4} 
for related accounts.  

\begin{proposition} \label{pr2.1}  The Schr\"{o}dinger picture state  
$|\Omega_T;\tau\rangle_{\mathsf{s}} := U(\tau,\tau_0)^{-1}|0_T\rangle$ 
evaluates on a finite lattice $\Lambda$ to 
\begin{eqnarray}
	\label{Schroed7} 
\Omega_T[u] &\! =\! & {\cal N}(\tau) \exp\Big\{ \frac{i}{2L^d} \sum_{p\in \hat{\Lambda}} 
\Xi_p(\tau) u(p) u(-p) \Big\} \,.
\nonumber \\[1.5mm]
\Xi_p(\tau) &\! =\! & \frac{\partial_{\tau} T_p(\tau)^*}{T_p(\tau)^*} = 
\frac{i + \partial_{\tau} |T_p(\tau)|^2}{2|T_p(\tau)|^2}\,,\quad 
\end{eqnarray}
with ${\cal N}(\tau) = \Omega_T[0]$. Separating modulus and phase,
$\Omega_T[u] = |\Omega_T[u]| e^{i A_T[u]}$, one has 
\begin{eqnarray}
	\label{Schroed8} 
|\Omega_T[u]| &\! =\! & |\Omega_0(\tau)| \prod_{p\neq 0} |\Omega_p(\tau)|\,, \quad 
A_T[u] = A_0(\tau) + \sum_{p \neq 0} A_p(\tau)
\nonumber \\[1.5mm] 
|\Omega_0(\tau)|&\! =\! &  \frac{1}{(2\pi L^d)^{1/4}}\frac{1}{\sqrt{T_0(\tau)}}\exp\Big\{ -\frac{u_0^2}{4L^d |T_0(\tau)|^2} \Big\}\,,
\nonumber \\[1.5mm]
|\Omega_p(\tau)|&\! =\! &  \frac{1}{(\pi L^d)^{1/4}}\frac{1}{\sqrt{T_p(\tau)}}\exp\Big\{ -\frac{u_p^2}{4L^d |T_p(\tau)|^2} \Big\}\,, 
\nonumber \\[1.5mm]
A_0(\tau) &\! =\! & \frac{1}{2}\arg T_0(\tau)+\frac{1}{2L^d}\partial_\tau \ln |T_0(\tau)|u_0^2\,,
\nonumber \\[1.5mm] 
A_p(\tau) &\! =\! & \frac{1}{2}\arg T_p(\tau)+\frac{1}{2L^d}\partial_\tau \ln |T_p(\tau)|\,|u_p|^2\,,
\end{eqnarray}
with normalization
\begin{equation}
	\label{Schroed9} 
\int\! \prod_p du(p) |\Omega_T[u]|^2 := 
\int\! du_0 |\Omega_0(\tau)|^2 
\int\! \prod_{p_d>0} du(p) |\Omega_p(\tau)|^4 =1\,.
\end{equation}
\end{proposition} 

With this in place we can return to (\ref{Schroed6}) and evaluate 
\begin{eqnarray}
	\label{Schroed10} 
{}_\mathsf{s}\langle \Omega_T;\tau| i \partial_{\tau} | \Omega_T;\tau \rangle_\mathsf{s} = 
\int\! \prod_{p} du(p) \Big\{ \frac{i}{2} \partial_{\tau} |\Omega_T[u]|^2   
- \partial_{\tau} A_T[u] |\Omega_T[u]|^2 \Big\}\,.
\end{eqnarray}
The imaginary part vanishes because $\Omega_T[u]$ is $L^2$ normalized. 
 The real part essentially is a Gaussian with a $|u|^2$ insertion. 
We interpret $|\Omega[u]|$ as in (\ref{Schroed8}) and find 
\begin{eqnarray}
\label{Schroed11} 
{}_\mathsf{s}\langle \Omega_T;\tau| i \partial_{\tau} | \Omega_T;\tau \rangle_\mathsf{s}
&\! =\! & -\frac{1}{2}\sum_p \Big\{|T_p(\tau)|^2 \partial_{\tau}^2 \ln |T_p(\tau)|
+\partial_{\tau} \arg T_p(\tau)\Big\}\,.
\end{eqnarray}
Next we use 
\begin{equation}
	\label{Schroed12}
\partial_{\tau} \arg T_p(\tau) = \frac{1}{2i} \partial_{\tau} 
\ln \frac{T_p(\tau)}{T_p(\tau)^*}  = - \frac{1}{2 |T_p(\tau)|^2}\,,
\quad \partial_{\tau}^2 \xi_p + \omega_p(\tau)^2 \xi_p = 
\frac{1}{4 \xi_p^3}\,,
\end{equation}
with $\xi_p(\tau) := |T_p(\tau)|$. The differential equation for $\xi_p$ 
is the Ermakov-Pinney equation. Together 
\begin{eqnarray}
	\label{Schroed13} 
{}_\mathsf{s}\langle \Omega_T;\tau| i \partial_{\tau} | \Omega_T;\tau \rangle_\mathsf{s}  &\! =\! &  
\frac{1}{2}\sum_p \Big\{ (\partial_{\tau} \xi_p)^2 + \omega_p(\tau)^2 \xi_p^2 + 
\frac{1}{4 \xi_p^2} \Big\}
\nonumber \\[1.5mm] 
&\! =\! & \frac{1}{2}\sum_p \Big\{|\partial_\tau T_p(\tau)|^2
+\omega_p(\tau)^2 |T_p(\tau)|^2 \Big\}\,.
\end{eqnarray}
Upon temporal averaging the right hand side equals 
$\sum_p {\cal E}_p[T]$, with ${\cal E}_p[T]$ from (\ref{HHam3}). Hence 
\begin{equation}
	\label{Schroed14} 
\int\! d\tau \, f(\tau)^2\,
{}_\mathsf{s}\langle \Omega_T;\tau| i \partial_{\tau} | \Omega_T;\tau \rangle_\mathsf{s} = 
 \sum_p {\cal E}_p[T]\,.
\end{equation}
As expected, the right hand side equals the  $L^{-d}\sum_{p}$ summation over 
$p$-fibres of \eqref{HamDens} in the  Heisenberg picture.  
The Schr\"{o}dinger picture, however, lends itself to a 
different minimization procedure described in Section \ref{sec2.3}.

\subsection{SLE in Heisenberg picture and independence of fiducial states}\label{sec2.2}

So far $T_p$ has been an arbitrary solution of (\ref{FLode1}).   
We now regard ${\cal E}_p[T]$ from (\ref{HHam3}) as a functional 
of $T_p$ and aim at minimizing it for fixed $p$. This is 
a finite dimensional minimization problem because the 
solutions of (\ref{HHam3}) are in one-to-one correspondence 
to their Wronskian normalized complex initial data. We shall 
pursue this route towards minimization in Section \ref{sec2.3}.

{\bf SLE via fiducial solutions.} 
Alternatively, one can fix a fiducial solution $S_p(\tau)$ 
of (\ref{FLode1}) and write any solution in the form 
\begin{equation}
	\label{sle1}  
T_p(\tau)=\lambda_p S_p(\tau)+\mu_p S_p(\tau)^\ast\,,\quad 
|\lambda_p|^2 - |\mu_p|^2 =1\,.
\end{equation}
With $S_p$ and $p$ held fixed the minimization is then over the 
parameters $\lambda_p,\mu_p \in \mathbb{C}$.  Since $e^{- i {\rm Arg \mu_p}} T_p(\tau)$
is a solution of (\ref{FLode1}) if $T_p(\tau)$ is we may 
assume wlog that $\mu_p$ is real. Since $|\lambda_p| = 
\sqrt{1 + \mu_p^2}$, only $\mu_p$ and the phase of $\lambda_p$ 
are real parameters over which the minimum of ${\cal E}_p[T_p]$ is sought. 
Inserting (\ref{sle1})  with the simplified parameterization 
into (\ref{HHam3}) one has 
\begin{eqnarray}
	\label{sle2} 
{\cal E}_p[T] &\! =\! & (1 + 2 \mu_p^2) {\cal E}_p[S] + 
\mu_p \sqrt{1 + \mu_p^2} \big( e^{i \arg \lambda_p}
{\cal D}_p[S] + e^{- i \arg \lambda_p} {\cal D}_p[S]^* \big)\,,
\nonumber \\[1.5mm]
{\cal D}_p[T] &\! =\! & (1 + \mu_p^2)e^{2 i \arg \lambda_p} {\cal D}_p[S] + \mu_p^2 {\cal D}_p[S]^* 
+ 2 \mu_p \sqrt{1 + \mu_p^2} e^{i \arg \lambda_p} {\cal E}_p[S]\,.
\end{eqnarray}
Clearly, the minimizing phase is such that 
$e^{i \arg \lambda_p} e^{i \arg {\cal D}_p[S]} =-1$. The minimization in 
$\mu_p$ then is straightforward and results in \cite{Olbermann} 
\begin{align}
\label{sle3} 
\mu_p=\sqrt{\frac{c_1}{2\sqrt{c_1^2-|c_2|^2}}-\frac{1}{2}}\,,
\quad \lambda_p=-\,e^{-i\,{\rm Arg} \,c_2}
\sqrt{\frac{c_1}{2\sqrt{c_1^2-|c_2|^2}}+\frac{1}{2}}\,,
\end{align}
where whenever the fixed fiducial solution is clear from the 
context one sets 
\begin{eqnarray}
	\label{sle4}
c_1 &:=& {\cal E}_p[S]= \frac{1}{2} \int\! d\tau f(\tau)^2 
\big[ |\partial_{\tau} S_p|^2 + \omega_p^2 |S_p|^2 \big] > |c_2|\,,
\nonumber \\[1.5mm]
c_2 &:=& {\cal D}_p[S]= \frac{1}{2} \int\! d\tau f(\tau)^2 
\big[ (\partial_{\tau} S_p)^2 + \omega_p^2 S_p^2 \big]\,.
\end{eqnarray}
Since only a phase choice has been made in 
arriving at (\ref{sle4}) it is clear that the minimizing 
linear combination is unique up to a phase, {\it for a fixed 
fiducial solution $S$.} It is called the {\it State of Low Energy} 
(SLE) solution of (\ref{FLode1}) with fiducial solution $S$. 
We write 
\begin{equation}
	\label{sle5} 
{T_{S,p}(\tau) := \lambda_p[S] S_p[\tau] + \mu_p[S] S_p(\tau)^*\,,}
\end{equation}
with {$\lambda_p[S], \mu_p[S]$} the functionals from (\ref{sle3}), (\ref{sle4}).  
Olbermann's theorem \cite{Olbermann} states that the 
homogeneous pure quasifree state associated with $T_S(\tau)$ 
via (\ref{FLcorr1}) is an {\it exact} Hadamard state. This
is an important result which improves earlier ones based on 
the adiabatic expansion in several ways, as noted in the 
introduction. Its practical usefulness is somewhat hampered by the 
fact that one still needs to know an exact solution $S$ of the 
wave equation to begin with and that the resulting Hadamard 
state off-hand depends on the choice of $S$. The second caveat 
is addressed in {Theorem \ref{th2.1} below. In preparation, we note the following proposition, where we omit the subscript $p$ for simplicity.} 

\begin{proposition}\label{pr2.2} Consider the following functionals: 
${\cal I} : C[\tau_i,\tau_f] \rightarrow \mathbb{R}_+\cup \{0\}$, and ${\cal J}, {\cal K}: 
C[\tau_i,\tau_f] \rightarrow C[\tau_i,\tau_f]$
\begin{eqnarray}
	\label{sleinv1}
{\cal I}[S] &:=& {\cal E}[S]^2 - |{\cal D}[S]|^2\,,
\nonumber \\[1.5mm]
{\cal J}[S](\tau) &:=& 2 {\cal E}[S] |S(\tau)|^2 - {\cal D}[S]^* S(\tau)^2 
- {\cal D}[S] {S(\tau)^*}^2\,,
\nonumber \\[1.5mm]
{\cal K}[S](\tau) &:=& 2 {\cal E}[S] |\partial_{\tau}S(\tau)|^2 - 
{\cal D}[S]^* [\partial_{\tau} S(\tau)]^2 
- {\cal D}[S] [\partial_{\tau}S(\tau)^*]^2\,.
\end{eqnarray} 
For $a,b \in \mathbb{C}$ they obey
\begin{eqnarray}
	\label{sleinv2}
{\cal I}[a S + b S^*]  = (|a|^2 - |b|^2)^2 \,{\cal I}[S]\,,
\nonumber \\[1.5mm] 
{\cal J}[a S + b S^*](\tau)  = (|a|^2 - |b|^2)^2 \,{\cal J}[S](\tau)\,, 
\nonumber \\[1.5mm]
{\cal K}[a S + b S^*](\tau)  = (|a|^2 - |b|^2)^2 \,{\cal K}[S](\tau)\,. 
\end{eqnarray}
\end{proposition} 

This may be proven by lengthy direct computations; we shall present a 
more elegant derivation based on properties of the commutator function 
in Section \ref{sec2.3}.

\medskip
\begin{theorem} \label{th2.1} \makebox[4cm]{} 
\vspace{-3mm} 
\begin{itemize} 
\item[(a)] {The SLE two-point function based on a fiducial solution $S$ } 
\begin{eqnarray}
{	W[S](\tau,x;\tau',x'):=\int \frac{d^dp}{(2\pi)^d}e^{ip(x-x')}T_{S,p}(\tau)T_{S,p}(\tau')^\ast \,,}
\end{eqnarray}
 {is a Bogoliubov invariant, i.e. $W[a S+bS^\ast]=W[S]$, with $a,b\in \mathbb{C}$, $|a|^2-|b|^2=1$. Hence $W[S]$ is independent of the choice of the fiducial solution $S$.}
\item[(b)] The modulus of an SLE solution can be written as
a ratio of Bogoliubov invariants from Proposition \ref{pr2.2}.  
\begin{equation}
	\label{sleinv3} 
|T_{S,p}(\tau)|^2 = \frac{{\cal J}_p[S](\tau)}{2\sqrt{{\cal I}_p[S]}}\,.
\end{equation}  
This also implies (a).    
\end{itemize} 
\end{theorem} 

\begin{proof}\

\noindent
{For readability's sake, we omit the subscript $p$ in the following.}

(a) We first show that a minimum $T$ of ${\cal E}$ is 
a zero of ${\cal D}$. Assume to the contrary that $T$ minimizes ${\cal E}$
but ${\cal D}[T] \neq 0$. Consider $\mu T + \lambda T^*$, with 
$\mu >0, \lambda = e^{i \arg \lambda} \sqrt{1 + \mu^2}$ and compute 
${\cal E}[\mu T + \lambda T^*]$ as in (\ref{sle2})  
\begin{equation}
	\label{sleunique1} 
{\cal E}[\mu T + \lambda T^*] = (1 + 2 \mu^2) {\cal E}[T] + 
2 \mu \Re\big( \lambda {\cal D}[T]\big)\,.
\end{equation}
Then there exists a $\mu \neq 0$ such that 
${\cal E}[\mu T + \lambda T^*] < {\cal E}[T]$, contradicting the 
assumption that $T$ minimizes ${\cal E}$. Subject to the minimizing 
phase choice $e^{- i\arg \lambda} e^{i \arg {\cal D}[S]} =-1$ one can also 
see from (\ref{sle2}) that $(\partial {\cal E}[T]/\partial \mu)$ is proportional 
to ${\cal D}[T]$. 

Let now $T_{S_1}, T_{S_2}$ be two minimizers of ${\cal E}$ associated 
with fiducial solutions $S_1,S_2$. Then there exist 
some $a,\,b\in\mathbb{C}$ with $|a|^2-|b|^2=1$ such that 
$T_{S_2}=aT_{S_1}+bT_{S_1}^\ast$. Further, $e^{-i\,{\rm {a}rg}\,b}T_{S_2}$ 
is of the form used in (\ref{sleunique1}) so that 
\begin{equation}
	\label{sleunique2} 
{\cal E}[e^{-i\,{\rm {a}rg}\,b}T_{S_2}]={\cal E}[T_{S_2}] =(2b^2+1)\,{\cal E}[T_{S_1}]+2b\, 
\mathfrak{R}[a \,{\cal D}(T_{S_1})]\,.
\end{equation}
By the previous step, ${\cal D}(T_{S_1}) =0$ as 
$T_{S_1}$ is a minimizer of ${\cal E}$. Therefore (\ref{sleunique2}) 
reduces to ${\cal E}[T_{S_2}] =(2b^2+1)\,{\cal E}[T_{S_2}]$. Since 
${\cal E}[T_{S_2}]={\cal E}[T_{S_1}]$ we must have $b=0$. Hence
$e^{-i\,{\rm {a}rg}\,b}T_{S_2}=T_{S_1}$, {which implies (a).} 

(b) The expression (\ref{sleinv3}) follows by direct computation.  
Hence (\ref{sleinv2}) implies (a) via $|T_{S_1}(\tau)| 
= |T_{S_2}(\tau)|$, as any two fiducial solutions $S_1,S_2$ 
must be related by $S_2 = a S_1 + b S_2^*$, $|a|^2 - |b|^2 =1$. { This also implies (a) since a Wronskian normalized solution of \eqref{FLode1} is uniquely determined by its modulus, up to a time independent (but potentially $p$ dependent) phase.}
\end{proof}

\noindent
{\bf Remarks.} 

(i) Uniqueness up a phase of the SLE modes 
has been asserted in Theorem 3.1 of \cite{Olbermann} and 
justified (in the line preceding it) by noting that only 
a phase choice is being made in the process of obtaining the solution 
formulas (\ref{sle3}). In itself, however, this only yields  
uniqueness {\it relative} to a choice of fiducial solution,
as indicated in (\ref{sle5}). We are not aware of a presention 
of SLE \cite{Olbermann,Hackbook,SLE1,SLE2} alluding to results of 
the above type. Lemma 4.5 of \cite{Olbermann} shows the independence 
of a SLE solution from the order of the adiabatic vacuum used 
as a fiducial solution. This, however, only concerns the large 
momentum behavior, while Theorem \ref{th2.1} ascertains the independence 
(up to a phase) from {\it any} fiducial solution at {\it all} momenta.

(ii) Writing momentarily ${\cal E}_S(\mu,\arg \lambda)$ for the 
right hand side of ${\cal E}[T]$ in (\ref{sle2}) one can of course 
trade a Bogoliubov transformation in $S$ for one in 
the parameters. This gives ${\cal E}_{S_1}(\mu_1,\arg \lambda_1) = 
{\cal E}_{S_2}(\mu_2,\arg \lambda_2)$ for any two fiducial solutions. 
For this to imply the existence of a unique  
minimum the gradients of ${\cal E}_{S_1}$ and ${\cal E}_{S_2}$ must be 
related by a $2 \times 2$ matrix  which remains nonsingular 
on a zero of one (and then both) gradient(s). Further, the 
Hessian must be positive definite on a zero of the gradient. 
The above proof validates these properties, but they are 
not consequences merely of the fact that (\ref{sle3}) is 
unique up to a choice of phase.   

(iii) By rewriting (\ref{sle2}) in matrix form one finds the 
minimizing parameters (\ref{sle3}) to diagonalize the original 
$c_1 = {\cal E}[S], c_2 = {\cal D}[S]$ matrix
\begin{equation}
	\label{sle6}
\begin{pmatrix} {\cal E}[T_S] & {\cal D}[T_S] \\[1mm]
{\cal D}[T_S]^* & {\cal E}[T_S] \end{pmatrix} 
=
\begin{pmatrix} \lambda & \mu \\
\mu & \lambda^* \end{pmatrix} 
\begin{pmatrix} c_1 & c_2 \\
c_2^* & c_1 \end{pmatrix} 
\begin{pmatrix} \lambda^* & \mu \\
\mu & \lambda \end{pmatrix} = 
\begin{pmatrix} \sqrt{c_1^2 - |c_2|^2} & 0 \\
0 & \sqrt{c_1^2 - |c_2|^2} \end{pmatrix} \,.
\end{equation}
The off-diagonal entries confirm the ``Minimizer of ${\cal E}$ is a  
zero of ${\cal D}$'' assertion in part (a) of the proof of Theorem \ref{th2.1};
the diagonal entries display the value of the minimizing 
energy ${\cal E}[T_S]$. 
In fact, the relation (\ref{sle6}) could be taken as an 
alternative definition of the coefficients $\lambda,\mu$ with solution 
(\ref{sle3}).   
\medskip

{\bf Minimization in Fock space.} We temporarily return to the lattice 
formulation.  The minimization of ${\cal E}_p[T]$ already assumed that the 
time averaged Hamiltonian $\int\! d\tau f(\tau)^2 \mathbb{H}_p(\tau)$ 
is evaluated in the coordinated Fock vacuum $|0_T\rangle$, see 
(\ref{HamDens}). The operator (\ref{HHam2}) itself 
has well-defined expectation values on a dense subspace ${\cal F}_0$ of 
the Fock space on which it is also selfadjoint and positive 
semidefinite. Hence 
\begin{equation}
	\label{fock1}
\inf_{\psi \in {\cal F}_0} 
\frac{\langle \psi | 
\int\! d\tau \, f(\tau)^2 \, 
\mathbb{H}_p(\tau) |\psi \rangle}{\langle \psi| \psi \rangle} = 
E_p^{\rm inf} \,,
\end{equation}
is well defined with some $E_p^{\rm inf} \geq 0$. 
By the min-max theorem for (possibly unbounded) selfadjoint 
operators \cite{ReedSimon4}, the quantity $E_p^{\rm inf}$ 
also coincides with the infimum of the spectrum of 
$\int\! d\tau \, f(\tau)^2 \, \mathbb{H}_p(\tau)$. In order to 
determine the infimum of the spectrum one can try to  
diagonalize the operator. Using (\ref{HHam2}), (\ref{sle6}),  
one has 
\begin{eqnarray}
	\label{fock2}
\int\!\! d\tau f(\tau)^2 \, \mathbb{H}_p(\tau) &\! =\! & 
\big( {\bf a}_S(-p), {\bf a}_S^*(p) \big) 
\begin{pmatrix} {\cal E}_p[S] & {\cal D}_p[S] \\
{\cal D}_p[S]^* & {\cal E}_p[S] \end{pmatrix} 
\begin{pmatrix} {\bf a}^*_S(-p) \\
{\bf a}_S(p) \end{pmatrix} 
\nonumber \\[1.5mm]
&\! =\! & {\cal E}_p[T_S] \big( 
{\bf a}_{T_S}(-p) {\bf a}^*_{T_S}(-p) +  
{\bf a}^*_{T_S}(p) {\bf a}_{T_S}(p) \big) \,.
\end{eqnarray}
From (\ref{fock2}) it is clear that the infimum of the spectrum 
is a minimum and is assumed if $|\psi \rangle = |0_{T_S}\rangle$ is the 
Fock vacuum associated with the SLE solution. Hence 
\begin{equation}
	\label{fock4}  
E_p^{\rm inf} = {\cal E}_p[T_S] = \sqrt{{\cal E}_p[S]^2 - |{\cal D}_p[S]|^2} \,.
\end{equation}
Since the operator in (\ref{fock1}) can be written in an arbitrary 
Bogoliubov frame one would expect that the infimum 
is a Bogoliubov invariant. By Proposition \ref{pr2.2} this indeed the case.

{\bf Instantaneous limit.} In general, the Fock vacuum  
${\bf a}_T(p)|0_T\rangle=0$ is not an eigenstate of $\mathbb{H}_p(\tau)$. 
At any fixed time $\tau_0$ one has however: 
\begin{eqnarray}
	\label{Hinstant1}
&& \quad\;\;\; |\partial_{\tau} T_p(\tau_0)|^2 + \omega_p(\tau_0)^2 |T_p(\tau_0)|^2 
\stackrel{\displaystyle{!}}{=} {\rm min} \,,
\nonumber \\[1.5mm]
&& \mbox{iff} \quad 
T_p(\tau_0) = \frac{e^{i\nu_0}}{\sqrt{2 \omega_p(\tau_0)}}\,,
\quad 
(\partial_{\tau} T_p)(\tau_0) = -i e^{i \nu_0}  
\sqrt{\frac{\omega_p(\tau_0)}{2}}\,.
\nonumber \\[1.5mm]
&& \mbox{iff} \quad  
[\partial_{\tau} T_p(\tau_0)]^2 + \omega_p(\tau_0)^2 T_p(\tau_0)^2 =0 \,,
\end{eqnarray}
for some $\nu_0 \in [0,2\pi)$. Note that in 
Minkowski space the minimization reproduces 
$T_p(t) = e^{-i t \omega_p}/\sqrt{2 \omega_p}$, $\omega_p = \sqrt{p^2 + m_0^2}$. 
Generally, the value of 
the minimum in the first line is $\omega_p(\tau_0)$.
With the choice (\ref{Hinstant1}) of minimizing 
mode `functions' the Hamilton operator {\it at} $\tau_0$ 
simplifies to 
\begin{eqnarray}
	\label{Hinstant2} 
\mathbb{H}(\tau_0) = \frac{1}{2} \int\!
\frac{dp}{(2 \pi)^d} \omega_p(\tau_0)
\big( {\bf a}_{\tau_0}(p) {\bf a}_{\tau_0}^*(p) + 
{\bf a}_{\tau_0}^*(p) {\bf a}_{\tau_0}(p) \big)\,.
\end{eqnarray}
On a finite lattice this also turns the Fock 
vacuum ${\bf a}_{\tau_0}(p) |0_{\tau_0}\rangle =0$ into the ground state 
of $\mathbb{H}(\tau_0)$. This ``instantaneous diagonalization'' 
has originally been pursued in an attempt to 
introduce a particle concept at each instant. The 
``instantaneous Fock vacuum'' $|0_{\tau_0}\rangle$ does however 
{\it not} give rise to a physically viable state, as  
for $\tau \neq \tau_0$ the norm-squared of the normal-ordered
Hamiltonian, 
$\langle 0_{\tau_0}| :\!\mathbb{H}(\tau)\!:\, :\!\mathbb{H}(\tau)\!: |0_{\tau_0}\rangle$,
in general diverges \cite{FRWHamdiag0,Hadamardnec}.
The temporal averaging resolves this problem in a simple 
and satisfactory manner.

Consistency requires that in the instantaneous 
limit $f(\tau)^2 \rightarrow \delta(\tau\!-\!\tau_0)$ the SLE 
solution (\ref{sle5}) reduces to the one in (\ref{Hinstant1}). 
One can check that this indeed the case
\begin{eqnarray}
	\label{slelimit1} 
T_{S,p}(\tau_0) = \lambda_p[S] S_p(\tau_0) + \mu_p[S] S_p^*(\tau_0) 
&\longrightarrow & 
\frac{1}{\sqrt{2 \omega_p(\tau_0)}}\,,
\nonumber \\[1.5mm]
\partial_{\tau} T_{S,p}(\tau_0)= \lambda_p[S] (\partial_{\tau}S_p)(\tau_0) + 
\mu_p[S]  (\partial_{\tau}S_p)^*(\tau_0) & \longrightarrow & 
-i  \sqrt{\frac{\omega_p(\tau_0)}{2}}\,.
\end{eqnarray}


\subsection{SLE in Schr\"{o}dinger picture and minimization 
over initial data}\label{sec2.3}

As seen in (\ref{fock1}), (\ref{fock4}) a SLE can be 
obtained by a minimization over the state space in the Heisenberg 
picture. The relevant matrix element can be transcribed into 
the Schr\"{o}dinger picture via (\ref{Schroed6}). 
Since the state vectors now evolve, the natural minimization 
is over their initial vectors $|\psi;\tau_0\rangle_S$, which 
can be identified with the Heisenberg picture states. 
The minimization in the Schr\"{o}dinger picture therefore 
assumes the form 
\begin{equation}
	\label{ssle0} 
\inf_{|\psi;\tau_0\rangle_{\mathsf{s}} \in {\cal F}_0}
\int\! d\tau \, f(\tau)^2\,
{}_{\mathsf{s}}\langle \psi;\tau| i \partial_{\tau} | \psi;\tau \rangle_{\mathsf{s}} \,.
\end{equation}
The Fock vacua correspond to time dependent Gaussians 
(\ref{Schroed7}), (\ref{Schroed8}) satisfying 
the functional Schr\"{o}dinger equation. The identity 
(\ref{Schroed14}) shows that the functional ${\cal E}_p$ on the 
space of solutions of the wave equation to be minimized is 
the same as in the Heisenberg picture. However, the relevant 
parameters are now the {\it initial data}.

In order to reformulate the minimization problem as one 
with respect to the initial data we proceed as follows.  
The solution formula (\ref{Deltadef3}) can be applied to 
the mode functions themselves giving
\begin{equation}
	\label{ssle1} 
T_p(\tau) = \Delta_p(\tau,\tau_0) \partial_{\tau_0} T_p(\tau_0) 
- \partial_{\tau_0} \Delta_p(\tau, \tau_0) T_p(\tau_0)\,.
\end{equation}
Inserting (\ref{ssle1}) and its time derivative 
into the definitions of ${\cal E}_p$ and ${\cal D}_p$ gives
\begin{eqnarray}
	\label{ssle2} 
{\cal E}_p &\! =\! & J_p(\tau_0) |w_p|^2 + K_p(\tau_0) |z_p|^2 
- \partial_{\tau_0} J_p(\tau_0) \Re(w_pz_p)\,, 
\nonumber \\[1.5mm]
{\cal D}_p &\! =\! & J_p(\tau_0) w_p^2 + K_p(\tau_0) z_p^2 
- \partial_{\tau_0} J_p(\tau_0) w_pz_p \,,
\end{eqnarray}
with $z_p:= T_p(\tau_0)$, $w_p := \partial_{\tau_0} T_p(\tau_0)$, 
subject to $w_p z_p^* - w_p^* z_p = -i$. The coefficients 
\begin{eqnarray}
	\label{ssle3}
J_p(\tau_0) &\! =\! & \frac{1}{2} \int\! d\tau \, f(\tau)^2 
\big[ \big( \partial_{\tau} \Delta_p(\tau,\tau_0) \big)^2 + 
\omega_p(\tau)^2 \Delta_p(\tau,\tau_0)^2 \big] \,,
\nonumber \\[1.5mm]
K_p(\tau_0) &\! =\! & \frac{1}{2} \int\! d\tau \, f(\tau)^2 
\big[ \big( \partial_{\tau} \partial_{\tau_0} \Delta_p(\tau,\tau_0) \big)^2 + 
\omega_p(\tau)^2 \big(\partial_{\tau_0} \Delta_p(\tau,\tau_0)\big)^2 \big]\,, 
\end{eqnarray}
are manifestly positive and are invariant under Bogoliubov 
transformations because the commutator function is. They are 
also independent of the initial data because $\Delta_p(\tau, \tau_0)$ 
is uniquely characterized by (\ref{Deltadef1}). No reference to 
any fiducial solution is made, instead ${\cal E}_p, {\cal D}_p$ in 
(\ref{ssle2}) are functions of the constrained complex 
initial data $z_p,w_p$.

Neither the sign nor the modulus of of $\partial_{\tau_0} J_p(\tau_0)$ 
is immediate. For the subsequent analysis we anticipate the inequality 
\begin{equation}
	\label{ssle4} 
4 K_p(\tau_0) J_p(\tau_0) - (\partial_{\tau_0} J_p(\tau_0))^2 >0\,.
\end{equation}
Further we momentarily simplify the notation by writing 
$K,J,\dot{J}$ for $K_p(\tau_0), J_p(\tau_0)$, $\partial_{\tau_0} J_p(\tau_0)$,
respectively. In addition we omit the subscripts $p$ from 
$z_p,w_p, {\cal E}_p, {\cal D}_p$. Since $T_p(\tau)$ in (\ref{ssle1}) 
can be multiplied by a $\tau$-independent phase we may assume 
$z$ to be real and positive. The solution of the Wronskian condition 
then gives 
\begin{equation}
	\label{ssle5} 
w = w_R - \frac{i}{2 z}\,, \quad w_R, \,z >0\,.
\end{equation}
Inserting (\ref{ssle5}) into the above ${\cal E}$ one is lead to 
minimize
\begin{equation}
	\label{ssle6} 
{\cal E} = J\Big( w_R^2 + \frac{1}{4 z^2} \Big) + K z^2 - 
\dot{J} zw_R\,,
\end{equation}
which gives
\begin{equation}
	\label{ssle7}
(z^{\rm min})^2 = \frac{J}{\sqrt{ 4 KJ - \dot{J}^2}}\,,\quad 
w^{\rm min}_R = \frac{z^{\rm min}}{2} \frac{\dot{J}}{J}\,.
\end{equation}
On general grounds the minimizer should be a zero of ${\cal D}$. 
Since 
\begin{equation}
	\label{ssle8} 
\frac{w^{\rm min}}{z^{\rm min}} 
= \frac{\dot{J}}{2 J} - i \frac{\sqrt{ 4 KJ - \dot{J}^2}}{2 J},
\end{equation}
this is indeed the case. Reinserting (\ref{ssle7}) into ${\cal E}$ 
gives 
\begin{equation}
	\label{ssle9} 
{\cal E}^{\rm min} = \frac{1}{2} \sqrt{4 KJ - \dot{J}^2}\,.
\end{equation}
Since ${\cal E}$ in the original form (\ref{HHam3}) is manifestly 
non-negative this shows the selfconsistency of (\ref{ssle4}). 
The solution is unique up to a constant phase left undetermined 
by choosing $z>0$. Upon insertion of (\ref{ssle3}) in 
(\ref{ssle7}), (\ref{ssle8}) the minimizing initial data become 
functionals of $\Delta$, for which we write $z_p[\Delta](\tau_0) 
= z^{\rm min}$, $w_p[\Delta](\tau_0) = w^{\rm min}$. In summary
\pagebreak[3]
\begin{theorem} \label{sslethm} \makebox[4cm]{} 
\vspace{-3mm} 

\begin{itemize}
\item[(a)] 
A SLE can be characterized as a solution 
$|\psi;\tau\rangle_{\mathsf{s}}$ of the time dependent Schr\"{o}dinger 
equation (\ref{Schroed3}), (\ref{Schroed5}) with initial data 
$|\psi;\tau_0\rangle_{\mathsf{s}}$ that minimize (for fixed window function $f$) 
the quantity $\int\! d\tau f(\tau)^2 
{}_{\mathsf{s}}\langle \psi;\tau| i \partial_{\tau} | \psi;\tau\rangle_{\mathsf{s}}$.  
The minimizing wave function is a Gaussian $\Omega_T[u]$ 
of the form (\ref{Schroed7}) with $T = T^{\rm SLE}$, which is up to 
a  {time independent, potentially $p$ dependent,} phase uniquely determined by the commutator function. 
\item[(b)]
Specifically 
\begin{eqnarray}
	\label{ssle10}  
T^{\rm SLE}_p(\tau) &\! =\! & \Delta_p(\tau,\tau_0) w_p[\Delta](\tau_0) - 
\partial_{\tau_0} \Delta_p(\tau, \tau_0) z_p[\Delta](\tau_0)\,,
\nonumber \\[1.5mm]
z_p[\Delta](\tau_0) &\! =\! & \sqrt{\frac{J_p(\tau_0)}{ 2 {\cal E}_p^{\rm SLE}}} = 
T_p^{\rm SLE}(\tau_0)\,,
\nonumber \\[1.5mm]
w_p[\Delta](\tau_0) &\! =\! & \partial_{\tau_0} T_p^{\rm SLE}(\tau_0) - i 
\sqrt{ \frac{{\cal E}_p^{\rm SLE}}{ 2 J_p(\tau_0)}} = (\partial_{\tau} T_p^{\rm SLE})(\tau_0)\,,
\end{eqnarray}
{where  $J_p(\tau_0)$ is as in  \eqref{ssle3}, and
${\cal E}_p^{\rm SLE}$ is the minimal energy given by }
\begin{eqnarray}
	\label{ssle11}
\big({\cal E}_p^{\rm SLE}\big)^2 &\! =\! & \frac{1}{8} \int d\tau d \tau' 
f(\tau)^2 f(\tau')^2 \Big\{ 
\big(\partial_{\tau} \partial_{\tau'} \Delta_p(\tau,\tau') \big)^2 + 
2 \omega_p(\tau')^2 (\partial_{\tau} \Delta_p(\tau',\tau) \big)^2 
\nonumber \\[1.5mm]
&+& \omega_p(\tau)^2 \omega_p(\tau')^2 \Delta_p(\tau,\tau')^2 \Big\}\,.
\end{eqnarray}
For the modulus and the phase this gives  
\begin{equation}
	\label{ssle12}
\big| T^{\rm SLE}_p(\tau)\big|^2 = \frac{J_p(\tau)}{2 {\cal E}_p^{\rm SLE}}\,, 
\quad 
\tan\! \big(\! \arg T_p^{\rm SLE}(\tau) \big) = 
-\frac{{\cal E}_p^{\rm SLE} \Delta_p(\tau,\tau_0)}{J_p(\tau,\tau_0)}\,,
\end{equation}
with
\begin{equation}
	\label{ssle11a}
J_p(\tau,\tau_0) := \frac{1}{2} \int\! d\tau_1 \, f(\tau_1)^2 
\big[\partial_{\tau_1} \Delta_p(\tau_1,\tau) \partial_{\tau_1} \Delta_p(\tau_1,\tau_0) + 
\omega_p(\tau_1)^2 \Delta_p(\tau_1,\tau) \Delta_p(\tau_1,\tau_0)\big] \,.
\end{equation}
{We note that $J_p(\tau_0)$ coincides with $J_p(\tau_0,\tau_0)$.}
\end{itemize}
\end{theorem}

\begin{proof}\

(a) This follows from (\ref{Schroed6}), (\ref{fock1}), 
(\ref{fock4}) and (\ref{Schroed14}).

(b) Eq.~(\ref{ssle10}) is the explicit form of (\ref{ssle1}) with 
minimizing parameters (\ref{ssle7}), (\ref{ssle9}). In the 
explicit expressions (\ref{ssle12}) with (\ref{ssle11}) and (\ref{ssle11a}) 
a reduction of order occurs: where naively terms fourth or third order 
in $\Delta$ and its derivatives appear, repeated use of 
\begin{eqnarray}
	\label{deltaID} 
\partial_{\tau_0} \Delta_p(\tau,\tau_0) 
\Delta_p(\tau',\tau_0) - \Delta_p(\tau,\tau_0) 
\partial_{\tau_0} \Delta_p(\tau',\tau_0) = \Delta_p(\tau,\tau') \,,
\end{eqnarray}
(as well as its $\partial_{\tau}$, $\partial_{\tau'}$ and $\partial_{\tau} \partial_{\tau'}$ 
derivatives) leads to results merely quadratic in $\Delta$ and 
its derivatives. In detail, by inserting the definitions into 
$({\cal E}_p^{\rm SLE})^2 = K_p(\tau_0) J_p(\tau_0) - (\partial_{\tau_0} J_p(\tau_0))^2/4$,
one obtains an expression which is initially quartic in $\Delta$. 
Repeated application of (\ref{deltaID}) then leads to (\ref{ssle11}). 
Since the right hand side of (\ref{ssle11}) is manifestly 
non-negative also the anticipated inequality (\ref{ssle4}) follows 
(without presupposing the minimization procedure). 
The result for the modulus-square follows from 
\begin{eqnarray}
	\label{ssle13}
K_p(\tau_0) \Delta(\tau,\tau_0)^2 + J_p(\tau_0) 
(\partial_{\tau_0} \Delta(\tau,\tau_0))^2 
- \Delta(\tau,\tau_0) \partial_{\tau_0} \Delta(\tau,\tau_0) \partial_{\tau_0} J_p(\tau_0) 
= J_p(\tau)\,
\end{eqnarray}
and can be verified along similar lines. Finally, the ratio 
$\Im T_p^{\rm SLE}/\Re T_p^{\rm SLE}$ can be read off from (\ref{ssle10}) 
and gives the $\tan$ of the phase. Initially the ratio 
has as denominator the left hand side of 
\begin{eqnarray}
	\label{ssle13a} 
2 J_p(\tau_0) \partial_{\tau_0} \Delta_p(\tau_0,\tau) \!-\! \partial_{\tau_0} J_p(\tau_0) 
\Delta_p(\tau_0,\tau) = 2 J_p(\tau,\tau_0)\,. 
\end{eqnarray}
The reduction of order occurs as before. 
\end{proof} 

\noindent
{\bf Remarks}

(i) Modulo the dependence on the averaging 
function the expression (\ref{ssle10}) realizes the goal of 
constructing a Hadamard state solely from the state independent 
commutator function in a way different from \cite{SJstates,SJHadamard}.  

(ii) The parts (a) and (b) are logically independent and (b) can be obtained 
solely from minimizing ${\cal E}_p$ in (\ref{ssle2}). A minimization over 
initial data in the Heisenberg picture is however less compelling 
because for selfinteracting QFTs the fields (as operator valued 
distributions) do in general not admit a well-defined restriction 
to a sharp constant time hypersurface. On the other hand, the 
Schr\"{o}dinger picture in QFT is frequently by default defined 
on a spatial lattice, see Proposition \ref{pr2.1} here. The Gaussian 
(\ref{Schroed7}) is then uniquely  determined by the parameters $z_p = 
T_p(\tau_0)$, $w_p = (\partial_{\tau} T_p)(\tau_0)$ in its initial value
$\Omega_T[u]|_{\tau = \tau_0}$. Conceptually, therefore (b) is naturally 
placed in the context of (a). 

(iii) The relation $|T_p^{\rm SLE}(\tau)| \propto \sqrt{J_p(\tau)}$  
also implies that $J(\tau)$ solves the Ermakov-Pinney equation 
with very specific $f$-dependent initial conditions implicitly 
set by those of $\Delta_p$.

(iv) In terms of the data in (\ref{ssle12}) the SLE two-point  
function can be expressed as 
\begin{eqnarray}
	\label{sletwop0} 
&\!\!\!\!\!\!\!\!\!\!& T_p^{\rm SLE}(\tau) T_p^{\rm SLE}(\tau')^* 
\nonumber \\[1.5mm]
&\!\!\!\!\!\!\!\!\!\!& = 
\frac{\sqrt{ J_p(\tau) J_p(\tau')}}{2 {\cal E}_p^{\rm SLE}} 
\left(
\frac{J_p(\tau,\tau_0) - i {\cal E}_p^{\rm SLE} \Delta_p(\tau,\tau_0)}%
{J_p(\tau,\tau_0) + i {\cal E}_p^{\rm SLE} \Delta_p(\tau,\tau_0)} 
\frac{J_p(\tau',\tau_0) + i {\cal E}_p^{\rm SLE} \Delta_p(\tau',\tau_0)}%
{J_p(\tau',\tau_0) - i {\cal E}_p^{\rm SLE} \Delta_p(\tau',\tau_0)} \right)^{1/2}
\!\!.
\end{eqnarray}

(v) In principle, the equivalence of (\ref{ssle10}) to the original 
expression (\ref{sle5}) 
is a consequence of the respective, independently established, 
uniqueness and the identity (\ref{Schroed14}). It is nevertheless 
instructive to verify the equivalence of (\ref{ssle10}) and 
(\ref{sle5}) directly. The main ingredient is the 
postponed proof of Proposition \ref{pr2.2} to which we now turn.

We begin with a simple basic fact 

\begin{lemma}\label{lm2.3} Let $\Delta: C([\tau_i,\tau_f] )\rightarrow C([\tau_i,\tau_f]^2)$ 
be the following commutator functional $\Delta[S](\tau,\tau_0) 
= i (S(\tau) S(\tau_0)^* - S(\tau)^* S(\tau_0) )$. Then 
$\Delta[S]$ is real valued, antisymmetric in $\tau,\tau_0$, and obeys 
$\Delta[a S + b S^*](\tau,\tau_0) = (|a|^2 - |b|^2 ) 
\Delta[S](\tau,\tau_0)$, $a,b \in \mathbb{C}$. On a solution $S$  
of the differential equation (\ref{FLode1}) $\Delta[S]$ becomes 
the commutator function, which is characterized by (\ref{Deltadef1})
and is independent of the choice of Wronskian normalized 
fiducial solution.  
\end{lemma} 

\begin{proof}[Proof of Proposition \ref{pr2.2}]\

We can regard $J_p(\tau_0), K_p(\tau_0)$ as functionals over 
the differentiable functions $C^1([\tau_i,\tau_f])$,  by 
replacing the commutator function by the commutator functional 
$\Delta_p(\tau,\tau_0) \mapsto \Delta_p[S](\tau,\tau_0)
= i (S(\tau) S(\tau_0)^* - S(\tau)^* S(\tau_0) )$. Inserting this  
into (\ref{ssle3}) and comparing with the 
definitions (\ref{sle4}) one finds 
\begin{eqnarray}
	\label{ssle14} 
J_p(\tau_0) &\! =\! & 2|S_p(\tau_0)|^2 c_1 - [S_p(\tau_0)^*]^2 c_2 - 
S_p(\tau_0)^2 c_2^* = {\cal J}[S](\tau_0)\,,
\nonumber \\[1.5mm]
K_p(\tau_0) &\! =\! & 2 |\partial_{\tau_0}S_p(\tau_0)|^2 c_1 - 
[\partial_{\tau_0} S_p(\tau_0)^*]^2 c_2 - 
[\partial_{\tau_0} S_p(\tau_0)]^2 c_2^* = {\cal K}[S](\tau_0)\,.
\end{eqnarray}  
Using (\ref{ssle14}) one can compute the left hand side of 
(\ref{ssle4}) in terms of $c_1,c_2$. The result is 
\begin{equation}
	\label{ssle15} 
4 K_p(\tau_0) J_p(\tau_0) - (\partial_{\tau_0} J_p(\tau_0))^2 = 
4(c_1^2 - |c_2|^2) = 4 {\cal I}[S]\,. 
\end{equation}
Since $c_1 \geq |c_2|$ this reconfirms (\ref{ssle4}). The invariance 
(\ref{sleinv2}) of ${\cal I}, {\cal J}, {\cal K}$ follows from 
Lemma \ref{lm2.3}.  
\end{proof}

Finally, we verify the equivalence of (\ref{ssle10}) and (\ref{sle5}). 
For a general solution $T_p(\tau)$ one can match the 
parameterizations (\ref{sle1}) and (\ref{ssle1}) by realizing 
the commutator function in terms of $S$. This 
gives 
\begin{eqnarray}
	\label{ssle16} 
\lambda &\! =\! & i \big(S_p(\tau_0)^* w - \partial_{\tau_0} S_p(\tau_0)^* z \big) \,,
\nonumber \\[1.5mm]
\mu &\! =\! & i \big(\partial_{\tau_0} S_p(\tau_0) z - S_p(\tau_0) w \big) \,.
\end{eqnarray}  
The same must hold for the minimizing parameters. A brute force 
verification of the latter is cumbersome. Instead we compare 
the modulus square computed from (\ref{sle2}), i.e.
$(|\mu|^2 + |\lambda|^2) |S_p(\tau)|^2 + \mu \lambda^* S(\tau)^2 + 
\lambda \mu^* [S(\tau)^*]^2$ with $J_p(\tau)/(2\sqrt{c_1^2 - |c_2|^2})$, 
taking advantage of the directly verified Eq.~(\ref{ssle13}). 
Inserting (\ref{ssle13}) for $J_p(\tau)$ and comparing 
coefficients of $|S_p(\tau)|^2$, $S_p(\tau)^2$, one finds 
\begin{equation}
	\label{ssle17}  
|\mu^{\rm min}|^2 + |\lambda^{\rm min}|^2 = \frac{c_1}{\sqrt{c_1^2 - |c_2|^2}}\,, 
\quad 
(\lambda^{\rm min})^* \mu^{\rm min}= - \frac{c_2^*}{ 2 \sqrt{ c_1^2 - |c_2|^2}}\,.
\end{equation}
These can be solved for $\mu^{\rm min}, \lambda^{\rm min}$, and with the 
choice of phase $\arg \lambda^{\rm min} = \pi - \arg c_2$ 
one recovers (\ref{sle6}). This provides a direct verification
-- modulo phase choices -- of (\ref{ssle16}) for the minimizers 
(\ref{ssle7}) and (\ref{sle6}). The phases are however not 
necessarily matched, in particular real $\mu$ does not automatically 
correspond to real $z$. 
\medskip

\newpage 
\section{Convergent small momentum expansion for SLE} \label{sec3}

The SLE have been introduced on account of their Hadamard 
property, which relates to a Minkowski-like behavior at 
large spatial momentum. Here we show that SLE admit a convergent 
small momentum expansion, both for massive and for massless 
theories. Remarkably, the momentum dependence turns out to 
Minkowski-like also for small momentum. In the massless case 
this provides a cure for the infrared divergences plaguing 
the two-point functions on FL cosmologies with accelerated 
expansion. In fact, for any scale factor the leading terms are 
given by
\begin{eqnarray}
	\label{sleIR} 
T_p^{\rm SLE}(\tau) T_p^{\rm SLE}(\tau')^* = \frac{\bar{a}}{2p} - 
\frac{i}{2} (\tau - \tau') + O(p)\,,
\quad 
\bar{a}:=\bigg(\frac{\int\! d\tau f(\tau)^2}{\int \!d\tau f(\tau)^2 
a(\tau)^{2d-2}}\bigg)^{\frac{1}{2}}\,. 
\end{eqnarray}

\subsection{Fiducial solutions and their Cauchy product}\label{sec3.1} 

A SLE can be defined either in terms of a fiducial solution $S_p$ 
or in terms of the Commutator function $\Delta_p$. Here we 
prepare results establishing uniformly convergent series for 
these solutions as well as their Cauchy products. Throughout 
we consider the differential equation 
\begin{equation}
	\label{odedef0a} 
[\partial_{\tau}^2 + \omega_p(\tau)^2] S_p(\tau) =0\,, \quad 
\omega_p(\tau)^2 = \omega_0(\tau)^2 + p^2 \omega_2(\tau)^2\,,  
\end{equation}
where $\omega_0, \omega_2$ are continuous real-valued
functions on $[\tau_i,\tau_f]$ and $\omega_2$ is not identically zero. 
The case $\omega_0(\tau)^2 = m(\tau)^2 a(\tau)^{2d}$, 
$\omega_2(\tau)^2 = a(\tau)^{2d -2}$ 
corresponds to the dispersion relation arising from the Klein Gordon equation;
the function $m(\tau)$ may have zeros or vanish identically 
(massless case). Throughout we write $p$ for the modulus of the spatial momentum.

\begin{proposition} \label{Ssmallp} The differential equation 
(\ref{odedef0a}) admits convergent series solutions 
with a radius of convergence $ p_*>0$ on $[\tau_i,\tau_f]$, 
such that for any $p<p_*$ 
\begin{eqnarray}
\label{Sseries} 
S_p(\tau)&\! =\! & \sum_{n=0}^\infty S_n(\tau)p^{2n}\,,
\quad {\rm and} \quad \partial_\tau S_p(\tau)=
\sum_{n=0}^\infty \partial_\tau S_n(\tau)p^{2n}\,,
\end{eqnarray}
and the sums converge {\it uniformly} on $[\tau_i,\tau_f]$.   
\end{proposition}

These solutions in particular have IR finite initial data
\begin{eqnarray}
	\label{SSeriesinit} 
\lim_{p \rightarrow 0} S_p(\tau_0) =: z_0 < \infty\,,\quad 
\lim_{p \rightarrow 0} \partial_{\tau_0} S_p(\tau_0) =: w_0 < \infty\,.
\end{eqnarray}
The proof below entails that the subspace of solutions described by 
the proposition can be characterized by (\ref{SSeriesinit}).
In order to prove the proposition, we shall need the following standard 
existence and uniqueness result for the solutions of a second order 
linear ODE (which we state without proof):
\begin{lemma}\ \label{existlemma}
Consider the initial value problem   
\begin{eqnarray}
y''(\tau)+\alpha(\tau)y'(\tau)+\beta(\tau)y(\tau)
=g(\tau)\,,\quad y(\tau_0)=u\,,\,\,y'(\tau_0)=v\,.
\end{eqnarray}
If $\alpha,\,\beta,\,g$ are continuous functions on an open interval 
$I\ni \tau_0$, then there exists a unique solution of this initial 
value problem, and this solution exists throughout the interval $I$.
\end{lemma}

\begin{proof}[Proof of Proposition \ref{Ssmallp}] \

First consider the ``$p=0$'' equation, 
{\it i.e.\ } $[\partial_{\tau}^2 + \omega_0(\tau)^2] S_0(\tau) =0$. 
Lemma \ref{existlemma} implies that there exists a complex solution  
$S_0(\tau)$, which may be Wronskian normalized to satisfy 
$\partial_\tau S_0\,S_0^\ast -S_0\partial_\tau S_0^\ast =-i$. In the case $\omega_0(\tau)=0$ 
on $[\tau_i,\,\tau_f]$, the solution with initial data $w_0,\,z_0$ 
is $S_0(\tau)=w_0(\tau-\tau_0)+z_0$, $w_0 z_0^* - z_0 w_0^* =-i$. 
Remaining with general $\omega_0(\tau)$ we reformulate 
(\ref{odedef0a}) as an integral equation. Defining the kernel%
\footnote{This is the (generalized) Feynman Greens function. 
Any other choice of Greens function also renders $L$ in \eqref{smallp6} a contraction, 
merely the value of $p_*$ may change.}  
\begin{eqnarray}\label{smallp1}
K(\tau,\,\tau'):=i\theta(\tau-\tau')S_0(\tau)S_0(\tau')^\ast 
+i\theta(\tau'-\tau)S_0(\tau)^\ast S_0(\tau') \,,
\end{eqnarray}
a function $S(\tau)$ satisfying 
\begin{eqnarray}\label{smallp2}
S(\tau)&\! =\! & S_0(\tau)-p^2\int_{\tau_i}^{\tau_f}\!K(\tau,\,\tau')
\omega_2(\tau')^2S(\tau')d\tau'
\end{eqnarray}
solves (\ref{odedef0a}). Further, $\partial_{\tau}S(\tau)$ satisfies
\begin{eqnarray}\label{smallp3}
\partial_{\tau} S(\tau)= \partial_{\tau}S_0(\tau)-p^2\int_{\tau_i}^{\tau_f}
\partial_\tau K(\tau,\tau')\omega_2(\tau')^2S(\tau')d\tau'\,.
\end{eqnarray}
In terms of 
\begin{eqnarray}\label{smallp4}
	\mathcal{S}(\tau):=
	\begin{pmatrix}
		S(\tau)\\
		\tilde{S}(\tau)
	\end{pmatrix},\quad 
	\mathcal{S}_0(\tau):=
	\begin{pmatrix}
		S_0(\tau)\\
		\partial_\tau S_0(\tau)
	\end{pmatrix},\quad 
	 {\cal K}(\tau,\tau'):=
	\begin{pmatrix}
		K(\tau,\,\tau')\omega_2(\tau')^2&0\\
		\partial_\tau K(\tau,\tau')\omega_2(\tau')^2&0
	\end{pmatrix},
\end{eqnarray}
we search for a solution of the integral equation 
\begin{eqnarray}\label{smallp5}
\mathcal{S}(\tau )=\mathcal{S}_0(\tau)
-p^2\int_{\tau_i}^{\tau_f} {\cal K}(\tau,\tau')\mathcal{S}(\tau')d\tau'\,.
\end{eqnarray}
As the underlying Banach space we take $(X,\left\lVert\cdot\right\rVert)
:=\big(C([\tau_i,\tau_f],\mathbb{C}^2),\left\lVert\cdot\right\rVert)_{\sup}\big)$, where 
$\mathbb{C}^2$ is being equipped with the sup-norm. 
Next, we define the linear operator $L:X\to X$
\begin{eqnarray}\label{smallp6}
\forall\,u\in X:\,\big(Lu\big)(\tau):=
\mathcal{S}_0(\tau) -p^2\int_{\tau_i}^{\tau_f}\!{\cal K}(\tau,\,\tau')u(\tau')d\tau'\,,
\end{eqnarray}
and show that for sufficiently small $p$, this map is actually a contraction.

Since $S_0$ is a $C^1$ function, it is clear that both 
$K(\tau,\tau')$ and $\partial_\tau K(\tau,\tau')$ are bounded functions 
on $[\tau_i,\tau_f]^2$. As $\omega_2$ is also continuous, there is $R>0$ such that 
$|{\cal K}(\tau,\tau')_{ij}|<R$ on $[\tau_i,\tau_f]^2$. Then for any $u,\,v\in X$
\begin{eqnarray}
\label{smallp7}
|Lu(\tau)-Lv(\tau)|_{\max }&\! =\! & 
p^2\Big|\int_{\tau_i}^{\tau_f}\!{\cal K}(\tau,\tau')
\big(u(\tau')-v(\tau')\big)d\tau'\Big|_{\max }
\nonumber \\[1.5mm] 
&\leq & p^2\int_{\tau_i}^{\tau_f}\!
\Big|{\cal K}(\tau,\tau')\big(u(\tau')-v(\tau')\big)\Big|_{\max }d\tau'
\nonumber \\[1.5mm]
\implies \left\lVert Lu-Lv\right\rVert_{\sup}&\leq &  p^2(\tau_f-\tau_i) R  \left\lVert u-v\right\rVert_{\sup}\,,
\end{eqnarray}	
and so there is $p_*>0$ such that for all $p<p_*$,  $L$ is a contraction.	

Assuming that $p<p_*$, the Banach Fixed Point theorem implies that there 
exists a {\it unique} $\mathcal{S}_p=(S_p,\,\tilde{S}_p)^T\in X$ such 
that $L\mathcal{S}_p=\mathcal{S}_p$, {\it i.e.\ }
\begin{eqnarray}\label{smallp8}
S_p(\tau)&\! =\! & S_0(\tau)-p^2\int_{\tau_i}^{\tau_f}\!K(\tau,\,\tau')
\omega_2(\tau')^2S_p(\tau')d\tau'\,,
\nonumber \\[1.5mm] 
\tilde{S}_p(\tau)&\! =\! & \partial_\tau S_0(\tau)-p^2\int_{\tau_i}^{\tau_f}\!
\partial_\tau K(\tau,\,\tau')\omega_2(\tau')^2S_p(\tau')d\tau'\,.
\end{eqnarray}
Comparing (\ref{smallp8}) and (\ref{smallp3}), it is clear that 
$\partial_\tau S_p(\tau)$ satisfies the second equation above. The 
uniqueness of the fixed point $\mathcal{S}_p$ then implies that 
$\tilde{S}_p=\partial_\tau S_p$.

Further, the iterated sequence $L^m\mathcal{S}_0,\,m\in \mathbb{N}$, 
converges to $\mathcal{S}_p$ in the sup-norm. It is then easily 
verified that there is a sequence of 
$C^1$ functions $S_n(\tau)$ such that we have the {\it uniformly convergent} 
power series representations of the form asserted in (\ref{Sseries}).
\end{proof}

Next we consider the product of two series solutions and state, without proof, the following slight generalization of Merten's theorem.
\begin{lemma}\ \label{lm..2} Let 
\begin{eqnarray}
A(\tau)=\sum_{n=0}^\infty a_n(\tau)p^{2n}\,,\quad 
B(\tau)=\sum_{n=0}^\infty b_n(\tau)p^{2n}\,,
\end{eqnarray}
be power series in the Banach space $C([\tau_i,\tau_f],\mathbb{C})$ with 
radius of convergence $p_*>0$. 
Consider the map $C:[\tau_i,\tau_f]\times [\tau_i,\tau_f]\to \mathbb{C}$ defined 
by $C(\tau_1,\tau_2):=A(\tau_1)B(\tau_2)$, and the coefficients of the 
unequal time Cauchy product of $A$ and $B$,
\begin{eqnarray}
\label{Cprod1} 
c_n(\tau_1,\tau_2)&:=&\sum_{i=0}^na_i(\tau_1)b_{n-i}(\tau_2)\,.
\end{eqnarray}
Then for any $p<p_*$ 
\begin{eqnarray}
\label{Cprod2} 
\sum_{n=0}^\infty c_n(\tau_1,\tau_2)p^{2n}&\! =\! & C(\tau_1,\tau_2)\,,
\end{eqnarray}
with uniform convergence in $[\tau_i,\tau_f]\times [\tau_i,\tau_f]$.
The same holds for the equal time Cauchy product ($\tau_1 = \tau_2$
in (\ref{Cprod1}), (\ref{Cprod2}) ) with uniform convergence in 
$[\tau_i,\tau_f]$. 
\end{lemma}

An immediate corollary of Proposition \ref{Ssmallp}  and Lemma 
\ref{lm..2} is: 

\begin{corollary}\label{cor3.3} The Commutator function $\Delta_p(\tau,\tau')$ 
and the Greens functions defined in terms of it 
have uniformly convergent series expansions in $p< p_*$ for 
distinct $(\tau, \tau') \in [\tau_i,\tau_f]\times [\tau_i,\tau_f]$. 
\end{corollary}

So far these are mostly existence results. For the actual construction of 
these series solutions one will solve the implied recursion relations. 
For a solution $S_p(\tau)$ of the form (\ref{Sseries}) one has
\begin{eqnarray}
	\label{Srec1} 
[\partial_{\tau}^2 + \omega_0(\tau)^2 ]S_0(\tau) &\! =\! & 0\,, 
\nonumber \\[1.5mm]
[\partial_{\tau}^2 + \omega_0(\tau)^2 ]S_n(\tau) &\! =\! & - \omega_2(\tau)^2 S_{n-1}(\tau) \,, 
\quad n \geq 1\,.
\end{eqnarray}
Each $S_n$ is only unique up to addition of a solution of the homogeneous 
equation, characterized by two complex parameters. These ambiguities 
account for the initial data of the series solution 
\begin{eqnarray}
	\label{Srec2}
&& S_p(\tau_0) = \sum_{n \geq 0} z_n p^{2n} =: z_p\,, \quad 
\partial_{\tau_0} S_p(\tau_0) = \sum_{n \geq 0} w_n p^{2n} =: w_p \,,
\nonumber \\[1.5mm]
&& \quad \mbox{with} \quad 
\sum_{j=0}^{n} (w_j z^*_{n-j} - w_j^* z_{n-j})  =0\,,\quad n \geq 1\,,
\end{eqnarray}
where the constraint stems from the Wronskian normalization. 
One can use the same Greens function $G_0(\tau',\tau)$ at each 
order and adjust the initial data 
of the additive modification such that $S_n(\tau_0) = z_n$, 
$(\partial_{\tau} S_n)(\tau_0) = w_n$ holds, for given $z_n,w_n \in \mathbb{C}$, 
mildly constrained by (\ref{Srec2}). 

Later on a series solution of this form will play the role of the 
fiducial solution in the construction of the SLE.  Theorem \ref{th2.1} ensures 
that {\it any} such solution will produce the {\it same} SLE solution 
(within the implied radius of convergence) up to a phase. We are therefore 
free to choose one with especially simple, namely $p$-independent, initial 
data for $\tau_0 = \tau_i$: $z_n = 0 = w_n$, $n \geq 1$. In this case the 
relevant Greens function is the retarded Greens function 
$G^{\wedge}_0(\tau,\tau') := \theta(\tau-\tau') 
\Delta_0(\tau,\tau')$, with $\Delta_0$ the commutator function for 
$\partial_{\tau}^2 + \omega_0(\tau)^2$. Further, no additive, order dependent, 
modification is needed and the solution of the iteration is simply 
\begin{eqnarray}
	\label{Srec3} 
S_n(\tau) \!\!&\! =\! & \!\!\int_{\tau_i}^{\tau_f} \! d\tau' \, 
K_n(\tau,\tau')\, S_0(\tau')\,, \quad n \geq 1\,,
\\[2mm] 
K_1(\tau,\tau')  \!&\!:=\!&\! - G_0^{\wedge}(\tau,\tau') \omega_2(\tau')^2\,,
\nonumber \\[1.5mm]
K_{n+1}(\tau,\tau') \!&\!:=\!&\! (-)^{n+1} 
\int_{\tau_i}^{\tau_f} \!\! d\tau_1...d\tau_n 
\, G_0^{\wedge}(\tau,\tau_1) \omega_2(\tau_1)^2 
\, G_0^{\wedge}(\tau_1,\tau_2) \omega_2(\tau_2)^2  \ldots 
\, G_0^{\wedge}(\tau_n,\tau') \omega_2(\tau')^2\,. 
\nonumber
\end{eqnarray}    
The kernel $K_n$ is manifestly real and satisfies $K_n(\tau_i, \tau') =0= 
\partial_{\tau} K_n(\tau,\tau')|_{\tau = \tau_i}$, for $\tau' \in (\tau_i, \tau_f]$. 
The associated series solution $S_p(\tau)$ therefore satisfies
$S_p(\tau_i) = z_0$, $(\partial_{\tau} S_p)(\tau_i) = w_0$, for $p$-independent 
constants with $w_0 z_0^* - w_0^* z_0 = -i$.

The commutator function $\Delta_p(\tau,\tau')$ is likewise independent 
of the choice of the Wronskian normalized solution used to realize it, 
see Lemma \ref{lm2.3}. We are thus free to use the solution (\ref{Srec3}) 
for this purpose. Writing $\Delta_p(\tau,\tau') = 
\sum_{n\geq 0} \Delta_n(\tau,\tau') p^{2n}$, one finds 
\begin{eqnarray}
	\label{Deltaexp1} 
\Delta_n(\tau,\tau') &\! =\! & i \sum_{j=0}^n \big(S_j(\tau) S^*_{n-j}(\tau') - 
S_j^*(\tau) S_{n-j}(\tau') \big) 
\nonumber \\[1.5mm]
&\! =\! & \int_{\tau_i}^{\tau_f} \! ds [ K_n(\tau,s)\Delta_0(s,\tau')  - 
K_n(\tau',s) \Delta_0(s,\tau)]
\nonumber \\[1.5mm]
&+& 
\int_{\tau_i}^{\tau_f} \! ds_1 ds_2 \, 
\sum_{j=1}^{n-1} K_j(\tau, s_1) K_{n-j}(\tau',s_2) \Delta_0(s_1,s_2) \,.
\end{eqnarray}
One can check that the coefficients satisfy all the relations implied
by the expansion of the defining conditions (\ref{Deltadef1}) 
\begin{eqnarray}
	\label{Deltaexp2} 
[\partial_{\tau}^2 + \omega_0(\tau)^2 ]\Delta_n(\tau,\tau') &\! =\! & - \omega_2(\tau)^2
\Delta_{n-1}(\tau,\tau') \,, 
\quad 
\partial_{\tau} \Delta_n(\tau, \tau') \big|_{\tau = \tau'} =0\,,
\nonumber \\[1.5mm]
[\partial_{\tau'}^2 + \omega_0(\tau')^2]\Delta_n(\tau,\tau') &\! =\! & 
- \omega_2(\tau')^2 \Delta_{n-1}(\tau,\tau') \,, \;\quad n \geq 1\,.
\end{eqnarray}
The two recursion relations follow from $[\partial_{\tau}^2 + \omega_0(\tau)^2 ] 
K_n(\tau,\tau') = - \omega_2(\tau)^2 K_{n-1}(\tau,\tau')$, $n \geq 2$. 
For the third relation it is convenient to first verify $\partial_{\tau} 
[ \partial_{\tau} \Delta_n (\tau,\tau') |_{\tau = \tau'}] =0$. Then, 
it suffices to show  $\partial_{\tau} \Delta_n (\tau,\tau_i) |_{\tau = \tau_i} =0$, 
which follows from $K_n(\tau_i, \tau') =0=
\partial_{\tau} K_n(\tau,\tau')|_{\tau = \tau_i}$, for $\tau' \in (\tau_i, \tau_f]$. 


\subsection{IR Behavior of States of Low Energy}\label{sec3.2}

We use the formulas from Theorem \ref{sslethm} to derive 
convergent series expansions for the SLE. The basic expansion is 
$\Delta_p(\tau',\tau) = \sum_{n\geq 0} \Delta_n(\tau',\tau) p^{2n}$, 
with coefficients from (\ref{Deltaexp1}). In terms of it convergent 
expansions for the $J_p(\tau_0), \partial_{\tau_0} J_p(\tau_0), K_p(\tau_0)$ 
in (\ref{ssle3}) can be derived. The uniform convergence of the 
various pointwise products is ensured by the results of Section \ref{sec3.1} and 
allows one to exchange the order of summation and integration. 
The following notation is convenient 
\begin{eqnarray}
	\label{sleexp1} 
&& C(\tau,\tau_0) = \sum_{n \geq 0} C_n(\tau,\tau_0) \,p^{2n} \quad 
\Longrightarrow \quad 
C(\tau,\tau_0)^2 = \sum_{n \geq 0} C(\tau,\tau_0)^2_n \,p^{2n} 
\nonumber \\[1.5mm] 
&& \mbox{with} \quad C(\tau,\tau_0)^2_n := \sum_{j=0}^n 
C_j(\tau,\tau_0) C_{n-j}(\tau, \tau_0) \,.
\end{eqnarray}
In this notation one has 
\begin{eqnarray}
	\label{sleexp2}
J_p(\tau_0) &\! =\! & \sum_{n \geq 0} J_n(\tau_0) \,p^{2n}\,, \quad 
\quad 
K_p(\tau_0) = \sum_{n \geq 0} K_n(\tau_0) \,p^{2n}\,, 
\\[2mm] 
J_0(\tau_0) &\! =\! & \frac{1}{2} \int\! d\tau \, f(\tau)^2 
\Big[ \big( \partial_{\tau} \Delta_0(\tau,\tau_0) \big)^2 + 
\omega_0(\tau)^2 \Delta_0(\tau,\tau_0)^2 \Big] \,,
\nonumber \\[1.5mm]
J_n(\tau_0) &\! =\! & \frac{1}{2} \int\! d\tau \, f(\tau)^2 
\Big[ \big( \partial_{\tau} \Delta(\tau,\tau_0) \big)^2_n + 
\omega_0(\tau)^2 \Delta(\tau,\tau_0)^2_n + 
\omega_2(\tau)^2 \Delta(\tau,\tau_0)^2_{n-1} \Big] \,,
\nonumber \\[1.5mm]
K_0(\tau_0) &\! =\! & \frac{1}{2} \int\! d\tau \, f(\tau)^2 
\Big[ \big( \partial_{\tau} \partial_{\tau_0} \Delta_0(\tau,\tau_0) \big)^2 + 
\omega_0(\tau)^2 \big(\partial_{\tau_0} \Delta_0(\tau,\tau_0)\big)^2 \Big]\,, 
\nonumber \\[1.5mm]
K_n(\tau_0) &\! =\! & \frac{1}{2} \int\! d\tau \, f(\tau)^2 
\Big[ \big( \partial_{\tau} \partial_{\tau_0}\Delta(\tau,\tau_0) \big)^2_n + 
\omega_0(\tau)^2 \big( \partial_{\tau_0} \Delta(\tau,\tau_0) \big)^2_n 
+ \omega_0(\tau)^2 \big( \partial_{\tau_0} \Delta(\tau,\tau_0) \big)^2_{n-1} \Big] \,,
\nonumber
\end{eqnarray}
and $\partial_{\tau_0} J_p(\tau_0) = \sum_{n \geq 0} \partial_{\tau_0} J_n(\tau,\tau_0) 
\,p^{2n}$ with the implied coefficients. Interpreting (\ref{ssle11}) as 
\begin{eqnarray}
	\label{sleexp3} 
({\cal E}_p^{\rm SLE})^2 &\! =\! & \frac{1}{4} \int\! d\tau_0 f(\tau_0)^2 
\big[ K_p(\tau_0) + \omega_p(\tau_0)^2 J_p(\tau_0) \big] =: 
\sum_{n \geq 0} \varepsilon_n^2 \,p^{2n} \,,
\nonumber \\[1.5mm]
\varepsilon_0^2 &\! =\! & \frac{1}{4}   \int\! d\tau_0 f(\tau_0)^2 
\big[ K_0(\tau_0) + \omega_0(\tau_0)^2 J_0(\tau_0) \big]\,,
\\[2mm] 
\varepsilon_n^2 &\! =\! & \frac{1}{4}   \int\! d\tau_0 f(\tau_0)^2 
\big[ K_n(\tau_0) + \omega_0(\tau_0)^2 J_n(\tau_0) + 
\omega_2(\tau_0)^2 J_{n-1}(\tau_0)\big]\,,\quad n \geq 1\,,
\nonumber
\end{eqnarray}
one sees that the energy's expansion is determined by the same 
coefficients. As a consequence all quantities in Theorem 
\ref{sslethm}(b) admit convergent series expansions in powers of 
$p$ whose coefficients can be expressed in terms of those in 
(\ref{sleexp2}) only. 

In the following we focus on the expansion of 
the energy ${\cal E}_p^{\rm SLE}$ and the modulus squared $|T_p^{\rm SLE}(\tau)|^2$. 
It is useful to distinguish two cases (where the terminology will 
become clear momentarily). 
\medskip 

{\bf Massive:} $\varepsilon_0 >0$ and $K_0(\tau_0) >0$.  
\begin{eqnarray}
	\label{sleexp4} 
{\cal E}_p^{\rm SLE} &\! =\! & \varepsilon_0 + \frac{\varepsilon_1^2}{2 \varepsilon_0} p^2 
- \frac{\varepsilon_1^4 - 4 \varepsilon_0^2 \varepsilon_2^2}{8 \varepsilon_0^3} p^4 
+ O(p^6) \,,
\nonumber \\[1.5mm]
|T_p^{\rm SLE}(\tau)|^2 &\! =\! & \frac{J_0(\tau)}{2 \varepsilon_0} + 
\frac{ 2 J_1(\tau) \varepsilon_0^2 - J_0(\tau) \varepsilon_1^2}{4 \varepsilon_0^3} p^2  
\\
&+& \frac{1}{16 \varepsilon_0^5} 
\big( 8 J_2(\tau) \varepsilon_0^4 - 4 J_1(\tau) \varepsilon_0^2 \varepsilon_1^2 + 3 
J_0(\tau) \varepsilon_1^4 - 4 J_0(\tau) \epsilon_0^2 \varepsilon_2^2 \big) p^4 + 
O(p^6) \,.
\nonumber
\end{eqnarray}

{\bf Massless:} $\varepsilon_0 =0$ and $K_0(\tau_0) =0$ and $\varepsilon_1>0$.  
\begin{eqnarray}
	\label{sleexp5} 
{\cal E}_p^{\rm SLE} &\! =\! & \varepsilon_1 p + \frac{\varepsilon_2^2}{2 \varepsilon_1} p^3 
- \frac{\varepsilon_2^4 - 4 \varepsilon_1^2 \varepsilon_3^2}{8 \varepsilon_1^3} p^5 
+ O(p^7) \,,
\nonumber \\[1.5mm]
|T_p^{\rm SLE}(\tau)|^2 &\! =\! & \frac{J_0(\tau)}{ 2 \varepsilon_1} \frac{1}{p}    
+ \frac{ 2 J_1(\tau) \varepsilon_1^2 - J_0(\tau) \varepsilon_2^2}{4 \varepsilon_1^3} p  
\\
&+& \frac{1}{16 \varepsilon_1^5} 
\big( 8 J_2(\tau) \varepsilon_1^4 - 4 J_1(\tau) \varepsilon_1^2 \varepsilon_2^2 + 3 
J_0(\tau) \varepsilon_2^4 - 4 J_0(\tau) \epsilon_1^2 \varepsilon_3^2 \big) p^3 + 
O(p^5) \,.
\nonumber
\end{eqnarray}
 
The massive case corresponds to $\omega_0(\tau) = m(\tau)^2 a(\tau)^{2d}$, 
$\omega_2(\tau)^2 = a(\tau)^{2 d-2}$. Even the lowest order 
commutator function $\Delta_0(\tau,\tau')$ can then in general 
no longer be found in closed form. All other aspects of the 
expansions are however explicitly computable in terms of $\Delta_0$:
the $\Delta_n$'s via (\ref{Deltaexp1}), the $J_n,K_n$'s via (\ref{sleexp2}), 
the $\varepsilon_n$'s from (\ref{sleexp3}), and hence everything else.  
\medskip

{\bf Two-point function of massless SLE.} The massless case 
corresponds to $\omega_0(\tau) = 0$, $\omega_2(\tau)^2 
= a(\tau)^{2 d-2}$. The lowest order wave equation in (\ref{Srec1}) is 
then trivially soluble: $S_0(\tau) = w_0(\tau\!-\!\tau_0) + z_0$, with
$w_0 z_0^* - w_0^* z_0 = -i$. The coefficients of the commutator 
function are explicitly known 
\begin{eqnarray}
	\label{massless1} 
\Delta_0(\tau',\tau) &\! =\! & \tau' - \tau\,, 
\nonumber \\[1.5mm]
\Delta_1(\tau',\tau) &\! =\! & \int_{\tau_i}^{\tau_f} \! ds [ \theta(\tau\!-\!s) - 
\theta(\tau'\!-\!s)] (\tau\!-\!s) (\tau'\!-\!s) a(s)^{2d-2}\,,
\end{eqnarray}
etc. This entails $K_0(\tau_0)=0$, $\varepsilon_0=0$, and 
\begin{eqnarray}
	\label{massless2}
\varepsilon_1^2 &\! =\! & \frac{1}{4} 
\int\!d\tau f(\tau)^2 \int\! d\tau' f(\tau')^2 a(\tau')^{2 d-2}\,,
\nonumber \\[1.5mm]
J_0(\tau_0) &\! =\! &  \frac{1}{2} \int\!d\tau f(\tau)^2 \,, \quad 
J_1(\tau_0) = \int\! d\tau f(\tau)^2 \big[ 
\partial_{\tau}\Delta_1(\tau,\tau_0) + (\tau\!-\!\tau_0)^2 a(\tau)^{2 d-2} \big]\,, 
\nonumber \\[1.5mm]
K_1(\tau_0) &\! =\! & \frac{1}{2} \int\! d\tau f(\tau)^2 a(\tau)^{2 d-2} \,.
\end{eqnarray}
This gives 
\begin{eqnarray}
	\label{massless3}  
|T_p^{\rm SLE}(\tau)|^2 = \frac{\bar{a}}{2p}+O(p)\,,
\quad {\cal E}_p^{\rm SLE} = \frac{p}{2 \bar{a} } \int\!d\tau f(\tau)^2 \,,
\quad 
\bar{a}:=\bigg(\frac{\int\! d\tau f(\tau)^2}{\int \!d\tau f(\tau)^2 
a(\tau)^{2d-2}}\bigg)^{\frac{1}{2}}\,,
\end{eqnarray}
as claimed in (\ref{sleIR}). Since the leading term is $\tau$ 
independent one obtains from (\ref{ssle10}) 
\begin{eqnarray}
	\label{massless4} 
T_p^{\rm SLE}(\tau) &\! =\! & \Delta_p(\tau,\tau_0) w_p^{\rm min} - 
\partial_{\tau_0} \Delta_p(\tau,\tau_0) z_p^{\rm min} = 
\sqrt{\frac{\bar{a}}{2 p} } - i (\tau-\tau_0) \sqrt{\frac{p}{2 \bar{a}}}
+ O(p^{3/2})\,.
\nonumber \\[1.5mm]
z_p^{\rm min} &\! =\! & \sqrt{\frac{\bar{a}}{2 p}}\Big(1 + O(p^2)\Big) \,,
\quad 
w_p^{\rm min} = - i \sqrt{\frac{ p}{2 \bar{a}}}\Big(1 + O(p^2) \Big)\,.
\end{eqnarray}
This holds up an undetermined $p$-dependent phase which 
is fixed in the initial value formulation of the minimization 
procedure by taking $z$ real. This phase ambiguity disappears 
in the two-point function, for which one obtains 
\begin{equation}
	\label{sletwop3} 
T_p^{\rm SLE}(\tau) T_p^{\rm SLE}(\tau')^* = \frac{\bar{a}}{2p} - 
\frac{i}{2} (\tau - \tau') + O(p)\,.
\end{equation}
The same result can alternatively be obtained from (\ref{sletwop0}).   
\medskip 

\noindent
{\bf Remarks.}

(i) Based on (perhaps mislead by) the exactly soluble case of 
power-like scale factors one normally regards the  IR behavior of 
the solutions as directly determined by the cosmological 
scale factor. From the small argument expansion of the Bessel 
functions one has 
\begin{eqnarray}
	\label{powerlike}
|S_p(\tau)|^2 \propto p^{- 2 |\nu|}\quad \mbox{for} \quad 
a(\tau) \propto \tau^{\frac{1 - 2 \nu}{2(d-1) \nu}}\,.
\end{eqnarray}
Here $d/2< \nu <\infty$ corresponds to acceleration 
while $- \infty < \nu < 1/2$ corresponds to deceleration. 
The interval $1/2 \leq \nu \leq d/2$ does not give rise to a 
curvature singularity; the boundary values $\nu =1/2$ and $\nu= d/2$ 
model Minkowski space and deSitter space, respectively. 
The inverse Fourier transform is infrared finite whenever 
$\int_0^1 \! dp \,p^{d-1} |S_p(\tau)|^2 $ is finite. For the 
solutions (\ref{powerlike}) this is the case only in part of the 
decelerating window, $0 < \nu <1/2$, see \cite{FRWinfra1} for the 
original discussion.     

(ii) The leading IR behavior of the massless SLE solution 
(\ref{massless4}) is constant, pointwise in $\tau$. This corresponds 
to the expected freeze-out of the 
oscillatory behavior on scales much larger than the Hubble radius. 
The universality of the $1/\sqrt{p}$ behavior is however surprising,
as is the simple coefficient $\sqrt{\bar{a}/2}$, valid for {\it any} 
scale factor. The result (\ref{massless4}) could not have been 
obtained based on the traditional adiabatic iteration, which is 
incurably singular at small momentum. 

(iii) In arriving at (\ref{massless4}) we took the expressions 
from Theorem \ref{sslethm} as the starting point. It is instructive to 
go through the derivation based on the original parameterization  
(\ref{sle1}), (\ref{sle3}). The fiducial solution is 
constructed via (\ref{Srec3}) from its leading order, $S_0$. 
In the massless case the general 
(Wronskian normalized) solution to the leading order equation 
is $S_0(\tau) = w_0(\tau\!-\!\tau_0) + z_0$, $w_0 z_0^* - w_0^* z_0 = -i$. 
A somewhat longer computation then gives 
\begin{eqnarray}
\label{massless8}
\mu_p&\! =\! & |w_0|\sqrt{\frac{\bar{a}}{2p}}-
\frac{1}{|w_0|}\sqrt{\frac{p}{8\bar{a}}}+O(p^{\frac{3}{2}})\,,
\nonumber \\[1.5mm]
\lambda_p &\! =\! & -\frac{w_0^\ast}{w_0} |w_0|\sqrt{\frac{\bar{a}}{2p}}
-\frac{w_0^\ast}{w_0}\frac{1}{|w_0|}\sqrt{\frac{p}{8\bar{a}}}+O(p^{\frac{3}{2}})\,.
\\[2mm]
T_p^{\rm SLE}(\tau)
&\! =\! & -i\frac{w_0^\ast}{|w_0|}\sqrt{\frac{\bar{a}}{2p}}
-\frac{1}{|w_0|w_0}\sqrt{\frac{p}{8\bar{a}}}
\Big[2|w_0|^2(\tau-\tau_0)+2\Re(z_0w_0^\ast)\Big]+O(p^{\frac{3}{2}})\,.
\nonumber
\end{eqnarray}
One sees that all intermediate results depend on the 
parameters $w_0,z_0$ of the fiducial solution. In the two-point
function, however, these drop out and one recovers (\ref{sletwop3}).

(iv) While in the massive case minimization of ${\cal E}_p$ and 
expansion in $p^2$ are commuting operations, this is not 
the true in the massless case. In the SLE construction 
via a fiducial solution we chose one with a regular $p \rightarrow 0$ limit,
which is evidently not the case for (\ref{massless4}). The 
independence of the SLE solution from the choice of 
fiducial solution is crucial for the result.

(v) The IR behavior of (\ref{sletwop3}) 
is Minkowski-like for {\it all} scale factors $a$. This means that 
massless SLE are automatically IR finite and provide an elegant
solution to the long standing IR divergences in Friedmann-Lema\^{i}tre 
backgrounds with accelerated expansion \cite{FRWinfra1}. 

(vi) The existence of a pre-inflationary epoch with non-accelerated 
expansion typically removes the IR singularity. For generic powerlike 
scale factors the mode matching can (with some effort) be controlled 
analytically \cite{InflIRmatch3}; typically one focuses on a 
radiation dominated ($\nu=-1/2$ in (\ref{powerlike})) 
\cite{InflIRmatch4,InflIRmatch5}
or kinetic energy dominated ($\nu=0$ in (\ref{powerlike})) 
\cite{InflIRmatch0,InflIRmatch1} pre-inflationary period. 
Another take on the IR issue is to regard it as an artifact 
of using non-gauge invariant observables \cite{FRWinfra2,FRWinfra3}.

(vii) The mathematical principle underlying (\ref{sletwop3}) is very 
different from the ones in (vi).  As detailed in Section \ref{sec5}, there 
are independent reasons to regard the existence of a pre-inflationary 
period as part of the standard paradigm. Positing a massless SLE 
as primordial vacuum in this period  then ought to be consistent 
with the qualitative properties of the power spectrum at seed 
formation. This physics requirement will be taken up in Section \ref{sec5.2}.

(viii) As a consequence of (\ref{sletwop3}) the long range properties 
of the SLE position space two-point function will be similar to that 
of its Minkowski space counterpart. Further, 
the shift symmetry, $\phi(\tau,x) \mapsto \phi(\tau, x) + {\rm const}$, 
turns out to be spontaneously broken for $d \geq 2$, as it is 
for the massless free field in Minkowski space. A proper proof 
can be based on Swieca's Noether charge criterion \cite{Swieca,Lopusbook} 
and is omitted here.

\newpage 
\section{WKB type large momentum asymptotics} \label{sec4}

Any Wronskian normalized solution of the basic wave equation 
is uniquely determined by its modulus 
\begin{eqnarray}
	\label{Spviamod} 
S_p(\tau) = |S_p(\tau)| \exp\Big\{\!\!- \frac{i}{2} 
\int_{\tau_0}^{\tau} \! ds \frac{1}{ |S_p(s)|^2} \Big\}\,,  
\end{eqnarray}
up to a choice of $\tau_0$ where $S_p(\tau_0)$ is real. 
In this section we show that for each $N>1$ there exists an exact 
`order $N$' solution with a certain 
$N$-term positive frequency asymptotics. These solutions are such that 
$|S_p(\tau)|^2$ is asymptotic up to $O(p^{-2N-1})$ to a 
polynomial in odd inverse powers of $p$, whose coefficients are 
{\it local} differential polynomials in $\omega_0, \omega_2$ 
generalizing the heat kernel coefficients. The resulting order 
$N$ solutions will be referred to as WKB type solutions.%
\footnote{A WKB ansatz proper is one where only the integrand of the 
exponent is formally expanded in terms of local coefficients.}
An SLE solution will then be shown to be a WKB type solution 
of {\it infinite} order. Throughout this section we assume $\omega_0,\,\omega_2$ to be smooth.

\subsection{Existence of solutions with WKB type asymptotics} \label{sec4.1}

As a starting point the relation (\ref{Spviamod}) is cumbersome because  
the exponential needs to be re-expanded. In the following we 
establish the existence of asymptotic expansions of all 
quantities needed by starting from a simplified formal series 
ansatz for $S_p$'s large momentum asymptotics 
\begin{eqnarray}
\label{asym2.1}
S_p(\tau)=
\frac{\exp\big\{-ip\int_{\tau_i}^\tau ds\,\omega_2(s)\big\}}%
{\sqrt{2p\omega_2(\tau)}}\Big\{1+\sum_{n\geq 1}(ip)^{-n} s_n(\tau)\Big\}\,,
\end{eqnarray}
with real-valued $s_n$. As in Section \ref{sec3} we consider the basic 
differential equation 
$[\partial_{\tau}^2 + \omega_p(\tau)^2] S_p(\tau) =0$ with generic time 
dependent frequency $\omega_p(\tau) = \omega_0(\tau)^2 + p^2 
\omega_2(\tau)^2$. The leading term in (\ref{asym2.1}) is a positive 
frequency wave. The latter is known to be a necessary (but by no 
means sufficient property) for a solution to comply with the 
Hadamard condition.

Upon insertion of (\ref{asym2.1}) into the basic wave equation 
one finds the following recursion relations  
\begin{eqnarray}
\label{asym2.2}
\partial_{\tau} s_n &\! =\! &  \partial_{\tau} s_1 s_{n-1} + \partial_{\tau} \Big( 
\frac{\partial_{\tau} s_{n-1}}{2 \omega_2} \Big)\,, \quad n \geq 2\,,
\nonumber \\[1.5mm]
\partial_{\tau} s_1 &\! =\! & \frac{\omega_0^2}{2 \omega_2} - \frac{1}{4 \omega_2} 
\bigg( \frac{\partial_{\tau}^2 \omega_2}{\omega_2} - \frac{3}{2} 
\bigg( \frac{\partial_{\tau} \omega_2}{ \omega_2} \Big)^2 \bigg) \,.
\end{eqnarray}	
Clearly, each $s_n$ can be obtained simply by integration and 
the only ambiguity arises from the choice of integration constants 
$s_n(\tau_i)$. We claim that 
\begin{equation}
	\label{asym2.2a}
s_n(\tau_i) =0\,, \;\;\;n\; \mbox{odd} \,, 
\end{equation}
uniquely determines all $s_n(\tau_i)$, $n$ even, such that the 
Wronskian normalization condition holds. The stipulation 
$s_n(\tau_i) =0$, $n$ odd, goes hand in hand with the fact 
(seen later on) that $|S_p(\tau)|^2$ admits an asymptotic expansion 
in odd inverse powers of $p$. Comparing with the $|S_p(\tau_i)|^2$ 
series arising from (\ref{asym2.1}) one sees that the odd 
$s_n$ must vanish at $\tau = \tau_i$. The stipulation is also 
consistent with the flat space limit $a(\tau)\equiv 1$. 

The second part of the claim is that the $s_n(\tau_i)$ for $n$ even 
are determined by imposing the Wronskian normalization condition
\begin{eqnarray}
\label{asym2.2c}
\partial_\tau S_p(\tau)S_p(\tau)^\ast -  
S_p(\tau)\partial_\tau S_p(\tau)^\ast 
\overset{\displaystyle{!}}{=}-i\,.
\end{eqnarray}
Using momentarily a `$\prime$' to denote a $\partial_\tau$ derivative 
and setting $s_0:=1$, a formal computation shows (\ref{asym2.2c}) 
to hold subject to (\ref{asym2.2a})  iff 
\begin{eqnarray}
	\label{asym2.2d} 
\sum_{m,\,n\geq 0\,, m+n=N} (s_{2n}s_{2m})(\tau_i)
-\sum_{m\geq 0,\,n\geq 1,\, 2m+n=2N-1}(\omega_2^{-1}s_n'\,s_{2m})(\tau_i) 
\overset{\displaystyle{!}}{=} 0\,, \quad N \geq 1\,.
\end{eqnarray}
To low orders, 
\begin{eqnarray}
\label{asym2.2e} 
&&N=1:\quad 2s_2(\tau_i)-\omega_2(\tau_i)^{-1}s_1'(\tau_i)
\overset{\displaystyle{!}}{=}0
\nonumber \\[1.5mm] 
&&N=2:\quad 2s_4(\tau_i)+s_2(\tau_i)^2-\omega_2(\tau_i)^{-1}s'_3(\tau_i)
-\omega_2^{-1}s'_1(\tau_i)s_2(\tau_i)
\overset{\displaystyle{!}}{=}0\,.
\end{eqnarray}
Clearly, $s_2(\tau_i)$ is determined by the unambiguous $s'_1(\tau_i)$
from (\ref{asym2.2}). In terms of it $s_4(\tau_i)$ is determined 
by the unambiguous $s'_3(\tau_i)$, and so forth. Hence (\ref{asym2.2d}) 
iteratively fixes the integration constants $s_n(\tau_i)$ for $n$ even, 
as claimed. Finally, we note that the recursion (\ref{asym2.2}) 
entails that if (\ref{asym2.2d}) holds at $\tau_i$, then (\ref{asym2.2c}) 
holds formally for all $\tau$.

Assume now that to some order $N$ the $s_1(\tau), \ldots s_N(\tau)$ 
have been computed by the recursion (\ref{asym2.2}) with initial data
(\ref{asym2.2a}), (\ref{asym2.2d}). Then  
\begin{eqnarray}
\label{asym2.3}
S_p^{(N)}(\tau):=\frac{\exp\big\{\!\!-ip\int_{\tau_0}^\tau ds\,\omega_2(s)\big\}}%
{\sqrt{2p\omega_2(\tau)}}\Big\{1+\sum_{n= 1}^{N}(ip)^{-n}s_n(\tau)\Big\}\,,
\end{eqnarray}
is unambigously defined. It enters our work horse Lemma:
\medskip

\begin{lemma} \label{SUVlemma} For some $N\!>\!1$ let $S_p^{(N)}(\tau)$ 
be as in (\ref{asym2.3}). Then, the differential equation 
$[\partial_{\tau}^2 + \omega_p(\tau)^2] 
S_p(\tau) =0$ admits an exact (though implicitly $N$-dependent), 
Wronskian normalized ($\partial_\tau S_p(\tau)S_p(\tau)^\ast -  
S_p(\tau)\partial_\tau S_p(\tau)^\ast=-i)$, complex solution $S_p$, 
such that 
\begin{eqnarray}
\label{asym2.4}
S_p(\tau)&\! =\! & S_p^{(N)}(\tau)\big[1+O(p^{-N})\big]
\nonumber \\[1.5mm] 
\partial_{\tau}S_p(\tau)&\! =\! & \partial_\tau S_p^{(N)}(\tau)\big[1+O(p^{-N})\big]\,,
\end{eqnarray}
{\it uniformly} in $\tau\in[\tau_i,\tau_f] $ as $p\to \infty$. 
\end{lemma}

Here and below the $O$ remainders refer to the supremum of the 
modulus of the function $f \in C[\tau_i,\tau_f]$ estimated, 
i.e.~$f(\tau) = O(p^{-N})$ means $\left\lVert f\right\rVert)_{\sup} = O(p^{-N})$.    
The existence of such estimates for an order dependent function 
in terms of partial sums will below be indicated by the 
``$\,\asymp_N\,$'' relation for the infinite series.  
For example, Lemma \ref{SUVlemma} amounts to the ``$\,\asymp_N\,$'' equality 
of both sides in (\ref{asym2.1}). The asymptotic expansion of 
a fixed ($N$-independent) function will be denoted by ``$\,\asymp\,$''.

\begin{proof} \

To establish the existence and asymptotics 
of the solution $S_p$, we substitute 
\begin{eqnarray}
\label{asym2.5}
S_p(\tau)=S_p^{(N)}(\tau)\cdot R_p(\tau)\,,
\end{eqnarray}
into the differential equation $[\partial_{\tau}^2 + \omega_p(\tau)^2] S_p(\tau) =0$ 
to obtain
\begin{eqnarray}
\label{asym2.6}
&&\partial_\tau^2 R_p+2\frac{\partial_\tau S_p^{(N)}}{S_p^{(N)}}\,
\partial_\tau R_p+F(\tau,p) R_p=0\,,\quad{\rm with}
\nonumber \\[1.5mm] 
&& F_p(\tau):=\frac{\partial_\tau^2 S_p^{(N)}+\omega_p(\tau)^2 S_p^{(N)}}{S_p^{(N)}}\,.
\end{eqnarray}
It is readily verified from the recursion relations (\ref{asym2.2}) 
that $\partial_\tau^2 S_p^{(N)}+\omega_p(\tau)^2 S_p^{(N)}=O(p^{-N-1/2})$, 
while $S_p^{(N)}=O(p^{-1/2})$, uniformly in $\tau\in [\tau_i,\tau_f]$ 
as $p\to \infty$.
This entails
\begin{eqnarray}
\label{asym2.7}
F_p(\tau)=O(p^{-N})\quad 
\text{uniformly in $\tau\in [\tau_i,\tau_f]$ as $p\to \infty$}\,.
\end{eqnarray}
Defining the kernel
\begin{eqnarray}
\label{asym2.8}
K_p(\tau,\tau'):=\int_{\tau'}^\tau S_p^{(N)}(\tau')^2\,
S_p^{(N)}(\tau'')^{-2}d\tau''\,,
\end{eqnarray}
it is easy to see that a function $R_p(\tau)$ satisfying 
the integral equation 
\begin{eqnarray}
\label{asym2.9}
R_p(\tau)=1+r_p-\int_{\tau_i}^\tau K_p(\tau,\tau')F_p(\tau')R_p(\tau')d\tau'\,,
\end{eqnarray}
solves (\ref{asym2.6}). Here $r_p\in \mathbb{R}$ is a constant, satisfying 
$R_p(\tau_i)=1+r_p$, that will be determined later on. Further, 
$K_p=O(1)$ uniformly on $[\tau_i,\tau_f]^2$; so for sufficiently 
large $p$ it follows from (\ref{asym2.7}) that the map
\begin{eqnarray}
\label{asym2.10}
u(\tau)\mapsto 1+r_p-\int_{\tau_i}^\tau K_p(\tau,\tau')F_p(\tau')u(\tau')d\tau'\,,
\end{eqnarray}
is a contraction on the Banach space $\big(C([\tau_i,\tau_f],\mathbb{C}),
\left\lVert \cdot \right\rVert)_{\sup}\big)$; c.f~(\ref{smallp7}). 
Hence (\ref{asym2.9}) has a unique solution 
by the Banach Fixed Point theorem. Moreover, $R_p(\tau)$ is 
differentiable, with 
\begin{eqnarray}
\label{asym2.10a}
\partial_\tau R_p(\tau)&\! =\! & -\int_{\tau_i}^\tau S_p^{(N)}(\tau')^2\,
S_p^{(N)}(\tau)^{-2}F(\tau',p)R_p(\tau')d\tau'\,.
\end{eqnarray}
	
We now determine the constant  $r_p$ by imposing the Wronskian 
condition (\ref{asym2.3}). Since $S_p(\tau)=S_p^{(N)}(\tau) \cdot R_p(\tau)$  
solves $[\partial^2_\tau +\omega_p(\tau)^2]S_p(\tau)=0$, the Wronskian is 
conserved in time. Thus it is sufficient to demand that the 
normalization (\ref{asym2.3}) holds for $\tau=\tau_i$. One has 
\begin{eqnarray}
\label{asym2.10b}
&&\big(\partial_\tau S_p\,S_p^\ast -  S_p\,\partial_\tau S_p^\ast\big)(\tau_i)
\nonumber \\[1.5mm] 
&\! =\! & \big[\partial_\tau S_p^{(N)}\,S_p^{(N)\,\ast } 
-  S_p^{(N)}\,\partial_\tau S_p^{(N)\,\ast }\big](\tau_i)\cdot 
R_p(\tau_i)R_p(\tau_i)^\ast 
\nonumber \\[1.5mm] 
&+& S_p^{(N)}(\tau_i)S_p^{(N)\,\ast }(\tau_i)\cdot
\big[\partial_\tau R_p\,R_p^\ast -  R_p\,\partial_\tau R_p^\ast\big](\tau_i)
\nonumber \\[1.5mm] 
&\! =\! & (1+r_p)^2\,\big[\partial_\tau S_p^{(N)}\,S_p^{(N)\,\ast } 
-  S_p^{(N)}\,\partial_\tau S_p^{(N)\,\ast }\big](\tau_i)\,.
\end{eqnarray}	
The expression $[\partial_\tau S_p^{(N)}\,S_p^{(N)\,\ast } -  
S_p^{(N)}\,\partial_\tau S_p^{(N)\,\ast }](\tau_i)$ may be expanded in 
powers of $p^{-2}$ as before. Although this is a finite sum, in order 
to make contact to the formal Wronskian normalization (\ref{asym2.2c}), 
(\ref{asym2.2d}),  it is convenient to regard the sum as being infinite, 
with the understanding that $s_n\equiv 0$ for $n>N$. With this 
understanding 
\begin{eqnarray}
\label{asym2.10c}
&&\big[\partial_\tau S_p^{(N)}\,S_p^{(N)\,\ast } 
-  S_p^{(N)}\,\partial_\tau S_p^{(N)\,\ast }\big](\tau_i)
\nonumber \\[1.5mm] 
&\! =\! & -i
+i\sum_{k\geq 1}(-)^{k+1}p^{-2k}\Big\{\sum_{m,\,n\geq 0,\, m+n =k}
(s_{2n}s_{2m})(\tau_i)-\sum_{m,\,n\geq 0,\, 2m+n =2k-1}
(\omega_2^{-1}s'_n\,s_{2m})(\tau_i)\Big\}
\nonumber \\[1.5mm] 
&=:& i(-1+\delta_p) 
\end{eqnarray}
Then  
\begin{eqnarray}
\label{asym2.10c1}
\big(\partial_\tau S_p\,S_p^\ast 
-  S_p\,\partial_\tau S_p^\ast\big)(\tau_i)&\! =\! &  
-i +i\big[(1+r_p)^2(-1+\delta_p)+1\big] \,,
\end{eqnarray}	
and the appropriate normalization is thus ensured by choosing $r_p$ 
such that the term in square brackets vanishes. In order to determine 
the large $p$ behavior of $r_p$, that of $\delta_p$ is needed. 
To this end we decompose the sum in (\ref{asym2.10c}) as 
\begin{eqnarray}
\label{asym2.10d}
\delta_p &\! =\! & 
\sum_{k\geq 1}^{\lfloor N/2\rfloor}(-)^{k+1}p^{-2k}
\Big\{\sum_{m,\,n\geq 0,\, m+n =k}(s_{2n}s_{2m})(\tau_i)
-\sum_{m,\,n\geq 0,\, 2m+n =2k-1}(\omega_2^{-1}s'_n\,s_{2m})(\tau_i)\Big\}
\nonumber \\[1.5mm] 
&+&\sum_{k> \lfloor N/2\rfloor}(-)^{k+1}p^{-2k}\Big\{\sum_{m,\,n\geq 0,\, m+n =k}
(s_{2n}s_{2m})(\tau_i)
-\sum_{m,\,n\geq 0,\, 2m+n =2k-1}(\omega_2^{-1}s'_n\,s_{2m})(\tau_i)\Big\}\,,
\nonumber \\[1.5mm] 
\end{eqnarray}
again with the understanding that $s_n \equiv 0, n>N$. The highest index 
of $s_n$ appearing in the first sum is $s_{2\lfloor N/2\rfloor}$, leaving it 
unaffected by setting $s_n\equiv 0$ for $n>N$.  Hence the first sum 
in  (\ref{asym2.10d}) vanishes as before, while the remainder 
contains only a finite number of nonzero terms
\begin{eqnarray}
\label{asym2.10e}
&&\delta_p=\sum_{k> \lfloor N/2\rfloor}(-)^{k+1}p^{-2k}
\Big\{\sum_{m,\,n\geq 0,\, m+n =k}(s_{2n}s_{2m})(\tau_i)
-\sum_{m,\,n\geq 0,\, 2m+n =2k-1}(\omega_2^{-1}s'_n\,s_{2m})(\tau_i)\Big\}\,.
\nonumber \\[1.5mm] 
\end{eqnarray}	
In general this remainder is nonzero, but it manifestly obeys 
$\delta_p=O(p^{-N-1})$. Solving $(1 + r_p)^2( -1 + \delta_p) 
+1 =0$ for $r_p$ and choosing the positive square root one has 
\begin{eqnarray}
\label{asym2.10g}
r_p&\! =\! & -1+ \sqrt{1+\frac{\delta_p}{1-\delta_p}}=O(p^{-N-1})\,,
\end{eqnarray}
on account of $\delta_p=O(p^{-N-1})$.	
	
Having established the normalization (\ref{asym2.2c}) we now proceed 
to showing (\ref{asym2.4}). It follows from (\ref{asym2.7}), 
(\ref{asym2.9}), and (\ref{asym2.10g}) that
\begin{eqnarray}
\label{asym2.11}
R_p(\tau)=1+O(p^{-N})\quad \text{uniformly in $\tau\in [\tau_i,\tau_f]$ 
as $p\to \infty$}\,,
\end{eqnarray}
proving the existence of an exact $S_p(\tau)$ such that 
$S_p(\tau)=S_p^{(N)}(\tau)[1+O(p^{-N})]$. 
On account of the same estimates (\ref{asym2.10a}) entails  
$\partial_\tau R_p(\tau)=O(p^{-N})$, from which it follows that 
\begin{eqnarray}
\label{asym2.12}
\partial_\tau S_p(\tau)&\! =\! & \partial _\tau S_p^{(N)}(\tau)
\bigg[R_p(\tau)+\frac{S_p^{(N)}(\tau)}{ \partial _\tau S_p^{(N)}(\tau)}\, 
\partial _\tau R_p(\tau)\bigg]
\nonumber \\[1.5mm] 
&\! =\! & \partial _\tau S_p^{(N)}(\tau)\big[1+O(p^{-N})\big]\,.
\end{eqnarray}
This completes the proof.

\end{proof}	

\noindent
{\bf Remarks.} 

(i) Using the results of \cite{FRWHadamard1} one can show 
that $s_n, n=1,\ldots N$, coincide with the ones induced by 
the adiabatic iteration for suffiently large order upon 
expansion in $1/p$. The recursion (\ref{asym2.2}) with initial 
data (\ref{asym2.2a}), (\ref{asym2.2d}) in this sense replaces 
the adiabatic iteration.

(ii) A WKB ansatz of the form (\ref{asym2.1}) has been analyzed in 
\cite{Nemes} recently, and was shown to be Borel summable under 
additional assumptions. These assumptions are typically not satisfied 
in massive theories, but may be attainable in massless ones. 
Our Lemma gives a weaker result which however directly applies 
to both situations. 

(iii) The Lemma implies analogous asymptotic expansions for products 
of $S_p(\tau)$'s, both at identical and at distinct times. We prepare 
below the requisite notation for the two-point function (\ref{twoplargep}), 
the modulus square (\ref{modslargep}), and the commutator function 
(\ref{commlargep}).

For the two-point function's Fourier kernel the Lemma implies 
\begin{eqnarray}
	\label{twoplargep} 
S_p(\tau) S_p(\tau')^*  &\asymp_N & \frac{ \exp\big\{ 
\!\!- i p \int_{\tau'}^{\tau} \! ds \, \omega_2(s) \big\} }%
{2 p \sqrt{\omega_2(\tau) \omega_2(\tau')}} \sum_{n \geq 0} V_n(\tau,\tau') (ip)^{-n}  
\nonumber \\[1.5mm]
V_n(\tau,\tau') &\! =\! & \sum_{j=0}^n (-)^{n-j} s_j(\tau) s_{n-j}(\tau') \,, 
\quad n \geq 0\,. 
\end{eqnarray}
To low orders $V_0 =1, V_1(\tau,\tau') = s_1(\tau) - s_1(\tau')$, 
$V_2(\tau,\tau') = s_2(\tau) - s_1(\tau) s_1(\tau') + s_2(\tau')$, 
etc.. Generally, the coefficients obey 
\begin{equation}
	V_{2 j}(\tau,\tau') = V_{2 j}(\tau',\tau) \,, \quad
V_{2 j+1}(\tau,\tau') = - V_{2 j+1}(\tau',\tau) \,, \quad j \geq 0\,.
\end{equation}
They can be evaluated from (\ref{asym2.2}), (\ref{asym2.2a}), 
(\ref{asym2.2d})  recursively to any desired order and are 
increasingly nonlocal; see (\ref{largep10}) for $n=1,2,3$. 

For the modulus square this results in an asymptotic expansion 
in odd inverse powers of $p$, 
\begin{equation}
	\label{modslargep}
|S_p(\tau)|^2 \asymp_N \frac{1}{2\omega_2(\tau)} 
\sum_{n \geq 0} (-)^n V_{2n}(\tau,\tau) \frac{1}{p^{2n+1}}\,.
\end{equation}
When used in (\ref{Spviamod}) this establishes the existence of WKB type 
asymptotic expansions. 

For the commutator function the Lemma implies 
\begin{eqnarray}
	\label{commlargep}
\Delta_p(\tau,\tau') &= & \Lambda_p^+(\tau,\tau') 
\sin\Big(p\! \int_{\tau'}^{\tau} \!\! ds\, \omega_2(s) \Big)+
\Lambda_p^{-}(\tau,\tau') 
\cos\Big(p\!\int_{\tau'}^{\tau}\! \! ds\, \omega_2(s)\Big) \,. 
\nonumber \\[1.5mm]
\Lambda_p^{+} (\tau,\tau') &\asymp_N&  
\frac{1}{\sqrt{\omega_2(\tau) \omega_2(\tau')} }
\sum_{j \geq 0} p^{-2 j -1} (-)^j V_{2j}(\tau,\tau') \,,
\nonumber \\[1.5mm]
\Lambda_p^{-} (\tau,\tau') &\asymp_N&  
\frac{1}{\sqrt{\omega_2(\tau) \omega_2(\tau')} }
\sum_{j \geq 0} p^{-2 j -2} (-)^j V_{2j+1}(\tau,\tau') \,.
\end{eqnarray}

\subsection{Generalized resolvent expansion} \label{sec4.2}

As highlighted in (\ref{Spviamod}), a Wronskian normalized solution 
of the basic wave 
equation is fully determined by its modulus square. By 
(\ref{modslargep}) we know the form of the modulus square's 
asymptotic expansion. 
The coefficients $V_{2n}(\tau,\tau)$ are in principle determined 
by the basic recursion (\ref{asym2.2}). Since at each order an 
additional integration enters, one would expect these coefficients 
to be highly nonlocal in time. Remarkably, this is not the case:
the $V_{2n}(\tau,\tau)$ turn out to be local differential polynomials 
in the frequency functions $\omega_0(\tau)^2, \omega_2(\tau)^2$ of the 
differential operator $\partial_{\tau}^2 + \omega_0(\tau)^2 + p^2 \omega_2(\tau)^2$.

The main ingredient in the derivation is the Gelfand-Dickey equation.
Using only the basic differential equation and the Wronskian 
normalization (\ref{FLode1}) one finds $|S_p(\tau)|^2$ to satisfy the 
(nonlinear form of the) Gelfand-Dickey equation  
\begin{equation}
	\label{GD1}
2 |S_p|^2 \partial_{\tau}^2 |S_p|^2 - \big( \partial_{\tau} |S_p|^2 \big)^2 
+ 4 \omega_p^2 |S_p|^4 =1\,.
\end{equation}
In view of the expected relation to (\ref{asym2.1}) it is convenient to 
set 
\begin{equation}
	\label{GD2}
|S_p(\tau)|^2 =: i G_{ip}(\tau)\,. 
\end{equation}
Then 
\begin{eqnarray}
	\label{GD3}  
&& 2 G_z \partial_{\tau}^2 G_z - (\partial_{\tau} G_z)^2 + 
4[\omega_0^2 - z^2 \omega_2^2] G_z^2 =-1\,,
\nonumber \\[1.5mm]
&& \partial_{\tau}^3 G_z + 4 [\omega_0^2 - z^2 \omega_2^2] \partial_{\tau}G_z + 
2 \partial_{\tau}[\omega_0^2 - z^2 \omega_2^2] G_z =0\,.
\end{eqnarray}
Here the second, linear version 
of the Gelfand-Dickey equation follows by differentiating 
the nonlinear form. For $\omega_2^2 =1$ and $\omega_0^2 = v$ the 
same equations govern the diagonal of the resolvent 
kernel of the differential operator $\partial_{\tau}^2 + v$, with 
$z^2 = - p^2$ playing the role of the resolvent parameter
\cite{Dickeybook}. The  diagonal of the resolvent kernel is known 
to admit an asymptotic expansion in inverse powers of $z$, whose 
coefficients coincide with the heat kernel coefficients on 
general grounds, see e.g.~\cite{Avramidibook}. The 
generalization to $[\partial_{\tau}^2 + \omega_0(\tau)^2 ] S = 
z^2 \omega_2(\tau)^2 S$, with non-constant $\omega_2(\tau)^2$ 
can be treated as follows.   

Inserting the ansatz 
\begin{eqnarray}
	\label{GD4} 
G_z(\tau) = \sum_{n \geq 0} \frac{G_{n}(\tau)}{2 \omega_2} z^{-2 n-1} \,,
\quad G_0 =1\,,
\end{eqnarray}
into the nonlinear Gelfand-Dickey equation results in the recursion
\begin{eqnarray}
	\label{GD6} 
G_n = \!\!\sum_{k,l\geq 0, k+l =n-1} \!\Big\{ 
\frac{1}{4} \frac{G_k}{\omega_2} 
\partial_{\tau}^2 \Big( \frac{G_l}{\omega_2} \Big) - \frac{1}{8} 
\partial_{\tau} \Big( \frac{G_k}{\omega_2} \Big) 
\partial_{\tau} \Big( \frac{G_l}{\omega_2} \Big) + 
\frac{1}{2} \frac{\omega_0^2}{\omega_2^2} G_k G_l \Big\} 
-\frac{1}{2} \sum_{k,l\geq 1, k+l =n} G_k G_l. \nonumber\\ 
\end{eqnarray} 
This expresses $G_n$ in terms of $G_{n-1}, \ldots, G_1$, 
and involves only differentiations. It follows that all $G_n$ 
are differential polynomials in $v := \omega_0^2,w:= \omega_2^2$. 
Denoting $\partial_{\tau}$ differentiations momentarily by a ``$\,'\,$'' one 
finds:
\begin{eqnarray}
	\label{GD7} 
G_1 &\! =\! & \frac{v}{2 w} + \frac{5}{32} \frac{{w'}^2}{ w^3} 
- \frac{1}{8} \frac{w''}{w^2}\,,
\nonumber \\[1.5mm]
G_2 &\! =\! & \frac{3}{8w^2} \Big( v^2 + \frac{1}{3} v'' \Big) 
- \frac{5}{16 w^3} \Big( v w'' + v'w' - v\frac{7 {w'}^2}{4w} \Big)     
\nonumber \\[1.5mm]
&+&\frac{1}{32 w^3} \Big(\! - w^{(4)} +\frac{21 {w''}^2}{4 w}
+\frac{7 w^{(3)} w'}{w} -\frac{231 {w'}^2w''}{8 w^2} 
+\frac{1155 {w'}^4}{64 w^3} \Big)\,.
\end{eqnarray}   
The recursion (\ref{GD6}) is easily programmed 
in {\tt Mathematica} and produces the $G_n$ to reasonably high orders. 
The $G_n$ can be seen as generalized heat kernel coefficients.
For $\omega_2 = 1$, $v=\omega_0^2$ plays the role of the potential 
and (\ref{GD7}) reproduces the well-known expressions \cite{Avramidibook} (up to 
overall normalizations). In the massless case $v = \omega_0^2 =0$, 
and only the purely $w$ dependent parts of the $G_n$ remain.
From the viewpoint of the initial expansion (\ref{asym2.1}), 
(\ref{asym2.2}) the concise differential polynomials (\ref{GD7}) 
are surprising: $G_n = V_{2n}(\tau,\tau)$ must hold by 
construction, but would seem to suggest highly nonlocal 
coefficients. At low orders one can see the cancellation of 
the nonlocal terms directly. For example, the $n=2$ recursion 
(\ref{asym2.2}) integrates to $s_2 = s_1^2/2 + \partial_{\tau} s_1/(2 \omega_2)$. 
Hence $G_1 = 2 s_2 - s_1^2 = \partial_{\tau} s_1/\omega_2$, which is 
indeed local. 

One can also relate the $G_n$'s more directly to the standard 
heat kernel coefficients. To this end, we transform the 
basic differential equation (\ref{FLode1}) into conformal 
time as in (\ref{FLode3}), but for generic frequency functions:
$\partial_{\eta} = \omega_2(\tau)^{-1} \partial_{\tau}$, $\chi_p(\tau) =
\omega_2(\tau)^{1/2} S_p(\tau)|_{\tau = \tau(\eta)}$. This replaces the 
differential operator $\partial_{\tau}^2 + \omega_0(\tau)^2 + p^2 \omega_2(\tau)^2$
by $\partial_{\eta}^2 + 2 E_1(\eta) + p^2$, with $E_1(\eta) = G_1(\tau(\eta))$,
the image of $G_1$ in (\ref{GD7}). The coefficient of $p^2$ is 
now unity and $2 E_1(\eta)$ plays the role of the potential. 
Inserting the $\eta$-version of the ansatz (\ref{GD4}) into 
the linear Gelfand-Dickey equation results in the one-step 
differential recursion  
\begin{equation}
	\label{GD5}
\partial_{\eta} E_{n+1} = \partial_{\eta} E_1 E_n + 2 E_1 \partial_{\eta} E_n 
+ \frac{1}{4} \partial_{\eta}^3 E_n\,,\quad n \geq 1\,.
\end{equation} 
This defines (up to a conventional normalization) the standard 
heat kernel coefficients with potential $2 E_1$. Undoing the 
transformation one has 
\begin{equation}
	\label{GD8} 
G_n = E_n\big|_{E_1 \mapsto G_1, \partial_{\eta} \mapsto 
\omega_2(\tau)^{-1} \partial_{\tau}}\,.
\end{equation}
For example, for $n=2$ this gives  
\begin{equation} 
\label{GD9} 
G_2 = \frac{3}{2} G_1^2 + 
\frac{1}{4 \omega_2} \partial_{\tau} \Big( \frac{\partial_{\tau} G_1}{\omega_2} \Big)\,,
\end{equation}
which is indeed satisfied by (\ref{GD7}). Generally, the agreement of 
(\ref{GD6}) with (\ref{GD8}) provides a welcome check. 
\medskip 

An analogous interplay exists
for the asymptotics of the phase as induced by the basic 
expansion (\ref{asym2.1}) and the resolvent expansion (\ref{GD4}),
respectively.  
Starting from the basic expansion (\ref{asym2.1}) the phase 
is determined by $\tan (\arg S_p(\tau)) = \Im S_p(\tau)/ \Re S_p(\tau)$.  
One finds 
\begin{eqnarray}
	\label{phlargep1}
&& \tan \big(\arg S_p(\tau)\big) = - 
\frac{ {\cal S}_p^-(\tau) {\bf C}_p  + {\cal S}_p^+(\tau) {\bf S}_p}%
{ {\cal S}_p^+(\tau) {\bf C}_p  - {\cal S}_p^-(\tau) {\bf S}_p}\,,
\nonumber \\[1.5mm]
&& {\bf S}_p =\sin\Big(p\! \int_{\tau'}^{\tau} \!\! ds\, \omega_2(s) \Big)\,,
\quad 
{\bf C}_p = \cos\Big(p\!\int_{\tau'}^{\tau}\! \! ds\, \omega_2(s)\Big) \,.
\end{eqnarray} 
with 
\begin{eqnarray}
	\label{phlargep2} 
&& {\cal S}_p^+(\tau) \asymp_N \frac{1}{\sqrt{2 p \omega_2(\tau)}} 
\sum_{j \geq 0} (-)^j s_{2 j}(\tau) p^{-2j}\,, 
\nonumber \\[1.5mm]
&& {\cal S}_p^-(\tau) \asymp_N \frac{1}{\sqrt{2 p \omega_2(\tau)}} 
\sum_{j \geq 0} (-)^j s_{2 j+1}(\tau) p^{-2j-1}\,,
\end{eqnarray}
To low orders 
\begin{eqnarray}
	\label{phlargep3} 
&& \tan \big(\!\arg S_p(\tau)\big) \asymp_N - \frac{{\bf S}_p}{{\bf C}_p}  
- \frac{1}{p} s_1(\tau) \frac{1}{ {\bf C}_p^2} 
- \frac{1}{p^2} s_1(\tau)^2 \frac{{\bf S}_p}{ {\bf C}_p^3} 
- \frac{1}{p^3} s_1(\tau)^3 \frac{{\bf S}_p^2}{ {\bf C}_p^4} 
\nonumber \\[1.5mm]
&& - \frac{1}{p^3} (s_1 s_2 - s_3)(\tau) \frac{1}{ {\bf C}_p^2} 
+ O\Big(\frac{1}{p^4} \Big)\,.
\end{eqnarray}
Writing $s_1 s_2 -s_3 = s_1^3/3 -u_3$, the ratios of trigonometric 
functions are just the derivatives of the $\tan$ function; 
so (\ref{phlargep3}) is equivalent to 
\begin{eqnarray}
	\label{phlargep4}
\arg S_p(\tau) &\asymp_N& - p\int_{\tau_0}^{\tau} \! ds \, \omega_2(s) - 
\frac{s_1(\tau)}{p} + \frac{u_3(\tau)}{p^3} 
+ O\Big(\frac{1}{p^5} \Big)\,, 
\nonumber \\[1.5mm]
u_3(\tau) &=& \frac{\partial_{\tau} G_1}{4 \omega_2} + \frac{1}{2}\int_{\tau_0}^{\tau}
\! ds\, \omega_2(s) G_1(s)^2\,,\quad G_1(\tau) = \frac{\partial_{\tau} s_1}{\omega_2}\,, 
\end{eqnarray}
where the explicit form of $u_3$ follows from the recursion 
(\ref{asym2.2}). Proceeding along these lines, it is not immediate that 
at higher orders no oscillatory terms will occur in the phase itself 
and that the coefficients will be single integrals of local quantities.  

This is, however, the case and can be seen from the alternative 
realization of the phase entailed by (\ref{Spviamod}) and (\ref{GD2})
\begin{equation}
	\label{phlargep6} 
\arg S_p(\tau) = - \frac{1}{2} \int_{\tau_0}^{\tau} \! ds\, 
\frac{1}{ i G_{ip}(s)}\,.
\end{equation} 
Here the expansion (\ref{GD4}) can be used. It follows that 
$\arg S_p(\tau)$ admits an asymptotic expansion in odd inverse 
powers of $p$ whose coefficients are {\it single} integrals 
of polynomials in the $G_n$. To low orders 
\begin{equation}
	\label{phlargep7} 
\frac{1}{ i G_{ip}(\tau)} \asymp_N 2 \omega_2(\tau) p \Big\{ 1 + 
\frac{G_1(\tau)}{p^2} + \frac{(G_1^2 - G_2)(\tau)}{p^4} 
+ O \Big( \frac{1}{p^6} \Big) \Big\}\,.
\end{equation}
The equivalence to (\ref{phlargep4}) is ensured by (\ref{GD9}). 

\subsection{Induced asymptotic expansion of SLE.} \label{sec4.3}

Using the formulas from Theorem \ref{sslethm} and (\ref{commlargep}) 
all SLE related quantities have induced asymptotic expansions in 
inverse powers of $p$ at some finite order $N>1$. The order 
can be increased arbitrarily, but in general the exact 
reference solution in Lemma \ref{SUVlemma} needs to be changed 
in order to do so. Here we show that the (unique, $N$-independent) 
SLE solution is asymptotic $\asymp_N$ to the previously constructed 
series for {\it all} $N$. In particular, the asymptotic expansion 
is independent of the window function $f$.

\begin{theorem} \label{sleUVthm} The modulus-square of the SLE solution 
admits an asymptotic expansion in odd inverse powers of $p$, whose 
coefficients are independent of the window function $f$ and 
are given by generalized heat kernel coefficients. Specifically 
\begin{equation}
	\label{SLElargep}
|T_p^{\rm SLE}(\tau)|^2 \asymp \frac{1}{2 p \omega_2(\tau)}
\left\{ 1 + \sum_{n\geq 1} \frac{(-)^n}{p^{2n}} G_n(\tau)\right\}\,,
\end{equation}
where the $G_n$ are determined recursively by (\ref{GD6}). The 
phase has an asymptotic expansion obtained from 
\begin{equation}
	\label{phSLElargep} 
\arg T_p^{\rm SLE}(\tau) \asymp - p\! \int_{\tau_0}^{\tau} \! ds \,\omega_2(s)  
\left\{ 1 + \sum_{n\geq 1} \frac{(-)^n}{p^{2n}} G_n(\tau)\right\}^{-1}\,.
\end{equation}
The massless limits are regular and have 
coefficients $G_n|_{\omega_0^2 =0}$. 
\end{theorem} 

\begin{proof}\

 We mostly need to show that $|T_p^{\rm SLE}(\tau)|^2$ 
admits an asymptotic expansion of the form (\ref{SLElargep}) 
with some coefficients $\tilde{G}_n(\tau)$. Since the SLE solution 
is a Wronskian normalized solution of the basic wave equation, its 
modulus square solves the nonlinear Gelfand-Dickey equation 
(\ref{GD1}). The coefficients $\tilde{G}_n(\tau)$  therefore also 
have to obey the recursion (\ref{GD6}). It then suffices to 
check by direct computation that $\tilde{G}_0 =1$. The latter will 
be done separately following the proof. Since $\tilde{G}_0 =1$ 
determines all other coefficients, it follows that $\tilde{G}_n = G_n$, 
for all $n \in \mathbb{N}$. The relation (\ref{phSLElargep}) between phase 
and modulus holds on account of the Wronskian normalization. 

In order to show that $|T_p^{\rm SLE}(\tau)|^2$ has an asymptotic 
expansion in odd inverse powers of $p$, we use the realization 
as $J_p(\tau)/(2 {\cal E}_p^{\rm SLE})$ from (\ref{ssle12}). The 
integrands of $J_p(\tau)$ and $({\cal E}_p^{\rm SLE})^2$ are built from 
$\Delta_p(\tau,\tau_0)$, $\partial_{\tau} \Delta_p(\tau,\tau_0)$, 
$\partial_{\tau} \partial_{\tau_0} \Delta_p(\tau,\tau_0)$. For these 
we prepare  
\begin{eqnarray}
	\label{largep7}
\Delta_p(\tau,\tau')&\! =\! & \Lambda_p^{+}(\tau,\tau') {\bf S}_p +
\Lambda_p^{-}(\tau,\tau') {\bf C}_p\,, 
\nonumber \\[1.5mm]
\partial_{\tau} \Delta_p(\tau,\tau') &\! =\! &  
\cap_p^{+}(\tau,\tau') {\bf S}_p +
\cap_p^{-}(\tau,\tau') {\bf C}_p \,,
\nonumber \\[1.5mm]
\partial_{\tau} \partial_{\tau'}\Delta_p(\tau,\tau') &\! =\! &  
\sqcap_p^{+}(\tau,\tau') {\bf S}_p +
\sqcap_p^{-}(\tau,\tau') {\bf C}_p \,,
\end{eqnarray} 
with ${\bf S}_p$, ${\bf C}_p$ as defined in \eqref{phlargep1}, and
\begin{eqnarray}
	\label{largep8}
\cap_p^{\pm}(\tau,\tau') &\! =\! & \partial_{\tau} \Lambda^{\pm }(\tau,\tau') 
\mp p \omega_2(\tau) \Lambda^{\mp }(\tau,\tau') \,,
\nonumber \\[1.5mm]
\sqcap_p^{\pm}(\tau,\tau') &\! =\! & \partial_{\tau'}\partial_{\tau} 
\Lambda^{\pm}_p(\tau,\tau') \pm p [ \partial_{\tau} \Lambda^{\mp}(\tau,\tau') 
\omega_2(\tau') - \partial_{\tau'} \Lambda^{\mp}(\tau,\tau') \omega_2(\tau) ] 
\nonumber \\[1.5mm]
&+& p^2 \omega_2(\tau) \omega_2(\tau') \Lambda^{\pm}(\tau,\tau')\,.
\end{eqnarray}
Note that $\Lambda^{\pm}_p(\tau,\tau') = \pm \Lambda^{\pm}_p(\tau',\tau)$,
$\sqcap^{\pm}_p(\tau,\tau') = \pm \sqcap^{\pm}_p(\tau',\tau)$,
while $\cap_p^{\pm}(\tau,\tau')$ has no manifest symmetry. 
The normalization of the commutator function implies, however, 
$\cap_p^-(\tau,\tau) =1$. The definitions in combination 
with (\ref{commlargep}) imply that $\Lambda_p^+, \cap_p^+, \sqcap_p^+$ 
have an asymptotic $\asymp_N$ expansion in odd inverse powers of $p$,
while $\Lambda_p^-, \cap_p^-, \sqcap_p^-$ have an asymptotic $\asymp_N$ 
expansion in even inverse powers of $p$. Crucially, while the fiducial solutions $S_{N}$ provided by Lemma \ref{SUVlemma} are implicitly $N$-dependent, Theorem \ref{th2.1} ensures that the induced expansion of  $|T_p^{\rm SLE}(\tau)|^2$ is independent thereof. Schematically, $|T_p^{\rm SLE}[S_{N}]|^2$ is the same for all $N$, which allows one to take $N$ arbitrarily large.
 
Next we use (\ref{largep7}) to evaluate the integrands of $J_p(\tau')$ 
from (\ref{ssle3}) and $({\cal E}_p^{\rm SLE})^2$ from (\ref{ssle11}). In a first 
step we merely insert (\ref{largep7}) and replace all powers of 
oscillatory terms by linear ones using 
\begin{equation}
	\label{largep11} 
{\bf S}_p^2 = \frac{1}{2}(1 - {\bf C}_{2 p})\,, \quad 
{\bf C}_p^2 = \frac{1}{2}(1 + {\bf C}_{2 p})\,, \quad 
{\bf S}_p {\bf C}_p = \frac{1}{2} {\bf S}_{2p} \,.
\end{equation}
This gives
\begin{eqnarray}
	\label{largep12} 
&& (\partial_{\tau} \Delta_p(\tau,\tau'))^2 + \omega_p(\tau)^2 \Delta_p(\tau,\tau')^2 
\nonumber \\[1.5mm]
&& \quad = 
\frac{1}{2}\Big[ \cap_p^+(\tau,\tau')^2 + \cap_p^-(\tau,\tau')^2 + 
\omega_p(\tau)^2 \big(\Lambda_p^+(\tau,\tau')^2 
+ \Lambda_p^-(\tau,\tau')^2 \big) \Big] 
\nonumber \\[1.5mm] 
&& \quad -
\frac{1}{2}\Big[ \cap_p^+(\tau,\tau')^2 - \cap_p^-(\tau,\tau')^2 + 
\omega_p(\tau)^2 \big(\Lambda_p^+(\tau,\tau')^2 
- \Lambda_p^-(\tau,\tau')^2 \big) \Big] {\bf C}_{2p} 
\nonumber \\[1.5mm]
&& \quad 
+ \Big[ (\cap_p^+ \cap_p^-)(\tau,\tau')  
+ \omega_p(\tau)^2 (\Lambda_p^+ \Lambda_p^-)(\tau,\tau')\Big] 
{\bf S}_{2p} \,.
\end{eqnarray}
The integrand of $({\cal E}_p^{\rm SLE})^2$ is of course symmetrized in 
$\tau,\tau'$; for brevity's sake we use the non-symmetric version   
\begin{eqnarray}
	\label{largep13} 
&& \big(\partial_{\tau} \partial_{\tau'}\Delta_p(\tau,\tau')\big)^2 
+ 2\omega_p(\tau')^2 \big(\partial_{\tau} \Delta_p(\tau,\tau')\big)^2 
+ \omega_p(\tau)^2 \omega_p(\tau')^2 \Delta_p(\tau,\tau')^2  
\nonumber \\[1.5mm]
&& \quad =
\frac{1}{2}\Big[ \sqcap_p^+(\tau,\tau')^2 + \sqcap_p^-(\tau,\tau')^2 
+ 2\omega_p(\tau')^2 \big(\cap_p^+(\tau,\tau')^2 + 
\cap_p^-(\tau,\tau')^2 \big) 
\nonumber \\[1.5mm]
&& \quad +
\omega_p(\tau)^2 \omega_p(\tau')^2 \big(\Lambda_p^+(\tau,\tau')^2 + 
\Lambda_p^-(\tau,\tau')^2 \big)  
\Big] 
\nonumber \\[1.5mm]
&& \quad -
\frac{1}{2}\Big[ \sqcap_p^+(\tau,\tau')^2 - \sqcap_p^-(\tau,\tau')^2 
+ 2\omega_p(\tau')^2 \big(\cap_p^+(\tau,\tau')^2 - 
\cap_p^-(\tau,\tau')^2 \big) 
\nonumber \\[1.5mm]
&& \quad +
\omega_p(\tau)^2 \omega_p(\tau')^2 \big(\Lambda_p^+(\tau,\tau')^2 - 
\Lambda_p^-(\tau,\tau')^2 \big)  
\Big]{\bf C}_{2p} 
\nonumber \\[1.5mm]
&& \quad +  
\Big[ (\sqcap_p^+ \sqcap^-)(\tau,\tau') + 
2 \omega_p(\tau')^2 (\cap_p^+ \cap_p^-)(\tau,\tau') 
\nonumber \\[1.5mm]
&& \quad +
\omega_p(\tau)^2 \omega_p(\tau')^2 (\Lambda_p^+\Lambda_p^-)(\tau,\tau') \Big] 
{\bf S}_{2p} \,.
\end{eqnarray}
The coefficients of the oscillatory terms have asymptotic 
expansions in inverse powers of $p$ which are uniform the 
both variables. Focussing on the integration variable we write 
$A_p(\tau)$ for such a coefficient. For smooth $\omega_0,\omega_2$ 
also $A_p$ will be smooth in $\tau$. By repeated use of the 
integrations-by-parts identities 
\begin{eqnarray}
	\label{largep14}
&& {\bf S}_{2p} = - \frac{1}{2p \omega_2(\tau)} \partial_{\tau} {\bf C}_{2p} \,, 
\quad 
{\bf C}_{2p} = \frac{1}{2p \omega_2(\tau)} \partial_{\tau} {\bf S}_{2p} \,, 
\nonumber \\[1.5mm]
&& \int\! d\tau f(\tau)^2 A_p(\tau) {\bf S}_{2p} = \frac{1}{2p} 
\int\! d\tau \partial_{\tau} \bigg( \frac{f(\tau)^2 A_p(\tau)}{\omega_2(\tau)} \bigg) 
{\bf C}_{2p}\,, 
\nonumber \\[1.5mm]
&& \int\! d\tau f(\tau)^2 A_p(\tau) {\bf C}_{2p} = -\frac{1}{2p} 
\int\! d\tau \partial_{\tau} \bigg( \frac{f(\tau)^2 A_p(\tau)}{\omega_2(\tau)} \bigg) 
{\bf S}_{2p}\,, 
\end{eqnarray}
the oscillatory terms can therefore be made 
subleading at any desired order of the asymptotic expansion.  

It follows that at any order the asymptotic expansion of 
$J_p(\tau')$ and $({\cal E}_p^{\rm SLE})^2$ is generated by the 
non-oscillatory terms in (\ref{largep12}), (\ref{largep13}).  
By inspection of the orders induced by (\ref{commlargep}) and 
(\ref{largep8})  one sees that the non-oscillatory term in 
(\ref{largep12}) has an expansion in even inverse powers of 
$p$, starting with a $O(p^0)$ term. Similarly $p^{-2}$ times 
the non-oscillatory term in (\ref{largep13}) has an expansion 
in even inverse powers of $p$, starting with a $O(p^0)$ term.  
 Hence $J_p(\tau')$ has an asymptotic expansion in even 
inverse powers of $p$, starting with an $O(p^0)$ term. 
The square root of the non-oscillatory term in (\ref{largep13}) 
governs the expansion of $p^{-1} {\cal E}_p^{\rm SLE}$, which therefore 
likewise has an asymptotic expansion in even inverse powers 
of $p$, starting with a $O(p^0)$ term. Together, 
$J_p(\tau)/(2 {\cal E}_p^{\rm SLE})$ admits a asymptotic expansion 
in odd inverse powers of $p$, as claimed. 
Augmented by the explicit computation of the leading order, this 
implies the result.   
\end{proof} 

\noindent
{\bf Remarks.} 

(i) The exponent in $\exp\{ i \arg T_p^{\rm SLE}(\tau)\}$ can be 
re-expanded in powers of $1/p$ to obtain a simplified expansion 
of the form (\ref{asym2.1}). Theorem \ref{sleUVthm} implies 
that $T_p^{\rm SLE}(\tau)$ has the property described in Lemma 
\ref{SUVlemma} for {\it any $N>1$}. This replaces Olbermann's 
Lemma 4.5, where the adiabatic vacua of order $N$ play a role 
analogous to our approximants $S_p^{(N)}(\tau)$ (though not necessarily 
with matched orders). The adiabatic 
vacua are however far less explicit: first, the adiabatic 
iteration produces more complicated formulas of which only 
the large $p$ expansion is actually used. Second, the 
iterates are only well-defined for sufficiently large $p$, so 
for technical reasons they need to be extended 
in an ad-hoc manner to small momenta \cite{FRWHadamard1}.
Third, the result then enters an integral equation 
whose iteration produces the required exact solution, dubbed  
adiabatic vacuum of order $N$. The Lemma \ref{SUVlemma} short 
cuts these three steps. The ansatz (\ref{asym2.1}) only processes the 
information relevant for large $p$ and the iteration (\ref{asym2.2}) 
is manifestly well-defined without modifications. In combination 
with (\ref{GD4}), (\ref{GD6}) this yields a practically usable expansion.  

(ii) The simplified expansion from (i) for the product $T_p^{\rm SLE}(\tau) 
T_p^{\rm SLE}(\tau')^*$ can be viewed as the Fourier space 
version of the (state independent) Hadamard parametrix. 
The Hadamard parametrix also has a truncated version where 
only the solution of the recursion to some finite order is kept,
see e.g.~\cite{Moretti}. These truncations converge in a 
certain sense to the Hadamard parametrix proper, which in turn 
is a distributional solution of the wave equation in both 
arguments modulo a smooth piece. The fact that the inverse 
Fourier transform of the state independent WKB expansion has 
the form of the Hadamard parametrix was verified (in $d=3$ and 
in conformal time) by an instructive if formal computation 
in \cite{Pirk}. In Olbermann's proof of the Hadamard property 
this step is rigorously supplied by appealing to a general 
result of Junker and Schrohe \cite{JunkerS}, describing the wave 
front set of adiabatic vacua of order $N$. Since  our 
approximants have the same large $p$ asymptotics as the 
adiabatic vacua (though not necessarily with matched orders) 
this step carries over. It may be worthwhile to attempt a 
direct, simplified proof, specific for SLE and including 
the massless case. 

(iii) Assuming that the massless case can be treated along 
these lines the SLE would provide very relevant 
examples of {\it infrared finite Hadamard} states. Their relevance 
stems from the following {\it Proposal:} The primordial vacuum-like 
state (of a massless free QFT and the perturbation theory based on it) 
should be chosen to be an infrared finite Hadamard state 
and conceptually be associated with a pre-inflationary
period of non-accelerated expansion. The rationale for this 
proposal is detailed in Section \ref{sec5}.

\medskip
\noindent
{\bf Direct verification of Theorem \ref{sleUVthm} to subleading order.}
The proof of Theorem \ref{sleUVthm} hinges on the direct verification of 
the leading order asymptotics. Here we present an ab-initio 
evaluation of the $|T^{\rm SLE}_p(\tau|^2$ asymptotics to subleading 
order, starting from Eq.~(\ref{ssle12}) and the asymptotics 
(\ref{commlargep}) of the commutator function. 
We prepare to subleading order
\begin{eqnarray}
	\label{largep9}
\Lambda^{+}_p(\tau,\tau') &\! =\! & \frac{1}{p} \tilde{V}_0(\tau,\tau') - 
\frac{1}{p^3} \tilde{V_2}(\tau,\tau')  + O\Big(\frac{1}{p^5}\Big) \,, 
\nonumber \\[1.5mm]
\Lambda^{-}_p(\tau,\tau') &\! =\! & \frac{1}{p^2} \tilde{V}_1(\tau,\tau') 
- \frac{1}{p^4} \tilde{V}_3(\tau,\tau') + O\Big(\frac{1}{p^6}\Big) \,, 
\nonumber \\[1.5mm]
\cap^{+}_p(\tau,\tau') &\! =\! & \frac{1}{p} \big[\partial_{\tau}\tilde{V}_0(\tau,\tau') - 
\omega_2(\tau) \tilde{V}_1(\tau,\tau')\big] -
\frac{1}{p^3}\big[ \partial_{\tau} \tilde{V_2}(\tau,\tau')  
- \omega_2(\tau) \tilde{V}_3(\tau,\tau') \big] 
+ O\Big(\frac{1}{p^5}\Big) \,, 
\nonumber \\[1.5mm]
\cap^{-}_p(\tau,\tau') &\! =\! & \omega_2(\tau) \tilde{V}_0(\tau,\tau') 
+ \frac{1}{p^2}\big[\partial_{\tau}\tilde{V}_1(\tau,\tau') - \omega_2(\tau) 
\tilde{V}_2(\tau,\tau') \big]  + O\Big(\frac{1}{p^4}\Big)\,,
\nonumber \\[1.5mm]
\sqcap^{+}_p(\tau,\tau') &\! =\! & p \,\omega_2(\tau) \omega_2(\tau')\tilde{V}_0(\tau,\tau') 
+ \frac{1}{p} \big[ \partial_{\tau}\partial_{\tau'} \tilde{V}_0(\tau,\tau') 
+ \partial_{\tau} \tilde{V}_1(\tau,\tau') \omega_2(\tau') 
- \partial_{\tau'} \tilde{V}_1(\tau,\tau') \omega_2(\tau) 
\nonumber \\[1.5mm]
&-& \omega_2(\tau) \omega_2(\tau') 
\tilde{V}_2(\tau,\tau') \big] 
 + O\Big(\frac{1}{p^3}\Big) \,, 
\nonumber \\[1.5mm]
\sqcap^{-}_p(\tau,\tau') &\! =\! & -\partial_{\tau} \tilde{V}_0(\tau,\tau') \omega_2(\tau') 
+\partial_{\tau'} \tilde{V}_0(\tau,\tau') \omega_2(\tau) + 
\omega_2(\tau) \omega_2(\tau') \tilde{V}_1(\tau,\tau')  
\nonumber \\[1.5mm]
&+& \frac{1}{p^2}\big[\partial_{\tau}\partial_{\tau'} \tilde{V}_1(\tau,\tau') + 
\partial_{\tau} \tilde{V}_2(\tau,\tau') \omega_2(\tau') - 
\partial_{\tau'} \tilde{V}_2(\tau,\tau') \omega_2(\tau) 
\nonumber \\[1.5mm]
&-& \omega_2(\tau) \omega_2(\tau') \tilde{V}_3(\tau,\tau')\big] 
+ O\Big(\frac{1}{p^4}\Big) \,, 
\end{eqnarray}
with 
\begin{eqnarray}
	\label{largep10} 
\tilde{V}_n(\tau,\tau') &:=& 
\frac{V_n(\tau,\tau')}{ \sqrt{\omega_2(\tau) \omega_2(\tau')}} \,, \quad 
V_0 =1\,, 
\nonumber \\[1.5mm]
V_1(\tau,\tau') &=& s_1(\tau) - s_1(\tau') \,, \quad 
V_2(\tau,\tau') = \frac{1}{2} V_1(\tau,\tau')^2 + 
\frac{1}{2}[G_1(\tau) + G_1(\tau')]\,,
\nonumber \\[1.5mm]
V_3(\tau,\tau') &\! =\! & \frac{1}{6} V_1(\tau,\tau')^3 + V_1(\tau,\tau') 
\big[ G_1(\tau) + G_1(\tau') \big] 
\nonumber \\[1.5mm]
&+& \frac{\partial_{\tau} G_1(\tau)}{ 2 \omega_2(\tau)} - 
\frac{\partial_{\tau'} G_1(\tau')}{ 2 \omega_2(\tau')} - 
2\! \int_{\tau'}^{\tau} \! ds \,\omega_2(s) G_1(s)^2\,.
\end{eqnarray}
As described in the proof, it suffices to focus on the non-oscillatory 
in (\ref{largep12}), (\ref{largep13}). Keeping up to subleading terms 
in (\ref{largep12}) one finds 
\begin{eqnarray}
	\label{largep15}
&\!\!\!\!\!\!\!\!\!\! & 
(\partial_{\tau} \Delta_p(\tau,\tau'))^2 + \omega_p(\tau)^2 \Delta_p(\tau,\tau')^2 
\nonumber \\[1.5mm]
&\!\!\!\!\!\!\!\!\!\! & \quad \asymp \omega_2(\tau)^2 \tilde{V}_0(\tau,\tau')^2 
+ \frac{1}{p^2} \Big\{ \frac{1}{2} (\partial_{\tau}\tilde{V}_0(\tau,\tau'))^2 
+ \frac{1}{2} \omega_0(\tau)^2 \tilde{V}_0(\tau,\tau')^2
\nonumber \\[1.5mm]
&\!\!\!\!\!\!\!\!\!\! & \quad 
+ \omega_2(\tau) \big( \tilde{V}_0\partial_{\tau} 
\tilde{V}_1 - \partial_{\tau} \tilde{V}_0 
\tilde{V}_1\big)(\tau,\tau') 
+ \omega_2(\tau)^2 \big( \tilde{V}_1^2 - 
2 \tilde{V}_0\tilde{V}_2 \big)(\tau,\tau') \Big\} + 
O\Big( \frac{1}{p^4} \Big).
\end{eqnarray} 
Upon integration this gives 
\begin{eqnarray}
	\label{largep16}
J_p(\tau') \asymp \frac{\bar{\omega}_2}{2 \omega_2(\tau')}\bigg\{1 + 
\frac{1}{p^2}\Big[ - G_1(\tau') + \frac{1}{2 \bar{\omega}_2} 
\int\! d\tau f(\tau)^2 \Big( \frac{\omega_0^2}{\omega_2} + 
\frac{1}{4} \frac{(\partial_{\tau} \omega_2)^2}{\omega_2^2} \Big) \Big] + 
O\Big(\frac{1}{p^4} \Big) \bigg\}\,. 
\end{eqnarray}
Here we used 
\begin{eqnarray}
	\label{largep17} 
&& \tilde{V}_1^2 - 2\tilde{V}_0 \tilde{V}_2 = 
- \frac{G_1(\tau) + G_1(\tau')}{\omega_2(\tau) \omega_2(\tau')} \,,
\nonumber \\[1.5mm]
&& \frac{1}{2} (\partial_{\tau} \tilde{V}_0)^2 + 
\frac{1}{2} \omega_0(\tau)^2 \tilde{V}_0^2 = 
\frac{1}{2 \omega_2(\tau) \omega_2(\tau')}
\bigg( \omega_0^2 + \frac{1}{4}
\frac{(\partial_{\tau} \omega_2)^2}{\omega_2^2} \bigg)\,,
\nonumber \\[1.5mm]
&& 2 \omega_2^2 G_1 = \omega_0^2 + \frac{1}{4}
\frac{(\partial_{\tau} \omega_2)^2}{\omega_2^2} - \frac{1}{2} \partial_{\tau} \Big( 
\frac{\partial_{\tau} \omega_2}{\omega_2} \Big)\,.
\end{eqnarray}
Similarly, keeping up to subleading terms in (\ref{largep13}) one 
has  
\begin{eqnarray}
	\label{largep18} 
&& \big(\partial_{\tau} \partial_{\tau'}\Delta_p(\tau,\tau')\big)^2 
+ 2 \omega_2(\tau')^2 \big(\partial_{\tau} \Delta_p(\tau',\tau)\big)^2 
+ \omega_p(\tau)^2 \omega_p(\tau')^2 \Delta_p(\tau,\tau')^2  
\nonumber \\[1.5mm]
&& \quad \asymp p^2 \, 2 \omega_2(\tau)^2 \omega_2(\tau')^2 \tilde{V}_0^2 
+ \Big(\frac{1}{2} \omega_0(\tau)^2 \omega_2(\tau')^2 + \frac{3}{2} 
\omega_0(\tau')^2 \omega_2(\tau)^2 \Big) \tilde{V}_0^2 
\nonumber \\[1.5mm]
&& \quad + \omega_2(\tau) \omega_2(\tau') \tilde{V}_0 \partial_{\tau} \partial_{\tau'} 
\tilde{V}_0 + \omega_2(\tau')^2 (\partial_{\tau} \tilde{V}_0)^2  
+ \frac{1}{2} \big[ \partial_{\tau} \tilde{V}_0 \omega_2(\tau') - 
\partial_{\tau'} \tilde{V}_0 \omega_2(\tau) \big]^2 
\nonumber \\[1.5mm]
&& + \quad 
3 \omega_2(\tau) \omega_2(\tau')^2 \big( \tilde{V}_0 \partial_{\tau} \tilde{V}_1 
- \tilde{V}_1 \partial_{\tau} \tilde{V}_0 \big) 
- \omega_2(\tau)^2 \omega_2(\tau') \big( \tilde{V}_0 \partial_{\tau'} \tilde{V}_1 
- \tilde{V}_1 \partial_{\tau'} \tilde{V}_0 \big) 
\nonumber \\[1.5mm]
&& + \quad 
2 \omega_2(\tau)^2 \omega_2(\tau')^2 \big( \tilde{V}_1^2 - 2\tilde{V}_0 
\tilde{V}_2 \big) + O \Big( \frac{1}{p^2} \Big) .
\end{eqnarray}
For the simplification we use (\ref{largep17}) as well as
\begin{equation}
	\label{largep19} 
\tilde{V}_0 \partial_{\tau} \tilde{V}_1 - \tilde{V}_1 \partial_{\tau} \tilde{V}_0 = 
\frac{G_1(\tau)}{\omega_2(\tau')}\,, 
\quad 
\tilde{V}_0 \partial_{\tau'} \tilde{V}_1 - \tilde{V}_1 \partial_{\tau'} \tilde{V}_0 = 
- \frac{G_1(\tau')}{\omega_2(\tau)}\,, 
\end{equation} 
For the $O(p^0)$ term in (\ref{largep18}) this results in 
\begin{eqnarray}
	\label{largep20}
&& \omega_2(\tau) \omega_2(\tau') \big( G_1(\tau) - G_1(\tau') \big) 
+ \frac{1}{2} \frac{\omega_0(\tau)^2}{\omega_2(\tau)} \omega_2(\tau') 
+ \frac{3}{2} \frac{\omega_0(\tau')^2}{\omega_2(\tau')} \omega_2(\tau) 
\nonumber \\[1.5mm]
&& \quad + \frac{3}{8} \frac{(\partial_{\tau} \omega_2)^2}{ \omega_2(\tau)^3} 
\omega_2(\tau') + \frac{1}{8} \frac{(\partial_{\tau'} \omega_2)^2}{ \omega_2(\tau')^3} 
\omega_2(\tau)\,.
\end{eqnarray}  
Finally, 
\begin{eqnarray}
	\label{largep21} 
({\cal E}_p^{\rm SLE})^2 = \frac{p^2}{4} \bar{\omega}_2^2 + 
\frac{\bar{\omega}_2}{4} \int\! d\tau \, f(\tau)^2 
\bigg( \frac{\omega_0^2}{ \omega_2} + \frac{1}{4}
\frac{(\partial_{\tau} \omega_2)^2}{ \omega_2^3} \bigg) + O\Big(\frac{1}{p^4} \Big) \,.
\end{eqnarray}
This results in 
\begin{equation}
	\label{largep22} 
|T_p^{\rm SLE}(\tau)|^2 \asymp \frac{1}{2 p \omega_2(\tau)}
\left\{ 1 - \frac{1}{p^2} G_1(\tau)  + O \Big( \frac{1}{p^4} \Big)
\right\}\,.
\end{equation}
The leading term confirms $\tilde{G}_0=1$ in the proof of 
Theorem \ref{sleUVthm}. The subleading term verifies the assertion at this order 
by an ab-initio computation. 

As seen in before, the relation (\ref{phSLElargep}) 
between phase and modulus holds on account of the Wronskian 
normalization. However, it is not immediate how the expression 
(\ref{ssle12}) for $\tan (\arg T_p^{\rm SLE}(\tau))$ reproduces 
this simple answer. As a final check on the framework we 
verified the equivalence to subleading order by direct computation. 
Omitting the details, the result is 
 
\begin{equation}
	\label{largep28} 
\tan\big( \arg T_p^{\rm SLE}(\tau) \big) = 
- \frac{{\cal E}_p^{\rm SLE} \Delta_p(\tau,\tau_0)}{J_p(\tau,\tau_0) } 
\asymp 
- \frac{{\bf S}_p}{{\bf C}_p} - \frac{s_1(\tau)}{p} \frac{1}{ {\bf C}_p^2} 
+ O\Big( \frac{1}{p^3} \Big)\,.
\end{equation}
This agrees with \eqref{phlargep3} and hence \eqref{phlargep6}, \eqref{phlargep7} to the order considered.

\newpage 
\section{SLE as pre-inflationary vacua}\label{sec5}

One of the key empirical facts about the Cosmological Microwave Background 
(CMB) is its near scale invariance at large values of the multipole
expansion. This feature, realized at $t=t_{\rm decoupl}$, is thought to be rooted in a similar 
behavior of the primordial power spectrum $P_{\zeta}(t_*,p)$ at the 
(cosmological) time $t_* \ll t_{\rm decoupl}$ when the seeds for 
structure formation are laid, for any of the relevant fluctuation 
variables $\zeta$. In terms of the spatial Fourier 
momentum a behavior $P_{\zeta}(t_*,p) \sim |p|^{-2\nu}$ is needed, 
with $\nu$ close to $d/2$. Such a behavior is seemingly incompatible 
with the momentum dependence of the massless SLE modes. We show here 
that a qualitatively correct power spectrum arises at $t=t_*$, 
if a pre-inflationary period is followed by one of near-exponential 
expansion. 

It must be stressed that general relativity {\it demands} a period 
of non-accelerated expansion following the Big Bang, 
i.e.~for some interval $t \in (t_{\rm sing}, t_1]$. In particular, 
variants of the cosmological singularity theorems remain valid 
for generic inflationary spacetimes with positive cosmological 
constant \cite{GRsingthms}. For FL spacetimes a pre-inflationary phase 
with kinetic energy domination is preferred \cite{FRWasym1,FRWasym2}. 
As a consequence, 
the time-honored purely positive frequency Bunch-Davies vacuum, 
traditionally postulated at the beginning of the inflationary 
period cannot be physically realistic: the modes from 
the pre-inflationary period (whether themselves positive 
frequency or not close to the singularity) will generically {\it not} 
be positive frequency at $t_1$. As a consequence the modes 
at $t = t_1$ can also not comply with deSitter invariance.
This is because an admixture of positive and negative frequency 
modes compatible with deSitter invariance (known as $\alpha$ vacua) 
fails to define a Hadamard state. Perturbation theory in 
an $\alpha$ vacuum suffers from incurable UV divergences already 
at one loop order. One is thus led to search for Hadamard states 
on an FL background in the interval $(t_{\rm sing}, t_1]$ with implicitly 
defined bonus properties that lead to a qualitatively correct 
power spectrum at $t=t_*$. We propose massless SLE states as 
viable candidates.  
 

\subsection{Asymptotics of massless modes versus power spectrum}\label{sec5.1}

We return to the basic wave equation in conformal time (\ref{FLode3}) 
and specialize to the massless case and $d=3$
\begin{equation}
	\label{FLode4} 
\Big[ \partial_{\eta}^2 + p^2 - \frac{\partial_{\eta}^2 a}{a}\Big] \chi_p(\eta)=0\,,
\quad 
\partial_{\eta} \chi_p \chi_p^* - (\partial_{\eta} \chi_p)^* \chi_p =-i\,.
\end{equation} 
The wave equation (\ref{FLode4}) bears a two-fold relation  to 
lowest order cosmological perturbation theory, see 
e.g.~\cite{WeinbergCosmbook}, Chapter 10: (a) it coincides 
precisely with the  wave equation satisfied by the tensor 
perturbations, with $\chi_p$ playing the role of either of the 
coefficient functions $h_+(\eta,p)$ or $h_{\times}(\eta,p)$  
in the polarization decomposition $h_{ij}(\eta,x) = 
h_+(\eta,x) e^+_{ij} + h_{\times}(\eta,x) e^{\times}_{ij}$, 
and $ds^2 = a(\eta)^2 [ -d \eta^2 + (\delta_{ij} + h_{ij} ) dx^i dx^j]$. 
(b) With the replacement of $a$ by $z$, the Mukhanov-Sasaki variable,
it coincides with wave equation satisfied by the scalar (curvature) 
perturbations, where $\chi_p$ is often denoted by $v_p(\eta) = 
z(\eta) {\cal R}_p(\eta)$. 

The equation (\ref{FLode4}) can be solved for small $p$ and 
large $p$ as detailed in Sections \ref{sec3.1} and \ref{sec4.1}, respectively. 
For small $p$ one has a convergent power series expansion
$\chi_p(\eta) = \sum_{n \geq 0} \chi_n(\eta) p^{2n}$,  
which corresponds to the massless case of (\ref{Srec1}).
Since $\tau = \int^{\eta} \! ds \, a(s)^{-2}$ and $S_p(\eta) = 
\chi_p(\eta)/a(\eta)$, the leading order $S_0(\tau)$ from before 
(\ref{massless1}) reads 
\begin{equation}
	\label{chiasym1} 
\chi_0(\eta) = a(\eta) \Big[z_0 +  w_0 \int_{\eta_0}^{\eta} \frac{ds}{a(s)^2} 
\Big] \,, \quad w_0 z_0^* - w_0^* z_0 =-i\,.
\end{equation}
The higher orders then are determined recursively by transcribing 
(\ref{Srec3}). Heuristically, the leading order can be expected
to be a good approximation if $p \ll \partial_{\eta} a/a$, 
$p^2 \ll \partial_{\eta}(\partial_{\eta} a/a)$, so that $2 p^2 \ll \partial_{\eta}^2a /a$. 
In other words, the wavelength $1/p$ of the mode needs to be 
uniformly much larger than the comoving Hubble distance 
$a/\partial_{\eta} a$. Under these conditions $\int^{\eta} \! ds \,a(s)^{-2} 
\propto 1/(p a(\eta)^2)$ (with a small constant of proportionality) 
is selfconsistent and shows that the second term in $\chi_0$ 
will be decreasing in $\eta$, while the first term is increasing. 
With the replacement of $a(\eta)$ by $z(\eta)$ 
the same applies to the scalar perturbations. It must be stressed 
that the low momentum behavior (\ref{chiasym1}) is not generic;
there are relevant solutions with a different behavior, as 
highlighted by the SLE solution (\ref{chiasym3}) below.

In order to transcribe the WKB ansatz (\ref{asym2.1}) 
we note $\int_{\tau_i}^{\tau} d\tau' a(\tau')^2 = \eta- \eta_i$ and 
$\partial_{\tau} = a(\eta)^2\partial_{\eta}$, for $d=3$. Specializing also 
(\ref{asym2.2}) to $\omega_0(\eta) =0$, $\omega_2(\eta)= a(\eta)^2$
the WKB solution for (\ref{FLode4}) reads 
\begin{eqnarray}
	\label{chiasym2} 
\chi_p(\eta) &\asymp_N& \frac{e^{- i p (\eta - \eta_i)}}{\sqrt{2 p}} 
\Big\{ 1 + \sum_{n \geq 1} (i p)^{-n} s_n(\eta) \Big\}\,,
\nonumber \\[1.5mm]
\partial_{\eta} s_n &\! =\! & \partial_{\eta} s_1 s_{n-1} + \frac{1}{2} \partial_{\eta}^2 s_{n-1}\,,
\quad 
\partial_{\eta} s_1 = - \frac{1}{2} \frac{\partial_{\eta}^2 a}{a}\,.
\end{eqnarray}
For the modulus square this gives $2 p |\chi_p(\eta)|^2 \asymp_N 
1 + p^{-2} \partial_{\eta}^2a /(2a) + O(p^{-2})$, see (\ref{chiasym4}). 
Heuristically, the WBK approximation  is expected to be good 
in the regime opposite to (\ref{chiasym1}), 
i.e.~whenever the wave length $1/p$ of the mode is 
uniformly much smaller than the comoving Hubble distance 
$a/\partial_{\eta} a$, entailing $\partial_{\eta}^2a /a \ll 2p^2$. 
Again, simply replacing $a(\eta)$ by $z(\eta)$ gives the 
corresponding result for the scalar perturbations. 

The quantity of interest is the power spectrum at the time 
of seed formation $\eta_*$. Per tensor mode it is defined by  
\begin{equation}
	\label{Pdef} 
P_{\chi}(p) := \lim_{\eta \rightarrow \eta_*} \frac{p^3}{2 \pi^2} 
\frac{|\chi_p(\eta)|^2}{a(\eta)^2}\,,
\end{equation}
and similarly with $z$ replacing $a$ for the scalar perturbations. 
The time $\eta_*$ is often identified with the Hubble crossing time 
$\eta_p$, defined by $(\partial_{\eta} a/a)(\eta_p) = p$. This 
lies in the cross-over region of the $(\eta,p)$ plane not directly 
accessible via the small or large momentum expansions. 
A nearly scale invariant power spectrum is one where 
$P_{\chi}(p) \propto p^{-2 \epsilon}$ for a small positive coefficient
$\epsilon>0$. As indicated, the power spectrum also depends on 
the choice of solution $\chi_p$. The principles of QFT 
in curved spacetime require its large momentum behavior to 
be constrained by the Hadamard property. A necessary but by no 
means sufficient condition for a solution to be Hadamard it that 
it approaches a positive frequency wave for $p \rightarrow \infty$.  
The low momentum behavior is somewhat constrained along the lines 
discussed at the end of Section \ref{sec3}. In the present context, an additional
constraint arises from the requirement that $p^3 
|\chi_p(\eta)|^2/a(\eta)^2$ is approximately scale invariant 
in the cross-over region of the $(\eta,p)$ plane.

The SLE have been shown to meet the first two criteria. 
Here we explore the satisfiability of the last requirement. 
We first note the low and high momentum behavior by 
appealing to the results from Sections \ref{sec3.2} and \ref{sec4.2}. 
For the low momentum expansion the formulas 
(\ref{massless1}), (\ref{massless2}), (\ref{massless3})
require as input the directly transcribed massless 
commutator function $\Delta_0(\eta,\eta') = 
\int_{\eta'}^{\eta} \! ds\,a(s)^{-2}$. It solves  
$(a(\eta)^2 \partial_{\eta})^2 \Delta_0(\eta,\eta') =0$, where the 
field redefinition is not yet taken into account. (The latter 
generates an effective mass term and the computation would have to 
proceed differently). This leads to 
\begin{eqnarray}
	\label{chiasym3}
\frac{|\chi^{\rm SLE}_p(\eta)|^2}{a(\eta)^2} 
= \frac{\bar{a}}{2 p} + O(p) \,, \quad 
\quad 
\bar{a} = \frac{ \int\! d\eta \,a(\eta)^{-4} f^{\rm conf}(\eta)^2}%
{ \int\! d\eta \,f^{\rm conf}(\eta)^2}\,,
\end{eqnarray}
and similarly for $z$ replacing $a$. For large momentum 
the modulus square has the generic WKB asymptotics
\begin{eqnarray}
	\label{chiasym4} 
\frac{|\chi^{\rm SLE}_p(\eta)|^2}{a(\eta)^2} \asymp \frac{1}{2 p a(\eta)^2}
\Big\{ 1 + \frac{1}{2 p^2} \frac{\partial_{\eta}^2 a}{a} + 
O\Big( \frac{1}{p^4} \Big) \Big\}\,, 
\end{eqnarray}
and similarly for $z$ replacing $a$. As usual, the cross-over region 
needed for the power spectrum is not directly accessible via 
these expansions.  


\subsection{A model with pre-inflationary SLE}\label{sec5.2} 

To proceed, we consider an analytically soluble model, adopted 
from  \cite{InflIRmatch0}, where the seed formation time $\eta_*$ 
is $p$-independent and coincides with the end of a deSitter period. 
The deSitter period is preceded by one with kinetic energy domination.    
Computations of the power spectrum where a positive frequency 
solution in a pre-inflationary era is matched to a solution 
corresponding to accelerated expansion have been considered in 
\cite{InflIRmatch0,InflIRmatch1,InflIRmatch3,InflIRmatch4,InflIRmatch5}.

Following \cite{InflIRmatch0}, we use conformal time $\eta$ and 
consider an instantaneous transition between a kinetic dominated 
pre-inflationary period and de Sitter expansion. The scale factor reads
\begin{eqnarray}
\label{pinf01}
a(\eta)=
\begin{cases}
\sqrt{1+2H\eta}\,,& \eta\in (-\frac{1}{2H},0)\,,\\[2mm]
\frac{1}{1-H\eta}\,,& \eta \in [0,\frac{1}{H})\,,
\end{cases}
\end{eqnarray}
with the transition occurring at $\eta_1=0$, and $H$ denoting 
the (physical) Hubble parameter during inflation. The time 
of seed formation is $\eta_* = 1/H$ and the price to pay for the 
analytic solubility is the formal pole in the line element.

As seen in Section \ref{sec2} the modulus square of an SLE solution is 
strictly independent of the choice of fiducial solution. We are 
thus free to choose a convenient one, $S_p(\eta) = \chi_p(\eta)/a(\eta)$, 
in the process of evaluating $|\chi_p^{\rm SLE}(\eta)/a(\eta)|^2$ for  
a given window function $f \in C_c^{\infty}(- 1/(2H), 1/H)$.  
A useful choice adhering to the traditional Bunch-Davies 
solution during the deSitter period is 
\begin{eqnarray}
\label{pinf04}
S_p(\eta)&\! =\! & 
\begin{cases}
\alpha_p S^{\rm kin}_p(\eta)+\beta_p S^{\rm kin}_p(\eta)^\ast\,,& 
-\frac{1}{2H}<\eta\leq 0 \,, \\[2mm] 
S^{\rm BD}_p(\eta)\,, &0\leq  \eta<\frac{1}{H}\,,
\end{cases}
\end{eqnarray}
where 
\begin{eqnarray}
\label{pinf04a1}
S^{\rm kin}_p(\eta)&:= &\sqrt{\frac{\pi}{8H}} \,
H_0^{(2)}\bigg(p\eta+\frac{p}{2H}\bigg)\,,
\nonumber \\[1.5mm] 
S^{\rm BD}_p(\eta) &:=& \frac{e^{-ip(\eta-\frac{1}{H})}}{\sqrt{2p}}
(1 - H \eta) 
\Big(1+\frac{iH}{p}\frac{1}{1-H\eta}\Big)\,,
\end{eqnarray}
are solutions of (\ref{FLode4}) in their respective regimes.
The matching coefficients $\alpha_p\,,\beta_p$ are determined by 
demanding continuity of $S_p$ and $\partial_{\eta}S_p$ at the transition,
\begin{eqnarray}
\label{pinf04a}
\alpha_p&\! =\! & e^{ip/H} \sqrt{\frac{\pi p}{16 H}} 
\bigg[ H_0^{(1)}\Big(\frac{p}{2H}\Big)
-\Big(\frac{H}{p} -i \Big) H_1^{(1)}\Big(\frac{p}{2H}\Big)\bigg]\,,
\nonumber \\[1.5mm] 
\beta_p&\! =\! &  e^{ip/H} \sqrt{\frac{\pi p}{16 H}}
\bigg[- H_0^{(2)}\Big(\frac{p}{2H}\Big)
+\Big(\frac{H}{p} -i \Big) H_1^{(2)}\Big(\frac{p}{2H}\Big)\bigg]\,,
\end{eqnarray}
with $|\alpha_p|^2-|\beta_p|^2 = 1$ from the Wronskian condition.

This fiducial solution enters the SLE parameters 
$c_1,c_2$ and $\lambda_p,\mu_p$ from Section \ref{sec2.1}. The advantage of 
the choice (\ref{pinf04}) is that it leads to a relatively 
simple expression for the power spectrum in terms of the 
(numerically computed) SLE parameters $c_1$ and $c_2$. 
The SLE solution will however {\it not} be of the Bunch-Davies 
type during the deSitter period,
\begin{equation}
	\frac{\chi_p^{\rm SLE}(\eta)}{a(\eta)} =\lambda_p S^{\rm BD}(\eta)+\mu_p 
S^{\rm BD}(\eta)^\ast\,.
\end{equation}
For $\eta_* = 1/H$ the SLE's power spectrum (\ref{Pdef}) is 
given by 
\begin{eqnarray}
\label{pinf04b}
P_{\chi^{\rm SLE}}(p)&\! =\! & \frac{H^2}{(2\pi)^2}|\lambda_p-\mu_p|^2=
\frac{H^2}{(2\pi)^2} \frac{c_1+\Re c_2}{\sqrt{c_1^2-|c_2|^2}}\,.
\end{eqnarray} 
Here
\begin{eqnarray}
\label{pinf05}
c_1&\! =\! & \frac{1}{2} \int\! d\eta\, f(\eta)^2 a(\eta)^2 
\Big\{|\partial_{\eta}S_p(\eta)|^2+p^2 |S_p(\eta)|^2\Big\}\,,
\nonumber \\[1.5mm] 
c_2&\! =\! & \frac{1}{2} \int\!d\eta \,f(\eta)^2 a(\eta)^2 
\Big\{(\partial_{\eta} S_p(\eta))^2+p^2 S_p(\eta)^2\Big\}\,,
\end{eqnarray}
are determined by (\ref{pinf04}). With some slight caveats 
it follows from the earlier results that the right hand side is 
indeed a Bogoliubov invariant: by (\ref{ssle15}) this holds for 
$\sqrt{c_1^2 - |c_2|^2}$ and since $\lim_{\eta_0 \rightarrow 1/H} S_p^{\rm BD}(\eta_0) = 
iH/ \sqrt{2 p^3}$ one can interpret the first line of 
(\ref{ssle14}) as $\lim_{\eta_0 \rightarrow 1/H} J_p(\eta_0) = (H^2/p^3) 
(c_1 + \Re c_2)$. Further, the relation (\ref{chiasym3}) 
immediately suggests the low momentum asymptotics, while 
(\ref{chiasym4}) in combination with $\lim_{\eta \rightarrow 1/H} a(\eta)^{-2} =0$,
$\lim_{\eta \rightarrow 1/H} a(\eta)^{-3} \partial_{\eta}^2 a = 2 H^2$,
suggests $\lim_{\eta \rightarrow 1/H} |\chi_p^{\rm SLE}(\eta)/a(\eta)|^2 = 
H^2/(2 p^3) + O(p^{-5})$ for large $p$. The caveats are: that 
$\eta = 1/H$ lies at the boundary of the interval $[0, 1/H)$, 
that the line element (\ref{pinf01}) has a pole there, and that the 
window function may not have support in the deSitter phase only.      
We therefore present a more careful analysis of the small and large momentum behavior of $P_{\chi^{SLE}}(p)$, allowing for a generic window function 
with support in both the kinetic dominated and the deSitter period, thereby demonstrating that  the above conclusions are indeed valid.  

\begin{proposition}\label{pr5.1} Let $f\in C_c^\infty(-\frac{1}{2H},\,\frac{1}{H})$ 
be a window function for \eqref{pinf05}. Then
\begin{eqnarray*}
&(a)&\quad P_{\chi^{\rm SLE}}(p)=\frac{H^2}{(2\pi)^2} +O(p^{-2})\,\,\,
{\rm as} \,\,\,p\to \infty\,.
\nonumber \\[1.5mm] 
&(b)& \quad P_{\chi^{\rm SLE}}(p)=p^2 \frac{\bar{a}}{(2\pi)^2}+O(p^4)\,\,\,
{\rm as} \,\,\,p\to 0\,.
\end{eqnarray*}
\end{proposition}

\begin{proof}\

 (a) The large $p$ asymptotics are conveniently analyzed 
in terms of \eqref{pinf04b}, where the $\lambda_p,\,\mu_p$ coefficients 
refer to \eqref{pinf04}, \eqref{pinf04a1}, \eqref{pinf04a} as the 
fiducial solution for the SLE construction.
	
As the window function $f$ is allowed to have support both 
in the kinetic dominated and de Sitter periods, it is convenient 
to split the integrations in (\ref{pinf05})
\begin{eqnarray}\label{pinf010}
c_1=c_1^{<}+c_1^{>}\quad{\rm and}\quad c_2=c_2^{<}+c_2^{>}\,,
\end{eqnarray} 
with the $<(>)$ denoting the contribution from the kinetic 
dominated (de Sitter) regime. This takes into account the 
distinct forms of our fiducial solution (\ref{pinf04}) in the 
respective regimes. We may readily read off
\begin{eqnarray}
\label{pinf011}
c_1^> &\! =\! & \frac{p}{2} \int_{0}^{\frac{1}{H}}\! d\eta\,
f(\eta)^2 a(\eta)^2 (1-H\eta)^2
+\frac{1}{2}\frac{H^2}{2p} \int_{0}^{\frac{1}{H}}\! d\eta\,
f(\eta)^2 a(\eta)^2\,,
\nonumber \\[1.5mm] 
c_2^{>}&\! =\! &   \frac{1}{2} \int_{0}^{\frac{1}{H}} \!d\eta\,f(\eta)^2 
a(\eta)^2e^{-2ip(\eta-\frac{1}{H})}\bigg[iH-iH^2 \eta -\frac{H^2}{2p}\bigg]\,.
\end{eqnarray}
For the analysis of the $c_1^<,\,c_2^<$ terms, it will prove helpful to 
define  
\begin{eqnarray}
\label{pinf012}
\bar{c}_1&:=& \frac{1}{2} \int_{-\frac{1}{2H}}^0 \! d\eta \,f(\eta)^2 
a(\eta)^2 \Big\{|\partial_{\eta}S^{\rm kin}_p(\eta)|^2+p^2 |S^{\rm kin}_p(\eta)|^2\Big\}\,,
\nonumber \\[1.5mm] 
\bar{c}_2&:=& \frac{1}{2} \int_{-\frac{1}{2H}}^0 \! d \eta \,f(\eta)^2 a(\eta)^2 
\Big\{(\partial_{\eta} S^{\rm kin}_p(\eta))^2+p^2 S^{\rm kin}_p(\eta)^2\Big\}\,,
\end{eqnarray}
in terms of which we may express
\begin{eqnarray}\label{pinf014}
c_1^< &\! =\! & (|\alpha_p |^2 +|\beta_p|^2) \bar{c}_1+
2 \Re \big[ \alpha_p \beta_p^\ast \bar{c}_2\big]\,,
\nonumber \\[1.5mm] 
c_2^< &\! =\! & \alpha_p^2 \bar{c}_2+\beta_p^2 \bar{c}_2^{\,\,\ast }
+2\alpha_p \beta_p \bar{c}_1\,.
\end{eqnarray}
The $e^{-2ip(\eta-\frac{1}{H})}$ 
term in  the integrand of $c_2^>$  entails that $c_2^> \sim O(p^{-n})$ for 
any $n\in \mathbb{N}$ as $p\to \infty$. Hence, $c_2^>$ is negligible compared 
to $c_1^>$, for large enough $p$. Next, in order to understand the asymptotic 
behavior of (\ref{pinf012}), (\ref{pinf014}), it is sufficient to 
consider the leading asymptotic behavior of $S^{\rm kin}_p(\eta)$ 
as $p\to \infty$,
\begin{eqnarray}
\label{pinf015}
S^{\rm kin}_p(\eta) \sim \frac{1}{\sqrt{4Hp}}\bigg( \eta
+\frac{1}{2H}\bigg)^{-\frac{1}{2}}
e^{-ip(\eta+\frac{1}{2H})+i \pi/4}\big( 1+O(p^{-1})\big)\,,
\end{eqnarray}
leading to 
\begin{eqnarray}
\label{pinf016}
\bar{c}_1 &\sim&  \frac{p}{2} \int_{-\frac{1}{2H}}^0 \! d\eta\,f(\eta)^2\,,
\\[2mm] 
\bar{c}_2 &\sim & \frac{1}{2} \int_{-\frac{1}{2H}}^0 \!d\eta\,
f(\eta)^2a(\eta)^2 e^{-2ip(\eta+\frac{1}{2H})} \bigg[\frac{H}{(1+2H\eta)^2}
-\frac{iH}{2p(1+2H\eta)^3}\bigg]\big( 1+O(p^{-1})\big)\,.
\nonumber 
\end{eqnarray} 
As before, the presence of the $e^{-2ip(\eta+\frac{1}{2H})}$ entails that 
$\bar{c}_2 \ll \bar{c}_1$ as $p\to \infty$. 

Using (\ref{pinf014}) to express $|c_2|^2/c_1^2$ in terms of 
$c_1^>,\,c_2^>,\,\bar{c}_1,\,\bar{c}_2$, we may disregard relative 
contributions of $c_2^>\,, \bar{c}_2$ to $|c_2|^2/c_1^2$, and find
\begin{eqnarray}\label{pinf017}
	|c_2|^2&\! =\! & 4|\beta_p|^2 (1+|\beta_p|^2)\bar{c}_1^2\,,
	\nonumber \\[1.5mm] 
	c_1^2 &\! =\! & (\bar{c}_1+c_1^>)^2 + 4 |\beta_p|^2\big[(1+|\beta_p|^2)\bar{c}_1^2 + |\beta_p|^2 \bar{c}_1c_1^>\big]\,,
\end{eqnarray}
where we have used the fact that $|\alpha_p|^2-|\beta_p|^2=1$ to write $\alpha_p$ in terms of $\beta_p$. Examining (\ref{pinf011}) and (\ref{pinf016}), it is clear that $\bar{c}_1$ and $c_1^>$ have the same leading large $p$ behavior, and from (\ref{pinf04a}) it follows that $|\beta_p|^2 \sim \frac{9 H^4}{16 p^4}+O(p^{-6})$. Thus we use  (\ref{pinf017}) to estimate 
\begin{eqnarray}
\frac{|c_2|^2}{c_1^2} \sim O(p^{-4})\,.
\end{eqnarray}
Since 
\begin{eqnarray}
\label{pinf09a}
\mu_p&\! =\! & \frac{1}{\sqrt{2}}\sqrt{\frac{1}{1-\frac{|c_2|^2}{c_1^2}}-1}\,, 
\quad |\lambda_p|= \sqrt{1+\mu_p^2} \,,
\end{eqnarray}
this establishes part (a) of Proposition \ref{pr5.1}.

(b) The main obstruction to using \eqref{chiasym3} to infer the 
result is that the limit $\eta\to 1/H$ of the small $p$ SLE expansion 
is not a-priori well-defined. We remove this obstruction by a small 
modification of Proposition \ref{Ssmallp}.

In both the kinetic dominated and deSitter regimes, the mode equation reads
$S''_p(\eta)+2\frac{a'}{a}S'_p(\eta)+p^2S_p(\eta)=0$. Consistent with 
(\ref{chiasym1}) we choose the following solution for the $p=0$ equation
\begin{eqnarray}
S_0(\eta)=
\begin{cases}
\frac{1}{\sqrt{2}}\Big[\frac{\ln(1+2H\eta)}{2H}-\frac{1}{3H}+i\Big]\,,
& \text{kinetic domination}\,,
\\[3mm]
\frac{1}{\sqrt{2}}\Big[ -\frac{1}{3H}(1-H\eta)^3+i\Big]\,,
& {\rm deSitter}\,.
\end{cases}
\end{eqnarray}
Both cases satisfy $[\partial_\eta S_0\,S_0^\ast -S_0\,
\partial_\eta S_0^\ast](\eta)=-ia(\eta)^{-2}$, as well as 
\begin{eqnarray}
\lim_{\eta\to 1/H}S_0(\eta)=
\frac{i}{\sqrt{2}}\quad{\rm and}\quad \lim_{\eta\to 1/H}\partial_\eta S_0(\eta)=0\,.
\end{eqnarray} 
This shows that $S_0$ extends uniquely to a continuous function 
on $(-1/(2H),\,1/H]$.
	
Choosing some $0<\eta_i<1/H$ such that ${\rm supp}\,f\subset[\eta_i,\,1/H]$, 
it clear that a solution of the integral equation
\begin{eqnarray}
S(\eta)&\! =\! & S_0(\eta)-p^2\int_{\eta_i}^{1/H}\!K(\eta,\eta')S(\eta')d\eta'\,,
\nonumber \\[1.5mm] 
K(\eta,\eta')&\! =\! & i \theta(\eta-\eta')S_0(\eta)S_0(\eta')^\ast+
\theta(\eta'-\eta)S_0(\eta)^\ast S_0(\eta')\,,
\end{eqnarray}
is a solution of the mode equation on $(\eta_i,1/H)$. Since $S_0$ extends 
to a $C^1$ function on the closed interval $[\eta_i,\,1/H]$, the proof of 
Proposition \ref{Ssmallp} carries over on the Banach space 
$\big(C([\eta_i,1/H],\mathbb{C}^2),\,\left\lVert \cdot \right\rVert)_{\sup}\big)$.
	
Hence we have a convergent series $S_p(\eta)=\sum_{n=0}^\infty 
p^{2n}S_n(\eta)$, which we take as the fiducial solution for the SLE 
in the small $p$ regime. This then has a well-defined limit as 
$\eta\to 1/H$, namely
\begin{eqnarray}
\lim_{\eta \to 1/H}T^{\rm SLE}_p(\eta)&\! =\! & \lambda_p
\lim_{\eta \to 1/H}S_p(\eta)+\mu_p \lim_{\eta \to 1/H}S_p(\eta)^\ast 
\nonumber \\[1.5mm] 
&\! =\! & \lambda_pS_p(1/H)+\mu_pS_p(1/H)^\ast \,.
\end{eqnarray}
Both $p^\frac{1}{2}\lambda_p$ and $p^\frac{1}{2}\mu_p$ admit convergent 
power series expansions as in \eqref{massless8}, leading to 
\begin{eqnarray}
|T_p^{\rm SLE}(1/H)|=\frac{\bar{a}}{2p}+O(p)\,,
\end{eqnarray}
which proves part (b).
\end{proof}

The proposition provides an analytical description of the power 
spectrum's small and large momentum behavior. For intermediate 
momenta we evaluate $\chi_p^{\rm SLE}(\eta)$ numerically. 
For the numerical implementation a choice of 
window function in $ C_c^\infty(-\frac{1}{2H},\frac{1}{H})$ enters. A useful 
one-parametric family arises as follows. From the standard smoothened 
step function
\begin{eqnarray}
\label{pinf07}
 h(y):=
 \begin{cases}
 0 & y\leq 0\,,\\
 \frac{e^{-1/y}}{e^{-1/y}+e^{-1/(1-y)}}& 0<y<1\,,\\
 1 & y \geq 1\,,
 \end{cases}
 \end{eqnarray}
 we define the bump function of width $1+w$ centered at the origin,
 \begin{eqnarray}
\label{pinf08}
 {\tt bump}(y,w):=1-h\bigg( \frac{y^2-w^2}{(w+1)^2-w^2}\bigg)\,,
 \end{eqnarray}
where $w$ is the ratio of ``plateau'' of the bump to the ``walls'' of the bump.
Finally we define
\begin{eqnarray}
\label{pinf09}
F(\eta,\eta_1,\eta_2;w):={\tt bump}\bigg(
\frac{\eta-(\frac{\eta_1+\eta_2}{2})}{\frac{\eta_1+\eta_2}{2(w+1)}},w\bigg)\,,
\end{eqnarray} 
a positive smoothened ``top hat'' function centered at 
$\frac{\eta_1+\eta_2}{2}$. Here $\eta_1<\eta_2$ are the ``ends'' of 
the hat, specifying the cosmological period over which $F = (f^{\rm cosm})^2$ 
has support. 
The results of the 
power spectrum for various values of $\eta_1,\,\eta_2$ and $w=0.5$ 
are shown in the following figure
\begin{figure}[h]
\centering
\includegraphics[scale=0.4]{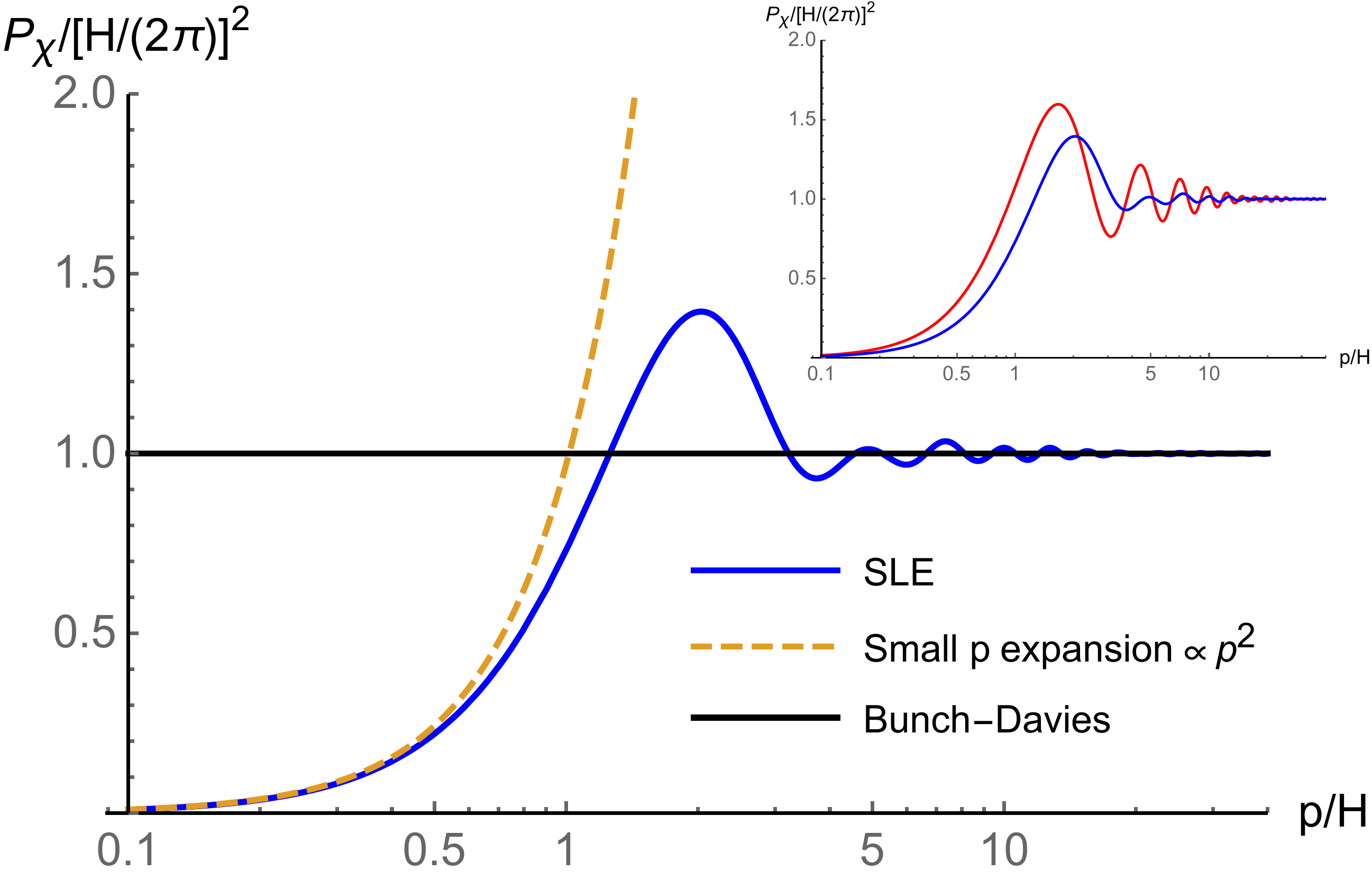}
\caption{\small Power spectrum for a primordial SLE and window function 
(\ref{pinf09}) with support in $[-0.3, 0.5]/H$. The insert shows in 
red the comparison with a situation where the window function 
has support in the pre-inflationary period $[-0.3,0]/H$ only.}
\end{figure}

\newpage
\section{Conclusions} 

The States of Low Energy (SLE) were introduced as Hadamard 
states \cite{Olbermann} on generic Friedmann-Lema\^{i}tre 
spacetimes with a physically appealing defining property.    
Here we showed that SLE have several bonus properties 
which make them mathematically and physically even more  
attractive. These bonus properties (a) -- (e) have been 
listed in the introduction and need not be repeated here. 
Instead, we comment on some extensions and future directions. 

As seen, the minimization over initial data results 
in an instructive alternative expression for the SLE 
solution solely in terms of the commutator function. 
A minimization over boundary data would likewise 
be relevant and occurs naturally when placing the basic wave 
equation into the setting of a regular Sturm-Liouville problem. 
Taking advantage of the literature on non-regular 
Sturm-Liouville problems might allow one to extend the 
SLE construction systematically to situations where the 
coefficient functions become singular within the interval 
considered. Covering the big bang singularity is of prime 
interest, but other singular points may be worthwhile 
treating as well, as the model from Section \ref{sec5} illustrates. 

The computation of the power spectrum requires 
access to the cross-over regime in the (time, momentum) 
plane. Ideally, one would be able to treat also the 
cross-over regime analytically by a suitable expansion.
Physicswise one would want to treat fully realistic 
cosmic evolutions where a pre-inflationary SLE replaces 
the positive frequency Hankel functions 
\cite{InflIRmatch1,InflIRmatch4,InflIRmatch5}
and to propagate the resulting primordial power spectrum 
to the actual CMB.

Finally, it would be desirable to have a streamlined proof 
of the Hadamard property directly for SLE and including the 
massless case. The  adiabatic vacua are a time-honored conduit 
and should be replaceable by more directly controllable WKB results
for the large momentum regime, see e.g.~\cite{Nemes}.

\noindent
{\bf Acknowledgements:} This research was supported by Pitt-PACC. R.B. also acknowledges support by the Andrew Mellon Predoctoral Fellowship from the University of Pittsburgh.



\end{document}